\documentclass{tlp}
\pdfoutput=1
\usepackage{multirow}
\usepackage{booktabs}
\usepackage{comment}
\usepackage{graphicx}
\usepackage[fleqn]{amsmath} 
\usepackage{amssymb}
\usepackage{helvet,times,courier} 
\usepackage{tikz}
\usepackage{url}
\usepackage{bold-extra}
\usepackage{bm}
\usepackage{mathtools} 
\usepackage{picinpar}

\makeatletter
\long\def\figwindownonum[#1,#2,#3,#4] {
  \begin{window}[#1,#2,{#3},{\centering#4\par}] }
\def\endfigwindownonum{\end{window}}
\makeatother

\newtheorem{definition}{Definition}
\newtheorem{example}{Example}
\newtheorem{lemma}{Lemma}
\newtheorem{theorem}{Theorem}
\newtheorem{proposition}{Proposition}
\newtheorem{corollary}{Corollary}

\newcommand{\eofex}{\mbox{}\nobreak\hfill\hspace{0.5em}$\blacksquare$}

\newcommand{\bigunion}{\bigcup}
\newcommand{\eset}[2]{\set{\rg{#1}{,}{#2}}}
\newcommand{\floor}[1]{\lfloor{#1}\rfloor}
\newcommand{\given}{\;\ifnum\currentgrouptype=16 \middle\fi|\;}
\newcommand{\intersection}{\cap}
\newcommand{\pair}[2]{\tuple{{#1},{#2}}}
\newcommand{\rg}[3]{{#1}{#2}\,\ldots{#2}\,{#3}}
\newcommand{\range}[2]{\rg{#1}{,}{#2}}
\newcommand{\rangeplus}[2]{{#1} + \cdots + {#2}}
\newcommand{\sel}[2]{\{{#1}\mid{#2}\}}
\newcommand{\sequence}[1]{[#1]}
\newcommand{\set}[1]{\{{#1}\}}

\newcommand{\tuple}[1]{\langle{#1}\rangle}
\newcommand{\union}{\cup}

\newcommand{\ttb}[1]{\texttt{\textbf{#1}}} 
\newcommand{\IF}{\mathbin{\ttb{:-}}}
\newcommand{\WEAKIF}{\mathbin{\ttb{:\raisebox{-0.7ex}{\textasciitilde{}}\,}}}
\newcommand{\AND}{\ttb{,}\,}
\newcommand{\naf}{\ttb{not}\,}
\newcommand{\END}{\ttb{.}}
\newcommand{\MIN}{\ttb{\#minimize}\,}
\newcommand{\AGG}[1]{\ttb{[}#1\ttb{]}}
\newcommand{\AGGC}[1]{\bm{\{}#1\bm{\}}}
\newcommand{\AGGR}[2]{\AGGC{{#1}\ttb{;\,...;\,}{#2}}}
\newcommand{\AGGCOUNT}[1]{\bm{\#count\,\{}#1\bm{\}}}
\newcommand{\SEP}{\ttb{,}\,}
\newcommand{\MID}{\ttb{:}\,}
\newcommand{\pred}[2]{\ttb{#1}_{#2}}
\newcommand{\sig}[1]{\mathrm{At}(#1)}
\newcommand{\GLred}[2]{{#1}^{#2}}
\newcommand{\as}[1]{\mathrm{AS}(#1)}
\newcommand{\supp}[1]{\mathrm{Supp}(#1)}
\newcommand{\iwires}[1]{\mathrm{InF}(#1)}
\newcommand{\assigntrue}[1]{\textbf{T}#1}
\newcommand{\assignfalse}[1]{\textbf{F}#1}
\newcommand{\nogoodsof}[1]{\Gamma_{#1}}
\newcommand{\binomproglt}[2]{P^{#1}_{#2}}
\newcommand{\binomprogopt}[2]{O^{#1}_{#2}}
\newcommand{\classicequiv}{\equiv_{c}}
\newcommand{\nogoodsforsupp}[1]{\Gamma_{\mathrm{supp}}(#1)}
\newcommand{\head}[1]{\mathrm{head}(#1)}

\newcommand{\weightfun}[2]{w_{#1,#2}}
\newcommand{\comparator}[3]{\tuple{#1, #2, #3}}

\newcommand{\matof}[1]{(#1)}
\newcommand{\matsym}[1]{\mathbf{#1}}
\newcommand{\operatorsym}[1]{\mathcal{#1}}
\newcommand{\propfun}[1]{\operatorsym{P}_{#1}}
\newcommand{\propfunfull}{\operatorsym{P}}
\newcommand{\aspify}[1]{\mathrm{ASP}(#1)}
\newcommand{\wirevalues}[2]{\matsym{WV}_{#1}(#2)}

\newcommand{\grayout}[1]{\textcolor{lightgray}{#1}}
\newcommand{\wzero}{\phantom{0}0}
\newcommand{\system}[1]{\textsc{#1}}
\newcommand{\pipeline}[1]{\textit{#1}}
\newcommand{\ve}[1]{\textrm{#1}}
\newcommand{\wt}[1]{\ensuremath{\mathfrak{#1}}}

\newcommand{\widenodelabel}[1]{\nodelabel{}\nodelabel{#1}\nodelabel{}}

\tikzset{snpicturestyledensedense/.style={
  xscale=0.14,
  yscale=-0.35,
  every node/.style={minimum size=0pt, inner sep=1.4pt}
}}

\tikzset{snpicturestyledense/.style={
  xscale=0.16,
  yscale=-0.4,
  every node/.style={minimum size=0pt, inner sep=1.6pt}
}}

\tikzset{snpicturestylelarge/.style={
  xscale=0.20,
  yscale=-0.5,
  every node/.style={minimum size=0pt, inner sep=2pt}
}}

\tikzset{snpicturestyle/.style={snpicturestyledense}}

\newenvironment{comparatornetwork}[2]
{
  \begin{tikzpicture}[snpicturestyle]
  \setcounter{sncolumncounter}{0}
  \foreach \i in {1, ..., #1}
  {
    \draw (0,\i)--(#2-1,\i);
  }
}
{
  \end{tikzpicture}%
}

\newcommand{\singlenodeconnection}[3]{%
    \draw[snlinestyle] (\value{sncolumncounter},#1)--(\value{sncolumncounter},#2);
    \draw (\value{sncolumncounter},#1) node[#3]{};
    \draw (\value{sncolumncounter},#2) node[#3]{};
}

% ...
\usepackage{tikz}
\usepackage[trim]{tokenizer} 

\newcounter{sncolumncounter}
\newcounter{snrowcounter}

\def \nodelabel#1{%
\setcounter{snrowcounter}{1}
 \foreach \i in {#1}{%
   \draw (\value{sncolumncounter},\value{snrowcounter}) node[anchor=south]{\i};
   \addtocounter{snrowcounter}{1}
 }
 \addtocounter{sncolumncounter}{1}
}

\tikzset{snlinestylethick/.style={
  line width=0.6pt,
  draw=white,
  double=black,
  double distance between line centers=2pt
}}

\tikzset{snlinestyle/.style={}}

\def \nodeconnection#1{%
  \foreach \i in {#1}{%
    \GetTokens{nodesrc}{nodedest}{\i}
    \draw[snlinestyle] (\value{sncolumncounter},\nodesrc) node[circle,fill=black]{}--(\value{sncolumncounter},\nodedest) node[circle,fill=black]{};
  }
  \addtocounter{sncolumncounter}{1}
}

\begin{document}

\title{Boosting Answer Set Optimization with \linebreak[4]
Weighted Comparator Networks}

\author[J. Bomanson and T. Janhunen]
{JORI BOMANSON${~}^{1}$ and TOMI JANHUNEN${~}^{1,2}$ \\
${}^{1)}$%
Department of Computer Science, Aalto University \\
P.O.\ Box 15400, FI-00076 AALTO, Finland \\
\email{Jori.Bomanson@aalto.fi, Tomi.Janhunen@aalto.fi} \\
~\\
${}^{2)}$%
Information Technology and Communication Sciences, Tampere University \\
FI-33014 TAMPERE UNIVERSITY, Finland \\
\email{Tomi.Janhunen@tuni.fi}}

\maketitle

\label{firstpage}

\noindent \textbf{Note:}
This article has been published in \emph{Theory and Practice of Logic Programming}, 20(4), 512-551, © The Author(s), 2020.

~

\begin{abstract}
Answer set programming (ASP) is a paradigm for modeling knowledge intensive
domains and solving challenging reasoning problems.
In ASP solving, a typical strategy is to preprocess problem instances by
rewriting complex rules into simpler ones.
Normalization is a rewriting process that removes extended rule types
altogether in favor of normal rules.
Recently, such techniques led to optimization rewriting in ASP, where the goal
is to boost answer set optimization by refactoring the optimization criteria of
interest.
In this paper, we present a novel, general, and effective technique for
optimization rewriting based on comparator networks, which are specific kinds
of circuits for reordering the elements of vectors.
The idea is to connect an ASP encoding of a comparator network to the literals
being optimized and to redistribute the weights of these literals over the
structure of the network.
The encoding captures information about the weight of an answer set in
auxiliary atoms in a structured way that is proven to yield exponential
improvements during branch-and-bound optimization on an infinite family of
example programs.
The used comparator network can be tuned freely, e.g., to find the best size
for a given benchmark class.
Experiments show accelerated optimization performance on several benchmark
problems.
\end{abstract}

\begin{keywords}
answer set programming,
comparator network,
normalization,
optimization rewriting,
translation
\end{keywords}

\section{Introduction}

Answer set programming (ASP) \cite{BET11:cacm,JN16:aimag} is a
declarative programming paradigm that offers rich rule-based languages
for modeling and solving challenging reasoning problems in knowledge
intensive domains. In ASP, various reasoning tasks reduce to the
computation of answer sets for a given input program and,
typically, the program is instantiated into a variable-free
\emph{ground program} in order to simplify the computation of answer sets.
Moreover, the actual search for answer sets may generally rely on
preprocessing steps where more complex ground rules are rewritten in
terms of simpler ones either to gain better performance or to accommodate
back-end solvers with limited language support.
Such preprocessing includes the process of \emph{normalization}, which produces
only normal rules \cite{BJ13:lpnmr,BGJ14:jelia,bomanson17}.

More recently, similar rewriting techniques were developed for
refactoring \emph{optimization statements} \cite{BGJ16:iclp},
giving rise to the concept of \emph{optimization rewriting} in ASP
with the goal of boosting search performance in answer set optimization.
Several of the explored designs for normalization and rewriting
rely on rule-based encodings of gadgets such as
\emph{binary decision diagrams} (BDDs) or
\emph{sorting networks} \cite{Batcher68:afips}.
Sorting networks form extensively studied classes of circuits with
applications to sorting on parallel computers and encoding cardinality
constraints or other pseudo-Boolean constraints in logical formalisms such as
Boolean satisfiability (SAT) and ASP.
They are used to sort vectors of elements by performing predetermined series of
\emph{compare-exchange} operations on pairs of input elements by elementary
circuits known as \emph{comparators}.
More generally, networks with such structure are known as
\emph{comparator networks} whether they guarantee to produce sorted output or
not.
It is convenient to represent comparator networks as Knuth diagrams,
as illustrated in Figure~\ref{fig:example-comparator-network}. The
input is fed to the left end of the circuit and it proceeds through
the network along the vertical lines representing the \emph{wires} of the
network. Individual
comparators are marked with bullets connected by lines and eventually
they produce the output at the right end of the circuit.

In this paper, we concentrate on rewriting optimization statements used in ASP
into modified optimization statements involving auxiliary atoms defined in terms of
newly introduced normal rules. The
motivation behind the introduction of new atoms is to offer modern
answer set solvers additional branching points as well as further
concepts to learn about.
There are theoretical proof complexity results--given in the context of ASP by
Lifschitz and Razborov \citeyear{LR06:acmtocl}, Anger et
al.~\citeyear{AGJS06:ecai}, and Gebser and Schaub
\citeyear{asp_tableau_calculi_journal}---illustrating the promise behind new
auxiliary atoms, potentially leading to exponentially smaller search spaces.

SAT encodings of sorting networks, among others, are specifically known to cut
down otherwise exponential numbers of clauses generated by a SAT Modulo
Theories solver (SMT) when used to express particularly troublesome sets of
cardinality constraints \cite{abniolrost13}.
The novel rewriting
scheme presented in this paper exploits comparator networks as the
underlying design, but in contrast to previous work by Bomanson
et al.\ \citeyear{BGJ16:iclp}
treats the weights of an optimization statement in a different way.
To formulate the essential idea of this paper independently of ASP, we
generalize comparators to accept weighted input signals and introduce
the resulting notion of
\emph{weighted comparator networks}.
We exploit these networks in deriving meaningful identities on linear
combinations of weights
and signals. The main technique redistributes weights
associated with the input signals of a comparator network over the
structure of the network. The net effect of the redistribution is that
weights get smaller and their number increases while the sum of
weights stays invariant. 
These weights are distributed such that, in general, they become increasingly
uniform towards the end of the network.
At the very end, given a sufficiently deep and well connected network, the
outputs of the comparator network will be weighted by the minimum $m$ of all
initial input weights.
In this way, the
sorted output atoms form a kind of a sliding switch using which the solver
may make assumptions on the weight of answer sets being sought for.
In
recursive designs, such as sorting networks based on mergers
\cite{Batcher68:afips}, the same line of reasoning can be applied
recursively at particular inner levels of the network.

The idea discussed above gives rise to a rewriting scheme for optimization
statements in ASP.
We analyze the scheme formally and prove that it enables an exponential
improvement in branch-and-bound solving performance on an example family of ASP
optimization programs.
Moreover, the optimization rewriting scheme is realized in a new tool called
\system{pbtranslate}.
We present an experimental study on the performance effect of the tool when used
as a preprocessor for the state-of-the-art answer set solver \system{clasp}.
Our results identify a number of benchmark problems where
the search for optimal answer sets is accelerated.

The rest of this article is organized as follows.
In Section~\ref{section:preliminaries}, we give the basic definitions
and notations related with comparator and sorting networks. Furthermore,
we review the basic notions of answer set programming to the extent
needed in this paper.
The process of weight propagation over comparator networks is
explained in Section~\ref{section:propagation} and shown to preserve
the correct interpretation of pseudo-Boolean expressions in general.
Section~\ref{section:application} concentrates on applying weight
propagation to rewriting ASP optimization statements.
In this context, a theorem is presented on the correctness of such rewritings,
when the underlying comparator networks are encoded with rules and weights are
propagated over the network according to a general scheme.
A formal analysis of the performance improving potential behind the rewritings
is also presented for an example family of answer set optimization programs.
This analysis is experimentally verified to be relevant to actual answer set
solvers on the family of programs.
Moreover, the rewritings are also evaluated in extensive experiments over a
range of relevant benchmarks from, e.g., answer set programming competitions.
An account of related work is provided Section \ref{section:related}.
Finally, the paper is concluded in Section~\ref{section:conclusion}
with a summary and a sketch of future work.

\begin{figure}
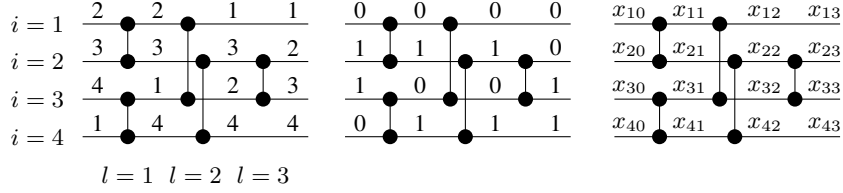

  \centering
  \tikzset{snpicturestyle/.style={snpicturestylelarge}}
  \begin{comparatornetwork}{4}{16}
    \draw (-3,1) node {$i = 1$};
    \draw (-3,2) node {$i = 2$};
    \draw (-3,3) node {$i = 3$};
    \draw (-3,4) node {$i = 4$};
    \widenodelabel{\ve{2},\ve{3},\ve{4},\ve{1}}
    \draw (\value{sncolumncounter},5) node {$l = 1$};
    \nodeconnection{{1,2},{3,4}}
    \widenodelabel{\ve{2},\ve{3},\ve{1},\ve{4}}
    \draw (\value{sncolumncounter}+0.5,5) node {$l = 2$};
    \nodeconnection{{1,3}}
    \nodeconnection{{2,4}}
    \widenodelabel{\ve{1},\ve{3},\ve{2},\ve{4}}
    \draw (\value{sncolumncounter},5) node {$l = 3$};
    \nodeconnection{{2,3}}
    \widenodelabel{\ve{1},\ve{2},\ve{3},\ve{4}}
  \end{comparatornetwork}
  ~~~
  \begin{comparatornetwork}{4}{16}
    \widenodelabel{\ve{0},\ve{1},\ve{1},\ve{0}}
    \draw (\value{sncolumncounter},5) node {\phantom{$l = 1$}};
    \nodeconnection{{1,2},{3,4}}
    \widenodelabel{\ve{0},\ve{1},\ve{0},\ve{1}}
    \draw (\value{sncolumncounter}+0.5,5) node {\phantom{$l = 2$}};
    \nodeconnection{{1,3}}
    \nodeconnection{{2,4}}
    \widenodelabel{\ve{0},\ve{1},\ve{0},\ve{1}}
    \draw (\value{sncolumncounter},5) node {\phantom{$l = 3$}};
    \nodeconnection{{2,3}}
    \widenodelabel{\ve{0},\ve{0},\ve{1},\ve{1}}
  \end{comparatornetwork}
  ~~~
  \begin{comparatornetwork}{4}{16}
    \widenodelabel{$x_{10}$,$x_{20}$,$x_{30}$,$x_{40}$}
    \draw (\value{sncolumncounter},5) node {\phantom{$l = 1$}};
    \nodeconnection{{1,2},{3,4}}
    \widenodelabel{$x_{11}$,$x_{21}$,$x_{31}$,$x_{41}$}
    \draw (\value{sncolumncounter}+0.5,5) node {\phantom{$l = 2$}};
    \nodeconnection{{1,3}}
    \nodeconnection{{2,4}}
    \widenodelabel{$x_{12}$,$x_{22}$,$x_{32}$,$x_{42}$}
    \draw (\value{sncolumncounter},5) node {\phantom{$l = 3$}};
    \nodeconnection{{2,3}}
    \widenodelabel{$x_{13}$,$x_{23}$,$x_{33}$,$x_{43}$}
  \end{comparatornetwork}
  \caption{\label{fig:example-comparator-network}%
    An example sorting network on four wires $i$ having three layers of
    comparators at levels $l$.
    On the left, the network is shown operating on the input numbers
    $\sequence{\ve{2},\ve{3},\ve{4},\ve{1}}$
    and in the middle, on the binary sequence
    $\sequence{\ve{0},\ve{1},\ve{1},\ve{0}}$.
    Each of the comparators takes the two values from its left and places them
    on its right in ascending order.
    For example, the comparator in the upper left corner keeps $[2, 3]$ as
    $[2, 3]$ whereas the one below it turns $[4, 1]$ to $[1, 4]$.
    By the properties of sorting networks, these and any other input sequences
    become sorted in the end.
    The rightmost diagram illustrates the indexing convention used in this paper
    for values associated with wires.
  }
\end{figure}

\section{Preliminaries}
\label{section:preliminaries}

In this section, we review the basic definitions of ASP as well as
comparator networks which also cover sorting networks
as their special case. To reach the goals of this paper, it is also
essential to translate comparator networks into ASP and to establish that
the resulting negation-free logic programs faithfully capture
the compare-exchange operations performed by networks.

\subsection{Answer Set Programs and Nogoods}

Below, we present (ground) answer set \emph{programs} as sets consisting of
\emph{normal rules},
which are typical primitives for modeling search problems~\cite{JN16:aimag},
and \emph{nogoods},
which are typical primitives for modeling search procedures~\cite{GKS12:aij}.
To this end, we first define concepts related to the latter.
In particular, an \emph{assignment} $A$ is a set of \emph{signed literals}
$\sigma$ of the form $\assigntrue{x}$ or $\assignfalse{x}$, each of which
expresses the assignment of an atom $x$ to true or false, respectively.
Intuitively, an assignment is a three-valued interpretation that may leave
any atoms as undefined.
A nogood $\delta$ is syntactically identical to an assignment, but a
nogood carries the meaning that all partial assignments $A \supseteq \delta$ are
forbidden.
A \emph{constraint} $\Gamma$ is a set of nogoods $\delta$.
An assignment $A$ is in \emph{conflict} with a nogood $\delta$ if $\delta
\subseteq A$, and with a set of nogoods $\Gamma$ if it is in conflict with any
$\sigma' \in \Gamma$.
Formally, an answer set program
$P$
is a set of normal rules of the form (\ref{eq:normal-rule}), shown below,
and nogoods $\delta$.
Each program is assumed to be associated with a predefined \emph{signature}
$\sig{P}$ that is a superset of the atoms occurring in the program.
Intuitively, the \emph{head} atom $a$ of a normal rule is to be derived, if the other
rules in $P$ can be used to derive all atoms $\rg{b_1}{,}{b_k}$ in the
\emph{positive body} and no atoms $\rg{c_1}{,}{c_m}$ in the
\emph{negative body} of the rule.
The set of all head atoms $a$ of rules in $P$ is denoted by $\head{P}$.
A default \emph{literal} $l$ is either an atom $a$ or its negation $\naf a$,
expressing success or failure to prove $a$, respectively.
An \emph{optimization program} $O$ is a pair $\pair{P}{e}$ where $P$
is an answer set program and $e$ is an \emph{objective function} in the form of
a pseudo-Boolean expression
$\rangeplus{w_1l_1}{w_nl_n}$ with weights $\range{w_1}{w_n}$ and literals
$\range{l_1}{l_n}$. The \emph{objective function} $e$ can be written
as a set of \emph{weak constraints} of the form (\ref{eq:weak-constraint})
in the \emph{ASP-Core-2} input language \cite{aspcore201c} or as a single
\emph{optimization statement} (\ref{eq:optimization-statement}).
For convenience, we consider certain further \emph{extensions} to answer
set programs: namely \emph{choice rules} of the form (\ref{eq:choice-rule}) and
\emph{cardinality constraints} of the form (\ref{eq:cardinality-constraint}).
Intuitively, a choice rule differs from a normal rule in that it justifies the
derivation of any subset of its head atoms $\range{a_1}{a_m}$ if its body
conditions are satisfied, and that subset is allowed to be empty.
A cardinality constraint forbids the pseudo-Boolean expression
$\rangeplus{1l_1}{1l_n}$ from taking a value less than $k$.
That is, it requires at least $k$ of the literals $\range{l_1}{l_n}$ to be true.
\begin{align}
&
\label{eq:normal-rule}
a \IF\rg{b_1}{\AND}{b_k}\AND\rg{\naf c_1}{\AND}{\naf c_m}\END
\\ 
\label{eq:weak-constraint}
&
\begin{array}{@{}ccc@{}}
\WEAKIF l_1 \END ~ \AGG{w_1\SEP 1}
&
\dots
&
\WEAKIF l_n \END ~ \AGG{w_n\SEP n}
\end{array}
\\ 
&
\label{eq:optimization-statement}
\MIN\AGGR{w_1\SEP 1\MID l_1}{w_n\SEP n\MID l_n} \END
\\ 
&
\label{eq:choice-rule}
\set{\range{a_1}{a_m}} \IF\rg{l_1}{\AND}{l_n}\END
\\ 
&
\label{eq:cardinality-constraint}
\IF \AGGCOUNT{\rg{l_1}{;}{l_n}} < k
\end{align}

An interpretation $I\subseteq\sig{P}$ of a program $P$ is an assignment that is
\emph{complete} in that it leaves no atom $a\in\sig{P}$ undefined, and which is
here represented as the set of atoms assigned true.
An interpretation $I\subseteq\sig{P}$ \emph{satisfies}
a nogood $\delta$ if there is any $\assigntrue{a} \in \delta$ such that
$a \not\in I$ or any $\assignfalse{a} \in \delta$ such that $a \in I$;
it satisfies
a rule (\ref{eq:normal-rule})
if it satisfies the nogood
$\{
   \assignfalse{a}$,
  $\assigntrue{b_1}$, $\dots$, $\assigntrue{b_k}$,
  $\assignfalse{c_1}$, $\dots$, $\assignfalse{c_m}
\}$;
and it satisfies the answer set program $P$
if it satisfies all nogoods and rules in $P$.
The \emph{reduct} $\GLred{P}{I}$ of~$P$ with respect to~$I$ contains the
rule $a\IF\rg{b_1}{\AND}{b_k}$ for each rule
(\ref{eq:normal-rule}) in $P$ with
$\eset{c_1}{c_m}\intersection I=\emptyset$.
The set $\as{P}$ of \emph{answer sets} of a program $P$ is the set of
interpretations $I\subseteq\sig{P}$ that satisfy $P$ and are
$\subseteq$-minimal among the interpretations $J$ that satisfy all rules in
$P^I$ and the condition $J \setminus \head{P} = I \setminus \head{P}$.
This last condition is an extension to the usual definition of answer
sets~\cite{BET11:cacm} that supports the convenient use of monotone constructs
in the form of nogoods and non-monotone constructs in the form of normal rules
in a single answer set program.
In particular, for programs with nogoods, the defined answer sets coincide with
the standard ones and for programs with only nogoods, they coincide with
the classical models of the program.
Regarding optimization, given a pseudo-Boolean expression
$e=\rangeplus{w_1l_1}{w_nl_n}$, the \emph{value} $e(I)$ of $e$ in an
interpretation $I\subseteq\sig{P}$ is the sum of weights $w_i$ for literals
$l_i$ satisfied by $I$.
An answer set $I$ of a program $P$
is optimal for the optimization program $\pair{P}{e}$ iff $e(I)$ equals
$\min\sel{e(J)}{J\in\as{P}}$.  In general, an answer set program $P$ has a set of
$e$-optimal answer sets, which can be enumerated by modern ASP solvers
such as \system{clasp}~\cite{gekakarosc15a}.

Regarding the semantics of choice rules and cardinality constraints, we treat
both as syntactic shortcuts.
To this end, a choice rule~(\ref{eq:choice-rule}) stands for the
set of normal rules
\begin{align*}
\set{d \IF \rg{l_1}{\AND}{l_n} \END}
\union
\sel{
a_i \IF \naf a'_i, d \END
\hspace{4pt}
a'_i \IF \naf a_i \END
}{
1 \le i \le m
},
\end{align*}
where atoms $d$ and $a'_i$ for $1 \le i \le m$ are new auxiliary atoms not
appearing elsewhere in the program.
On the other hand, a cardinality constraint~(\ref{eq:cardinality-constraint})
stands for the set of nogoods
\begin{align*}
  \sel{
    \sel{\assignfalse{a}}{\text{atom $a \in X$}}
    \union
    \sel{\assigntrue{a}}{\naf a \in X}
  }{
    X \subseteq \eset{l_1}{l_n},
    |X| = n - k + 1
  },
\end{align*}
which individually forbids each $(n - k + 1)$ subset of the literals
$\eset{l_1}{l_n}$ from being falsified.
This ensures that at least $k$ of the literals may be satisfied.

In addition to the answer sets of a program, we consider a superset
of them, namely the set of \emph{supported models} of the program
\cite{apblwa87a}.
These are important in answer set solving due to this superset relation: answer
sets can be characterized as supported models that satisfy additional
conditions.
This is the approach behind, e.g., the ASP solver
\system{clasp}~\cite{GKS12:aij}.
Formally,
the set $\supp{P}$ of supported models of a program $P$ is the set of
interpretations $I \subseteq \sig{P}$
that satisfy the program $P$ and the condition that for every atom $a \in I$,
there is some rule
(\ref{eq:normal-rule}) in $P$ with $a$ as the head and with
$\eset{b_1}{b_k} \subseteq I$
and
$\eset{c_1}{c_m}\intersection I=\emptyset$.
In order to capture this semantics in the form of a constraint,
we define
the \emph{supported model constraint}
$\nogoodsforsupp{P}$
of an answer set program $P$ to be the set of nogoods
\begin{align*}
\{
  \sel{\assigntrue{a}}{a \in I}
  \union
  \sel{\assignfalse{a}}{a \in \sig{P} \setminus I}
\given
  \text{$I \subseteq \sig{P}, I \not\in \supp{P}$}
\}
\end{align*}
that is satisfied exactly by the supported models of $P$.
The constraint $\nogoodsforsupp{P}$ defined here is naive and generally huge.
However, it is used only for theoretical considerations in this paper,
and it is therefore sufficient.
In actual implementations \cite{GKS12:aij}, it is better to approach supported
models via \emph{Clark's completion} \cite{clark78a}.

The role of nogoods in the definition of an answer set in this paper is to
reject unwanted answer sets.
In accordance with this, the addition of nogoods to an answer set program has a
\emph{monotone} impact on the answer sets of the program when the nogoods
involve no new atoms.
In particular, the set $\as{P \union \Gamma}$ of answer sets of
the union of an answer set program $P$ and a
constraint $\Gamma$ on the atoms in $\sig{P}$
can be obtained as
$
\as{P \union \Gamma}
=
\sel{M \in \as{P}}{\text{$M$ satisfies $\Gamma$}}
$.
This is in contrast with the generally \emph{non-monotone} behavior of answer
set semantics, due to which the addition of a rule, such as a fact, may
not only decrease, but also increase the number of answer sets.

We define two programs $P$ and $Q$ over the same signature $D$ to be
\emph{classically equivalent}, denoted by $P \classicequiv Q$,
if each interpretation $I \subseteq D$ satisfies either both $P$ and $Q$ or
neither $P$ nor $Q$.
Observe that this equivalence concept is preserved under the addition of
nogoods in the following way.
Given any constraint $\Lambda$, if two programs $P$ and $Q$ are classically
equivalent, then so are $P \union \Lambda$ and $Q \union \Lambda$.
In other words, $\classicequiv$ is a \emph{congruence} relation for $\union$.

\subsection{Comparator/Sorting Networks}
\label{section:comparator-networks}

Intuitively, a \emph{comparator} checks whether a predetermined pair of
input elements is ordered and, if not, changes their order. Formally,
we define a comparator to be a tuple $\comparator{i}{j}{l}$ consisting of
\emph{wires} $1 \le i < j$ and a \emph{level} $l \ge 1$.
In this notation, the comparators of the network in
Figure~\ref{fig:example-comparator-network} are
$\comparator{1}{2}{1},
\comparator{3}{4}{1},
\comparator{1}{3}{2},
\comparator{2}{4}{2},
\comparator{2}{3}{3}$.
We consider two
comparators \emph{compatible} if their sets $\set{i, j}$ of wires
are disjoint or their levels $l$ distinct.
A \emph{(comparator) network} $N$ is a set of mutually compatible
comparators. The independence of the comparators from the input beyond
the input size makes comparator networks \emph{data oblivious}.
This facilitates their implementation in hardware and representation in logical
formalisms.
We say that a network $N$ is \emph{confined} to a set $I$ of wires and
an interval $E$ of levels if every comparator
$\comparator{i}{j}{l} \in N$ satisfies $\set{i,j}\subseteq I$ and $l\in
E$. Unless stated otherwise, we assume that each network
$N$ is confined to both $I=\eset{1}{n}$ and $E=\eset{1}{d}$ where $n$
and $d$ give the width and the depth parameters of the network $N$,
respectively.
A \emph{layer} $L$ is a network of comparators confined to a single
level.  Accordingly, the wires of comparators of a layer must be
distinct.  The \emph{layer $L$ of a network $N$ at level $l$} is $L =
\set{\comparator{i}{j}{l}\in N}$.

Given an \emph{input vector} $\vec{x}$ consisting of $n$ comparable values, a
layer $L$ of comparators permutes them by swapping every wrongly-ordered pair
$x_i > x_j$ occurring on the wires $i, j$ of any single comparator
$\comparator{i}{j}{l}\in L$, while leaving all other values intact.
Furthermore, the output of a network $N$ of depth $d$ is
$f_d(\cdots(f_1(\vec{x}))\cdots)$ where each function $f_l$ gives the output of
the layer at level $l$.
Consequently, as the output of a single layer is always some permutation of its
inputs, so is the output of the entire network.
Put otherwise, the output is identical to the input when regarded as a multiset
of values.
Given an input vector $\vec{x}$, we define a network $N$ of depth $d$ to yield a
two-dimensional array of \emph{wire values}
$\wirevalues{N}{\vec{x}} = \matof{x_{il}}$
indexed by wire $i \in \eset{1}{n}$ and level $l \in \eset{0}{d}$, such that
the column of values at level $l$ is the output of the network of layers up to
$l$, i.e., the network $\set{\comparator{i}{j}{l'}\in N \given l' \le l}$.
We illustrate wire values superimposed over networks as in
Figure~\ref{fig:example-comparator-network}.

A \emph{sorting network} $N$ is a comparator network that sorts every input
$\vec{x}$ into a respective output $\vec{y}$ such that $\rg{y_1}{\le}{y_n}$.
A \emph{confined network} $C$ is a tuple $C = \tuple{N, I, E}$ where
$N$ is a network confined to the sets of wires $I$ and levels $E$.
Confined networks
$\tuple{N_1, I_1, E_1}$ and $\tuple{N_2, I_2, E_2}$
are \emph{compatible} if $I_1 \intersection I_2 = \emptyset$ or $E_1
\intersection E_2 = \emptyset$.
A \emph{decomposition} $D=\eset{S_1}{S_m}$ of a network $N$ is a set
of mutually compatible confined networks $S_i = \tuple{N_i, I_i, E_i}$
such that $\bigunion_{i = 1}^m N_i = N$.

\subsection{Capturing Comparator Networks with ASP}
\label{section:translation}

A comparator network $N$ for Boolean input vectors $\vec{x}$ can be translated
into ASP as follows.
We introduce an atom $\pred{x}{il}$
to capture the wire value $x_{il}$ for each wire $i$ and level $l$ so that
$\pred{x}{il}$ is to be true iff $x_{il}=1$ in the
matrix $\wirevalues{N}{\vec{x}}=\matof{x_{il}}$.
The effect of a comparator $\comparator{i}{j}{l}\in N$ can be captured
in terms of the following rules \cite{BGJ16:iclp} for $0<l\le n$:
\begin{center}
$\begin{array}{l@{\hspace{2em}}l@{\hspace{2em}}l}
    \pred{x}{il}\IF\pred{x}{i(l-1)}\AND\pred{x}{j(l-1)}\END &
    \pred{x}{jl}\IF\pred{x}{i(l-1)}\END &
    \pred{x}{jl}\IF\pred{x}{j(l-1)}\END
\end{array}$
\end{center}
In addition, if a wire $i$ at level $l$ is not incident with
any comparator, a rule of \emph{inertia} is introduced:
\begin{center}
$\pred{x}{il}\IF\pred{x}{i(l-1)}\END$
\end{center}
We write $\aspify{N}$ for the ASP translation of $N$ in this way
and state the following result:

\begin{lemma}
\label{lemma:correspondence}
Let $N$ be a comparator network of width $n$ and depth $d$ and
$\aspify{N}$ its translation into a negation-free answer set program.
Also, let $\vec{x}$ be any Boolean input vector for $N$,
$\wirevalues{N}{\vec{x}}=\matof{x_{il}}$ the resulting matrix of wire values,
and $\iwires{\vec{x}}=\sel{\pred{x}{i0}}{x_{i0}=1}$ an
encoding of the input vector $\vec{x}$ as facts.
Then $\aspify{N}\union\iwires{\vec{x}}$ has a unique answer set
$X\subseteq\sig{\aspify{N}}$ such that for all wires $1\leq i\leq n$
and for all levels $0\leq l\leq d$, the atom $\pred{x}{il}\in X$ iff
the wire value $x_{il}=1$.
\end{lemma}

\begin{proof}
For the base case $l=0$, we note that $\pred{x}{i0}\in X$ $\iff$
$\pred{x}{i0}\in\iwires{\vec{x}}$ $\iff$ $x_{i0}=1$, since $\pred{x}{i0}$ is
defined by a fact or no rule in $\iwires{\vec{x}}$ and by no rule in
$\aspify{N}$, and $X$ is $\subseteq$-minimal.
Induction step $0<l\le d$ follows.

If $i$ and $j$ are wires incident with a comparator
$\comparator{i}{j}{l}$ at level $l$, then the rules
of $\aspify{N}$ and the $\subseteq$-minimality of $X$
guarantee that (i)
$\pred{x}{il}\in X$ $\iff$
$\pred{x}{i(l-1)}\in X$ and $\pred{x}{j(l-1)}\in X$ $\iff$
$x_{i(l-1)}=1$ and $x_{j(l-1)}=1$ (by inductive hypothesis) $\iff$
$x_{il}=\min\set{x_{i(l-1)},x_{j(l-1)}}$ $\iff$
$x_{il}=1$;
and (ii)
$\pred{x}{jl}\in X$ $\iff$
$\pred{x}{i(l-1)}\in X$ or $\pred{x}{j(l-1)}\in X$ $\iff$
$x_{i(l-1)}=1$ or $x_{j(l-1)}=1$ (by inductive hypothesis) $\iff$
$x_{jl}=\max\set{x_{i(l-1)},x_{j(l-1)}}$ $\iff$
$x_{jl}=1$.

If a wire $i$ is not incident with comparators, it follows by the
inertia rule and $\subseteq$-minimality that
$\pred{x}{il}\in X$ $\iff$ $\pred{x}{i(l-1)}\in X$ $\iff$
$x_{i(l-1)}=1$ (by inductive hypothesis) $\iff$ $x_{il}=1$.
\end{proof}

\section{Propagating Weights Over Comparator Networks}
\label{section:propagation}

In this section, we consider contexts where comparator networks are supplemented
by \emph{weight} information. Namely, we wish to model comparator networks with fixed
multipliers, or weights, applied to their input wires.
Such input can be extracted from, e.g., pseudo-Boolean constraints or
optimization statements that are the targets of optimization rewriting
techniques, to be discussed in Section~\ref{section:application}.
Our goal is to explore the performance implications of moving these weights
along the network using propagation operations that we define in this section.

We begin by introducing the concept of \emph{wire weights}
for a network on $n$ wires and $d$ layers. They are non-negative numbers
$a_{ij}$ indexed by wires $1 \le i \le n$ and levels $0 \le j \le d$
in the same way as wire values.
A network with wire weights relates to a linear function as follows.

\begin{definition}
For a comparator network $N$
with wire weights $\matsym{A} = \matof{a_{ij}}$,
the \emph{weight function} $\weightfun{N}{\matsym{A}}$
is defined on input $\vec{x}$ yielding the wire values
$\wirevalues{N}{\vec{x}}=\matof{x_{ij}}$ by
\begin{align}
\label{eq:weight-function}
\weightfun{N}{\matsym{A}}(\vec{x})
= \sum_{i = 1}^n \sum_{j = 0}^d a_{ij} x_{ij} .
\end{align}
\end{definition}
\tikzset{snpicturestyle/.style={snpicturestylelarge}}

\newcommand{\exampletriplenetwork}[3]{%
  \begin{comparatornetwork}{3}{11}
    \widenodelabel{#1}
    \nodeconnection{{1,2}}
    \widenodelabel{#2}
    \nodeconnection{{2,3}}
    \widenodelabel{#3}
  \end{comparatornetwork}
}

\newcommand{\exampletriplenothing}[3]{%
  \begin{emptynetwork}{3}{11}
    \widenodelabel{#1}
    \nodeconnection{}
    \widenodelabel{#2}
    \nodeconnection{}
    \widenodelabel{#3}
  \end{emptynetwork}
}

\newenvironment{exampletripleblock}
{
  \begin{minipage}{0.16\textwidth}
  \centering
}
{\end{minipage}
}

\newenvironment{exampleblock}
{
  \begin{minipage}{0.10\textwidth}
  \centering
}
{\end{minipage}
}

\newcommand{\wpair}[2]{\ensuremath{#1 \cdot #2}}

\makeatletter
\newlength{\mylength}
\setlength{\mylength}{0pt}
\newcommand{\mathleft}{\@fleqntrue\@mathmargin\mylength}
\newcommand{\mathcenter}{\@fleqnfalse}
\makeatother

\begin{example}
  \label{example:wire-weights}
  In the following, we have a network with wire weights, wire values based on
  an input vector $\vec{x} = \sequence{\ve{1}, \ve{2}, \ve{0}}$, and a
  calculation that yields the respective weight function value $130$.
  This is an example on how a network combined with wire weights relates to a
  linear function on input vectors such as $\vec{x}$.
  Here and in the sequel, we emphasize wire weights with a distinct font.
  \tikzset{snpicturestyle/.style={snpicturestylelarge}}
  \begin{center}
    \begin{exampletripleblock}\exampletriplenetwork{\wt{\wzero},\wt{10},\wt{\wzero}}{\wt{\wzero},\wt{\wzero},\wt{20}}{\wt{30},\wt{40},\wt{40}}\end{exampletripleblock}
    \hspace{30pt}
    \begin{exampletripleblock}\exampletriplenetwork{\ve{1},\ve{2},\ve{0}}{\ve{1},\ve{2},\ve{0}}{\ve{1},\ve{0},\ve{2}}\end{exampletripleblock}
    \hspace{30pt}
    \begin{minipage}{0.30\textwidth}
      \mathleft
      \begin{align*}
        \mathrel{\phantom{+}}
          \wpair{\wt{\wzero}}{\ve{1}} + \wpair{\wt{\wzero}}{\ve{1}} + \wpair{\wt{30}}{\ve{1}} \\[-2pt]
        + \wpair{\wt{10}}{\ve{2}}     + \wpair{\wt{\wzero}}{\ve{2}} + \wpair{\wt{40}}{\ve{0}} \\[-2pt]
        + \wpair{\wt{\wzero}}{\ve{0}} + \wpair{\wt{20}}{\ve{0}}     + \wpair{\wt{40}}{\ve{2}} \mathrlap{~ = 130}                                    \\[-7pt]
      \end{align*}
      \mathcenter
    \end{minipage}
  \end{center}

  \eofex
\end{example}

As can be seen, the nonzero wire weights in Example~\ref{example:wire-weights}
are already scattered around the comparator network, occupying all the layers.
This state represents the goal that we want to achieve from a starting point,
where only the leftmost input weights of comparators are nonzero.
Indeed, given a comparator network $N$ and wire weights $\matsym{A}$,
we are interested in modifying the weights by propagating as much of them
as deep inside the network as possible.
To this end, we develop a \emph{propagation function} $\propfunfull$
that produces new weights $\matsym{U} = \propfunfull(\matsym{A})$
so that the respective weight function stays the same,
i.e.,
$
\weightfun{N}{\matsym{A}}(\vec{x})
=
\weightfun{N}{\matsym{U}}(\vec{x})
$
for all input vectors $\vec{x}$.
To obtain an idea of how this can be achieved in practice, let us study a simple
example.

\newcommand{\examplesingletonnetwork}[2]{%
  \begin{comparatornetwork}{2}{7}
    \widenodelabel{#1}
    \nodeconnection{{1,2}}
    \widenodelabel{#2}
  \end{comparatornetwork}
}

\begin{example}
\label{example:singleton}
\tikzset{snpicturestyle/.style={snpicturestylelarge}}
\begin{figwindownonum}[0,r,{\mbox{%
        \begin{exampleblock}\examplesingletonnetwork{\wt{40},\wt{50}}{\wt{\wzero},\wt{\wzero}}
        \end{exampleblock}
        \begin{exampleblock}\examplesingletonnetwork{\wt{\wzero},\wt{10}}{\wt{40},\wt{40}}
        \end{exampleblock}}},{Propagation over a single comparator}]
Consider a network $N = \set{\comparator{1}{2}{1}}$ with a single comparator.
Example initial and propagated weights $\matsym{A}$ and $\matsym{U}$ for $N$,
respectively, are shown on the right.
The difference between these weights is that at first, all weight
is on the input, whereas afterwards, a weight amount of $40$ has
been propagated from the input to the output on both wires.  By
comparing the weights, we may observe that the weights~$\matsym{A}$
yield the weight function
$\weightfun{N}{\matsym{A}}(\vec{x}) = 40 x_1 + 50 x_2$,
whereas the weights~$\matsym{U}$ yield the weight function
$\weightfun{N}{\matsym{U}}(\vec{x})
 = 10 x_2 + 40 (\min \set{x_1, x_2} + \max \set{x_1, x_2})
 = 10 x_2 + 40 (x_1 + x_2)$.
Namely,
$\weightfun{N}{\matsym{A}}(\vec{x})=\weightfun{N}{\matsym{U}}(\vec{x})$
for every input $\vec{x}$. Therefore, the change of weights from
$\matsym{A}$ to $\matsym{U}$ preserves the semantics of the
network as a linear function.
\eofex
\end{figwindownonum}
\end{example}

The idea behind the preceding example generalizes to larger networks.
The result is a weight propagation function that can be applied to wire weights
of comparator networks without altering the values of any weight functions
associated with them.
Namely, considering an arbitrary comparator network as a \emph{black-box}, one
may move a constant amount of weight from each of its inputs to each of its
outputs, while keeping all other weights inside the network intact.
The propagated weight will contribute the same total weight to the value of the
weight function (\ref{eq:weight-function}) before and after the move on any
input $\vec{x}$.
Therefore, the move preserves the semantics of the network as a linear function,
in the same way as the propagation step in Example~\ref{example:singleton} does.
To see this, one may consider that if the input to a single comparator is known,
then the function of the comparator can be represented as a permutation.
Either the permutation swaps the input pair, or keeps it as it is.
Furthermore, this inductively holds for any comparator network:
given any input, the output is a permutation of it,
although the permutation is generally more complex.
Consequently, the cardinality of Boolean input and output pairs are always
equal.
This preservation of cardinality guarantees the preservation of weight
functions under this propagation step.
As for the choice of the constant amount of weight that is moved, one can pick
the minimum of the input weights.
This will maximize the moved weight without producing any negative weights.
In the following, the resulting weight propagation function is defined formally,
its weight function preservation property is captured in a theorem, and a proof
for the theorem is provided following the strategy sketched above.

\begin{definition}
\label{def:propfunfull}
Given wire weights $\matsym{A} = \matof{a_{ij}}$ for a comparator network $N$
on $n$ wires and $d$ layers, and $c = \min \sel{a_{i0}}{1 \le i \le n}$,
the \emph{weight propagation function} $\propfunfull$ maps $\matsym{A}$ to
the wire weights $\propfunfull(\matsym{A}) = \matsym{U} = \matof{u_{ij}}$
of $N$ where
\begin{align*}
  u_{ij} =
  \begin{cases}
    a_{ij} - c , & \text{if $j = 0$,} \\
    a_{ij}     , & \text{if $0 < j < d$,} \\
    a_{ij} + c , & \text{if $j = d$}.
  \end{cases}
\end{align*}
\end{definition}

\begin{theorem}
\label{theorem:propfunfull}
Given wire weights $\matsym{A}$ and $\matsym{U} = \propfunfull(\matsym{A})$
for a comparator network $N$, for any input vector $\vec{x}$, it holds that
$\weightfun{N}{\matsym{A}}(\vec{x}) = \weightfun{N}{\matsym{U}}(\vec{x})$.
\end{theorem}

\begin{proof}
Let $\matsym{A} = \matof{a_{ij}}$, $c=\min\sel{a_{i0}}{1 \le i \le n}$, and
$\matsym{U} = \matof{u_{ij}} = \propfunfull(\matsym{A})$.
For any input vector $\vec{x}$,
\begin{align}
\label{eq:network-weight-difference}
\begin{split}
\weightfun{N}{\matsym{A}}(\vec{x})-\weightfun{N}{\matsym{U}}(\vec{x})
  &
  =
    \sum_{i = 1}^n \sum_{j = 0}^d (a_{ij} - u_{ij}) x_{ij}
  \\ &
  =
    \sum_{i = 1}^n \bigg[
      (a_{i0} - u_{i0}) x_{i0}
      + \sum_{j = 1}^{d - 1} (a_{ij} - u_{ij}) x_{ij}
      + (a_{id} - u_{id}) x_{id}
    \bigg]
  \\ &
  =
      c \sum_{i = 1}^n x_{i0}
    - c \sum_{i = 1}^n x_{id} .
\end{split}
\end{align}
Let $\sigma$ be the permutation capturing
the effect of $N$ on the input vector $\vec{x}$
(cf.~Section \ref{section:comparator-networks}).
Then, we have $x_{i0} = x_{\sigma(i)d}$ for every $1 \le i \le n$.
Since $\sigma$ is a permutation, we thus obtain
$\sum_{i = 1}^n x_{i0} = \sum_{i = 1}^n x_{\sigma(i)d}
= \sum_{i = 1}^n x_{id}$
and therefore \eqref{eq:network-weight-difference} evaluates to $0$.
\end{proof}

\newenvironment{examplequintupleblock}
{
  \begin{minipage}{0.32\textwidth}
  \centering
}
{\end{minipage}
}

\newcommand{\examplequintuplenetwork}[2]{%
  \begin{comparatornetwork}{5}{24}
    \widenodelabel{#1}
    \nodeconnection{{1,2},{3,4}}
    \widenodelabel{\grayout{\wt{\wzero}},\grayout{\wt{\wzero}},\grayout{\wt{\wzero}},\grayout{\wt{\wzero}},\grayout{\wt{\wzero}}}
    \nodeconnection{{1,3}}
    \nodeconnection{{2,5}}
    \widenodelabel{\grayout{\wt{\wzero}},\grayout{\wt{\wzero}},\grayout{\wt{\wzero}},\grayout{\wt{\wzero}},\grayout{\wt{\wzero}}}
    \nodeconnection{{2,3},{4,5}}
    \widenodelabel{\grayout{\wt{\wzero}},\grayout{\wt{\wzero}},\grayout{\wt{\wzero}},\grayout{\wt{\wzero}},\grayout{\wt{\wzero}}}
    \nodeconnection{{1,2},{3,4}}
    \widenodelabel{\grayout{\wt{\wzero}},\grayout{\wt{\wzero}},\grayout{\wt{\wzero}},\grayout{\wt{\wzero}},\grayout{\wt{\wzero}}}
    \nodeconnection{{2,3}}
    \widenodelabel{#2}
  \end{comparatornetwork}
}

As a special case, Definition~\ref{def:propfunfull} and
Theorem~\ref{theorem:propfunfull} are applicable to a network consisting of a
single comparator.
In fact, the preceding example, Example~\ref{example:singleton}, illustrates
this case, since the weights there are chosen so that that
$\matsym{U} = \propfunfull(\matsym{A})$.
Next, we show a larger example.

\begin{example}
  \label{example:black-box-quintuple}
  The weight propagation function $\propfunfull$ for the weights of any network
  $N$ is rather humble.
  Yet, it manages to push all the weight to the output in the special case that
  only the initial input weights are nonzero and they are all equal.
  Therefore, in this case of \emph{uniform} input weights, $\propfunfull$ is
  optimal in terms of moving weights forward.
  The following sorting network on five wires
  with initial and final weights shown on the left and right, respectively,
  illustrates this case.
  Weights kept intact are shown in gray.
  \tikzset{snpicturestyle/.style={snpicturestyledense}}
  \begin{center}
    \begin{examplequintupleblock}\examplequintuplenetwork{\wt{10},\wt{10},\wt{10},\wt{10},\wt{10}}{\wt{\wzero},\wt{\wzero},\wt{\wzero},\wt{\wzero},\wt{\wzero}}\end{examplequintupleblock}
    \hspace{40pt}
    \begin{examplequintupleblock}\examplequintuplenetwork{\wt{\wzero},\wt{\wzero},\wt{\wzero},\wt{\wzero},\wt{\wzero}}{\wt{10},\wt{10},\wt{10},\wt{10},\wt{10}}\end{examplequintupleblock}
  \end{center}
  These kinds of wire weights arise in practice in the context of ASP
  optimization statements with uniform weights.
  We will address the connection between weight propagation and ASP more
  thoroughly in Section~\ref{section:application}, however, we note here that
  our focus therein lies particularly in handling optimization statements
  with non-uniform weights.
  To this end, in the following we extend the usefulness of
  weight propagation to settings with more varied input weights.
  \eofex
\end{example}

To improve upon the lacking granularity in the discussed weight propagation
technique, we wish to propagate weights in smaller steps, spanning parts of
networks at a time.
We formulate these steps by constructing a weight propagation function
$\propfun{D}$ parameterized by a decomposition $D$ of the comparator network
$N$ at hand.
The role of the decomposition parameter $D$ is to determine components over
which weight propagation can be carried out gradually.
The intended design of the function $\propfun{D}$ is such that for example,
given a decomposition $D = \set{\tuple{N, \eset{1}{n}, \eset{1}{d}}}$ of $N$
where the entire network is treated as a single component, we replicate the
black-box behavior of $\propfunfull$.
For another example, given a decomposition $D = \eset{S_1}{S_m}$ in
which every comparator $C_k = \tuple{i, j, l} \in N$ is placed in a
separate component $S_k = \tuple{C_k, \set{i, j}, \set{l}}$,
the function will propagate a maximal amount of weight forward
over individual comparators at a time.
We call these two types of decompositions trivial and refer to
the end of this section for more complex, non-trivial ones that represent
intermediate decompositions between them. However,
before stating the formal definition of $\propfun{D}$, we lay out an
example of its intended outcome based on a trivial,
fine-grained decomposition $D$.

\begin{example}
  The following illustrates weight propagation steps over the comparators of
  a sorting network $N$ on four wires starting with the initial wire weights
  $\matsym{A}$ on the very left and ending in the fully propagated wire
  weights $\propfunfull(\matsym{A})$ on the very right. In each transition
  between a pair of diagrams, the comparators of a single layer are used
  independently as the basis of propagation.
  \tikzset{snpicturestyle/.style={snpicturestylelarge}}
  \begin{center}
    \begin{comparatornetwork}{4}{16}
      \widenodelabel{\wt{40},\wt{50},\wt{90},\wt{70}}
      \nodeconnection{{1,2},{3,4}}
      \widenodelabel{\wt{0},\wt{0},\wt{0},\wt{0}}
      \nodeconnection{{1,3}}
      \nodeconnection{{2,4}}
      \widenodelabel{\wt{0},\wt{0},\wt{0},\wt{0}}
      \nodeconnection{{2,3}}
      \widenodelabel{\wt{0},\wt{0},\wt{0},\wt{0}}
    \end{comparatornetwork}
    \hfill
    \begin{comparatornetwork}{4}{16}
      \widenodelabel{\wt{\wzero},\wt{10},\wt{20},\wt{\wzero}}
      \nodeconnection{{1,2},{3,4}}
      \widenodelabel{\wt{40},\wt{40},\wt{70},\wt{70}}
      \nodeconnection{{1,3}}
      \nodeconnection{{2,4}}
      \widenodelabel{\wt{0},\wt{0},\wt{0},\wt{0}}
      \nodeconnection{{2,3}}
      \widenodelabel{\wt{0},\wt{0},\wt{0},\wt{0}}
    \end{comparatornetwork}
    \hfill
    \begin{comparatornetwork}{4}{16}
      \widenodelabel{\wt{\wzero},\wt{10},\wt{\wzero},\wt{20}}
      \nodeconnection{{1,2},{3,4}}
      \widenodelabel{\wt{\wzero},\wt{\wzero},\wt{30},\wt{30}}
      \nodeconnection{{1,3}}
      \nodeconnection{{2,4}}
      \widenodelabel{\wt{40},\wt{40},\wt{40},\wt{40}}
      \nodeconnection{{2,3}}
      \widenodelabel{\wt{0},\wt{0},\wt{0},\wt{0}}
    \end{comparatornetwork}
    \hfill
    \begin{comparatornetwork}{4}{16}
      \widenodelabel{\wt{\wzero},\wt{10},\wt{\wzero},\wt{20}}
      \nodeconnection{{1,2},{3,4}}
      \widenodelabel{\wt{\wzero},\wt{\wzero},\wt{30},\wt{30}}
      \nodeconnection{{1,3}}
      \nodeconnection{{2,4}}
      \widenodelabel{\wt{0},\wt{0},\wt{0},\wt{0}}
      \nodeconnection{{2,3}}
      \widenodelabel{\wt{40},\wt{40},\wt{40},\wt{40}}
    \end{comparatornetwork}
  \end{center}
  To understand the above, let us focus on the comparator on the
  top left with input weights $40$ and $50$ at the beginning. Going from
  the first to the second diagram, an amount of $40$ is extracted from both of
  these weights and pushed over the comparator to its immediate output.
  In fact, this is precisely the same step as carried out in isolation in
  Example~\ref{example:singleton}.
  Moreover, the entire weight propagation process depicted here consists
  of repetitions of similar steps performed separately.  In this way, the
  network is taken as a white box with structure that guides the weight
  propagation process in fine detail.
  This is in contrast to the black-box treatment of the network in
  Example~\ref{example:black-box-quintuple}.
  \eofex
\end{example}

To ease the formal definition of the weight propagation function
$\propfun{D}$ for decompositions $D$, we first define a version,
$\propfun{C}$, for confined networks $C$, in order to express
individual propagations. 

\begin{definition}
Given wire weights $\matsym{A} = \matof{a_{ij}}$
for a network $N$ on $n$ wires and $d$ layers,
a confined comparator network $C = \tuple{N, I, E}$,
and $c = \min \sel{a_{ij}}{1 \le i \le n, j = \min E - 1}$,
the \emph{weight propagation function} $\propfun{C}$ maps $\matsym{A}$ to
the wire weights $\propfun{C}(\matsym{A}) = \matsym{U} = \matof{u_{ij}}$
for $N$ defined by
\begin{align*}
  u_{ij} =
  \begin{cases}
    a_{ij} - c , & \text{if $i \in I$ and $j = \min E - 1$,} \\
    a_{ij}     , & \text{if $\tuple{i, j} \in (\eset{1}{n} \times \eset{1}{d}) \setminus (I \times \set{\min E - 1, \max E})$,} \\
    a_{ij} + c , & \text{if $i \in I$ and $j = \max E$}.
  \end{cases}
\end{align*}
\end{definition}

\begin{lemma}
\label{lemma:propfun-confined}
Given wire weights $\matsym{A}$ and $\matsym{U} = \propfun{C}(\matsym{A})$
for a network $N$ of a confined comparator network $C$,
for any input vector $\vec{x}$, it holds that
$\weightfun{N}{\matsym{A}}(\vec{x}) = \weightfun{N}{\matsym{U}}(\vec{x})$.
\end{lemma}

The proof of Lemma \ref{lemma:propfun-confined} is analogous to the proof
of Theorem \ref{theorem:propfunfull} and is thus omitted.
One may think of $\propfun{C}$ as the function $\propfunfull$ affecting
only the inputs and outputs of a particular component $C$.

\newcommand{\exampleconfinednetwork}[2]{%
  \begin{comparatornetwork}{5}{20}
    \widenodelabel{\grayout{\wt{80}},\grayout{\wt{20}},\grayout{\wt{\wzero}},\grayout{\wt{70}},\grayout{\wt{30}}}
    \nodeconnection{}
    \widenodelabel{#1}
    \nodeconnection{{1,4}}
    \nodeconnection{{3,5}}
    \widenodelabel{\wt{50},\grayout{\wt{\wzero}},\wt{90},\wt{90},\wt{80}}
    \nodeconnection{{4,5}}
    \widenodelabel{#2}
    \nodeconnection{}
    \widenodelabel{\grayout{\wt{30}},\grayout{\wt{70}},\grayout{\wt{20}},\grayout{\wt{50}},\grayout{\wt{20}}}
  \end{comparatornetwork}
}

\newenvironment{examplemediumblock}
{
  \begin{minipage}{0.25\textwidth}
}
{\end{minipage}
}

\begin{example}
  Consider a weight propagation step over a confined network
  $C = \tuple{N, I, E}$ where the allowed wires are $I = \set{1, 3, 4, 5}$ and
  levels $E = \set{2, 3}$. The gray numbers indicate weights out of the
  scope of $C$. Only the leftmost and rightmost weights
  are modified, the middle ones stay intact.
  The specific comparators in the network do not matter, as long as
  they are confined to $I$ and $E$.
  \tikzset{snpicturestyle/.style={snpicturestyledense}}
  \begin{center}
    \begin{examplemediumblock}\exampleconfinednetwork{\wt{90},\grayout{\wt{40}},\wt{10},\wt{20},\wt{50}}{\wt{60},\grayout{\wt{50}},\wt{\wzero},\wt{10},\wt{30}}\end{examplemediumblock}
    \hspace{40pt}
    \begin{examplemediumblock}\exampleconfinednetwork{\wt{80},\grayout{\wt{40}},\wt{\wzero},\wt{10},\wt{40}}{\wt{70},\grayout{\wt{50}},\wt{10},\wt{20},\wt{40}}\end{examplemediumblock}
  \end{center}
  \eofex
\end{example}

We want to order confined networks in such a way that when propagating weights
over them, each propagation step picks up from where the previous step left
off, pushing weights forward naturally.
To this end, we write $C \le C'$ for pairs of mutually compatible confined
networks $C = \tuple{N, I, E}$ and $C' = \tuple{N', I', E'}$ that satisfy
$\min E \le \max E'$.
The intuition behind $C \le C'$ is that $C$ cannot possibly depend on the output
of $C'$ and can thus be propagated over first.
The weight propagation function $\propfun{D}$ for decompositions
$D$ based on compatible components is defined in the following, where we follow
the convention for function composition by which $(f \circ g)(x) = f(g(x))$.

\begin{definition}
Given a decomposition $D$ consisting of confined networks
$S_1 \le \cdots \le S_m$,
the \emph{weight propagation function} $\propfun{D}$
is defined as
$\propfun{D} = \propfun{S_m} \circ \cdots \circ \propfun{S_1}$.
\end{definition}

\begin{theorem}
\label{theorem:propfun-decomposition}
Given wire weights $\matsym{A}$ and $\matsym{U} = \propfun{D}(\matsym{A})$
for a comparator network $N$ and a decomposition $D$ of $N$,
for any input vector $\vec{x}$, it holds that
$\weightfun{N}{\matsym{A}}(\vec{x}) = \weightfun{N}{\matsym{U}}(\vec{x})$.
\end{theorem}

\begin{proof}
Let $S_1 \le \cdots \le S_m$ be the confined networks in $D$
and write $\matsym{U}_i = \propfun{\eset{S_1}{S_{i}}}(\matsym{A})$
for every $1 \le i \le m$
so that $\matsym{U}_0 = \matsym{A}$,
$\matsym{U}_i = \propfun{\set{S_i}}(\matsym{U}_{i - 1})$
and $\matsym{U}_m = \matsym{U}$.
Lemma~\ref{lemma:propfun-confined} proves each of the equalities
$
\weightfun{N}{\matsym{U}_0}(\vec{x})
= \cdots =
\weightfun{N}{\matsym{U}_m}(\vec{x})
$
and thus
$
\weightfun{N}{\matsym{A}}(\vec{x})
=
\weightfun{N}{\matsym{U}_0}(\vec{x})
=
\weightfun{N}{\matsym{U}_m}(\vec{x})
$.
\end{proof}

We end this section by detailing a family of \emph{sparse} decompositions
$D_k(N)$ for use with any network $N$ and the weight propagation function
$\propfun{D_k(N)}$.
The decompositions are parameterized by a sparseness factor $k$, which controls
the rough fraction $1 / k$ of nonzero weights remaining after weight
propagation.
These decompositions represent hybrids between the trivial ones
in terms of numbers of nonzero weights remaining after propagation.
In this way, they enable to experiment with the effectiveness of weight
propagation in more detail, which we do later in
Section~\ref{section:application}.
For context, recall the trivial decompositions in which all comparators are
either placed in a single component or separate components.
Propagation based on these decompositions results in either minimally or
maximally many weights being propagated.
In particular, in the expected setting where the initial weights are zero for
all but the input, this difference is reflected in the numbers of nonzero
weights that remain after propagation as follows.
When all comparators are in a single component, the number of remaining nonzero
weights is at most $2n$, and when all comparators are
in separate components, it is at most $(0.5 n + 1) \times d$.
As an alternative, the decomposition can be designed to provide a balance
between these two extremes.

The sparseness factor $k$ is a positive integer that reduces the number of
weights remaining after propagation by roughly a factor of $k$.
This is done by placing propagated weights only on levels that are multiples of
$k$, in addition to the last level.
We first define it formally and then show examples of how to create and use it.
To form the decomposition, the comparators in $N$ are first partitioned
based on which of the following ranges their levels $l$ fall into:
$E_1 = \eset{1}{k}$,
$E_2 = \eset{k + 1}{2k}$,
$\dots$,
$E_{p} = \eset{k p + 1}{d}$
where $p = \floor{d / k}$.
That is, the first $k$ layers are in one component, the next $k$ layers in
another, and so on.
Then the components are refined individually.
More specifically, for each $1 \le j \le p$, the wires $\eset{1}{n}$ are partitioned
into a number $n_j$ of minimal sets $I_{1j}, \dots, I_{n_{j}j}$ such that
for each comparator associated with $E_j$, its two wires fall into the same set.
Moreover, any wires not adjacent to those comparators, if any, form one of the
sets.
This amounts to a partition of the comparators into connected
components described indirectly in terms of wires.
The final decomposition is then obtained as
$
D_k(N) = \sel{
  \tuple{N_{ij}, I_{ij}, \sel{l \in E_j}{l \le d}}
}{
  1 \le j \le p,
  1 \le i \le n_j
}
$
where each network $N_{ij}$ is 
$N_{ij} = \sel{\tuple{i', j', l} \in N}{i', j' \in I_{ij}, l \in E_j}$.
One may observe that this construction places all comparators in separate
components when $k = 1$, and in the same component when $k \ge d$ and the network
is connected.
Therefore, for connected networks, $D_k(N)$ generalizes the trivial
decompositions.

\newcommand{\emphconnection}[3]{
\draw[snlinestylethick] (#1,#3) -- (#1+#2+4,#3);
}

\newcommand{\emphconnectionhorizontal}[2]{
  \emphconnection{\value{sncolumncounter}}{#2}{#1}%
}

\newcommand{\exampledecompositioncreation}[1]{%
  \begin{minipage}{0.30\textwidth}
  \centering
  \begin{comparatornetwork}{10}{29}
    #1

    \widenodelabel{}
    \nodeconnection{{1,2},{4,5},{6,7},{8,9}}

    \widenodelabel{}
    \nodeconnection{{2,3},{6,8}}
    \nodeconnection{{7,10}}

    \widenodelabel{}
    \nodeconnection{{1,4},{7,8},{9,10}}
    \nodeconnection{{3,5}}

    \widenodelabel{}
    \nodeconnection{{2,4},{6,7},{8,9}}

    \widenodelabel{}
    \nodeconnection{{1,2},{3,4},{7,8}}

    \widenodelabel{}
    \nodeconnection{{3,7}}

    \widenodelabel{}
  \end{comparatornetwork}
  \end{minipage}
}

\begin{example}
  The decomposition $D_2(N)$ can be formed in two steps for the network $N$ on
  $10$ wires shown below on the left.
  First, the layers of the network are partitioned and then the wires within
  those partitions are further partitioned by identifying connected components.
  \label{example:decomposition-creation}
  \tikzset{snpicturestyle/.style={snpicturestyledensedense}}
  \begin{center}
    \tikzset{snlinestyle/.style={}}
    \exampledecompositioncreation{
    }
    \exampledecompositioncreation{
      \draw[dashed] (10,0)--(10,11);
      \draw[dashed] (19,0)--(19,11);
    }
    \tikzset{snlinestyle/.style={snlinestylethick}}
    \exampledecompositioncreation{
      \draw[draw=white, line width=4pt] (10,0)--(10,11);
      \draw[draw=white, line width=4pt] (19,0)--(19,11);

      \emphconnection{3}{0}{2}
      \emphconnection{3}{0}{6}
      \emphconnection{3}{1}{7}
      \emphconnection{3}{0}{8}

      \emphconnection{12}{1}{4}
      \emphconnection{12}{1}{7}
      \emphconnection{12}{1}{8}
      \emphconnection{12}{1}{9}

      \emphconnection{21}{0}{3}
      \emphconnection{21}{0}{7}
    }
  \end{center}
  The transition from the first to the second diagram illustrates the partition
  of layer levels into the sets
  $E_1 = \set{1, 2}$, $E_2 = \set{3, 4}$, and $E_3 = \set{5, 6}$.
  In the second transition, these parts are further refined
  by partitioning wires in the context of $E_1$ to
  $I_{11} = \set{1, 2, 3}$, $I_{21} = \set{4, 5}$, $I_{31} = \eset{6}{10}$,
  in the context of $E_2$ to
  $I_{12} = \set{1, 2, 4}$, $I_{22} = \set{3, 5}$, $I_{32} = \eset{6}{10}$,
  and in the context of $E_3$ to
  $I_{13} = \set{1, 2}$, $I_{23} = \set{3, 4, 7, 8}$, $I_{33} = \set{9, 10}$.
  Note that the parts $I_{12}$ and $I_{22}$ are indeed distinct, despite their
  seeming overlap in the diagram.
  \eofex
\end{example}

\newcommand{\exampledecompositionpropagation}[3]{%
  \begin{minipage}{0.30\textwidth}
    \centering
    \begin{comparatornetwork}{5}{20}
      \widenodelabel{#1}
      \emphconnectionhorizontal{2}{0}
      \singlenodeconnection{1}{2}{circle,fill=black}
      \singlenodeconnection{4}{5}{circle,fill=white,draw}
      \addtocounter{sncolumncounter}{1}
      \widenodelabel{\wt{\wzero},\wt{\wzero},\wt{\wzero},\wt{\wzero},\wt{\wzero}}
      \singlenodeconnection{2}{3}{circle,fill=black}
      \addtocounter{sncolumncounter}{1}
      \widenodelabel{#2}
      \emphconnectionhorizontal{4}{1}
      \singlenodeconnection{1}{4}{rectangle,inner sep=2.7pt,fill=white,draw}
      \addtocounter{sncolumncounter}{1}
      \singlenodeconnection{3}{5}{rectangle,inner sep=2.7pt,fill=black}
      \addtocounter{sncolumncounter}{1} 
      \widenodelabel{\wt{\wzero},\wt{\wzero},\wt{\wzero},\wt{\wzero},\wt{\wzero}}
      \singlenodeconnection{2}{4}{rectangle,inner sep=2.7pt,fill=white,draw}
      \addtocounter{sncolumncounter}{1}
      \widenodelabel{#3}
    \end{comparatornetwork}
  \end{minipage}
}

\tikzset{snlinestyle/.style={snlinestylethick}}

\begin{example}
  The following shows propagation over a network $N$ taken from the
  top left of the network in Example~\ref{example:decomposition-creation}.
  The network is decomposed into the four confined networks in $D_2(N)$
  highlighted with thick lines and distinct markers.
  The separation to circles and squares stems from levels,
  and to black and white from wires.
  The transitions illustrate propagation over the two components with circles in
  any order, followed by the two components with squares in any order of
  colors.
  Observe that all nonzero weights are on levels $0$, $2$ and $4$ in the end.
  The fact that these are multiples of two stems from the choice of $k = 2$ in
  $D_k(N)$.
  \label{example:black-box-decomposition}
  \tikzset{snpicturestyle/.style={snpicturestyledense}}
  \begin{center}
    \exampledecompositionpropagation{\wt{20},\wt{90},\wt{80},\wt{30},\wt{70}}{\wt{\wzero},\wt{\wzero},\wt{\wzero},\wt{\wzero},\wt{\wzero}}{\wt{\wzero},\wt{\wzero},\wt{\wzero},\wt{\wzero},\wt{\wzero}}
    \exampledecompositionpropagation{\wt{\wzero},\wt{70},\wt{60},\wt{\wzero},\wt{40}}{\wt{20},\wt{20},\wt{20},\wt{30},\wt{30}}{\wt{\wzero},\wt{\wzero},\wt{\wzero},\wt{\wzero},\wt{\wzero}}
    \exampledecompositionpropagation{\wt{\wzero},\wt{70},\wt{60},\wt{\wzero},\wt{40}}{\wt{\wzero},\wt{\wzero},\wt{\wzero},\wt{10},\wt{10}}{\wt{20},\wt{20},\wt{20},\wt{20},\wt{20}}
  \end{center}
  \eofex
\end{example}

\section{Application of Weights and Sorting Networks to Answer Set Programming}
\label{section:application}

In this section, we focus on the application of comparator networks and
propagated weights to solving optimization problems expressed in ASP.
Specifically, we present a novel approach to optimization rewriting and prove
the correctness of the approach in Section~\ref{section:optimization-rewriting}.
Then, we give formal and experimental results proving the potential for
exponential improvements in time consumption when solving an example
family of programs in Section~\ref{section:formal}.
These promising performance indicators are complemented in
Section~\ref{section:challenges} with a discussion on potential drawbacks of
the approach concerning the impact of weight propagation on unit propagation.
Finally, a thorough experimental evaluation is given in
Section~\ref{section:experimental} in order to asses practical performance on
an extensive set of benchmarks stemming from prior ASP competitions.
The benchmarks are augmented with some newly generated instances to better match
the state-of-the-art performance level of contemporary ASP optimization.

\subsection{Optimization Rewriting using Sorting Networks}
\label{section:optimization-rewriting}

As demonstrated in Section~\ref{section:translation}, a comparator
network $N$ can be translated into an answer set program $\aspify{N}$ that captures
the wire values of $N$ when the input is encoded in atoms.
This can be used to translate any given optimization program
$\tuple{P,e}$ into another one $\tuple{P \union \aspify{N}, e'}$ that
yields the same answer sets with some added atoms and
unchanged optimization values.
The key observation relevant to this paper is that this presents an
opportunity to craft the new pseudo-Boolean expression $e'$ in terms of atoms
and weights that express the wire values and wire weights of an appropriately
chosen network $N$.
The techniques from Section~\ref{section:propagation} are applicable to
determining those wire weights: we may calculate them by propagating weights
taken from the original pseudo-Boolean expression $e$ across the network $N$.
A key benefit of this is that the fresh atoms in $\aspify{N}$ can help
tremendously in branch-and-bound optimization.
Namely, as will be demonstrated formally in Section~\ref{section:formal},
optimization rewriting using specifically sorting networks can yield even
exponential savings in terms of the numbers of learned nogoods that stem from
optimization statements.
Moreover, sorting networks can be generated efficiently with well-known schemes,
such as Batcher's odd-even merge sorting networks \cite{Batcher68:afips}.
Hence, sorting networks are a good starting point for $N$.
However, currently known practically feasible sorting networks are $O(n (\log
n)^2)$ in size, which is a problem when rewriting large optimization
statements.
Thus it can pay off to sacrifice some of the benefits of the fresh atoms by
using smaller networks that sort only some input sequences or only some
subsequences of inputs.
This question of which network to use is further addressed in experiments on
various networks in Section~\ref{section:experimental}.

Before formal results, we recall the \emph{splitting set theorem}
\cite{liftur94a} formulated for the respective \emph{bottom} and
\emph{top} programs $B$ and $T$ such that the rule bodies of $T$
may refer to atoms defined by the rule heads of $B$, but not vice versa:

\begin{proposition}
\label{prop:bottom-and-top}
An interpretation $I\subseteq\sig{B\union T}$ is an answer set of $B\union T$
iff
(i)
$I_B=I\intersection\sig{B}$ is an answer set of $B$ and
(ii)
$I_T=I\intersection\sig{T}$ is an answer set of
$T\union\sel{a}{a\in I_B \intersection\sig{T}}$.
\end{proposition}

Now we are ready to present the main formal result of this paper.
The result ensures that the answer sets of a program and the respective
optimization values are principally unchanged when its optimization statement
is rewritten based on a network whose translation is added to the answer set program.
The rewritten optimization statement contains the atoms of the translation
weighted by the original weights after propagating them over the network.
For convenience, we consider cases where the original pseudo-Boolean expression
$e$ being minimized is given in terms of literally the same atoms $\pred{x}{i0}$
that are used to encode the input vector in the translation $\aspify{N}$.

\begin{theorem}
\label{theorem:answer-set-bijection}
Let $N$ be a comparator network on $n$ wires,
$D$ a decomposition of $N$,
$\matsym{A} = \matof{a_{ij}}$ and
$\matsym{U} = \matof{u_{ij}} = \propfun{D}(\matsym{A})$ wire weights for $N$
where $a_{ij} = 0$ for every level $j>0$,
$\aspify{N}$ the translation of $N$ into an answer set program,
$P$ an answer set program such that $j = 0$ for every $\pred{x}{ij} \in \sig{P}$,
and
$e=\sum_{i = 1}^n a_{i0} \pred{x}{i0}$ and
$e'=\sum_{i = 1}^n \sum_{j = 0}^d u_{ij} \pred{x}{ij}$
pseudo-Boolean expressions.
Then there is a bijection $f : \as{P} \to \as{P \union \aspify{N}}$ such that
$e(M) = e'(f(M))$ for every $M \in \as{P}$.
\end{theorem}

\begin{proof}
For an interpretation $M \subseteq \sig{P}$, define a vector
$\vec{x}_M$ that has at index $1 \le i \le n$ value $1$, if
$\pred{x}{i0} \in M$, and $0$, otherwise.
Let $X_M$ be the set $X_M = \sel{\pred{x}{ij}}{x_{ij} = 1}$ where
$x_{ij}$ refer to the wire values $\wirevalues{N}{\vec{x}_M}$
of the network $N$ given the input $\vec{x}_M$.
In the following, we prove that $f(M)=M \union X_M$ defined for
$M\subseteq\sig{P}$ provides the bijection of interest. First, let us
establish that $f$ maps an answer set $M\in\as{P}$ to an answer set
$f(M) \in \as{P \union \aspify{N}}$.
This holds because $M \union X_M \in \as{P \union \aspify{N}}$, which
follows from the ``if'' direction of Proposition~\ref{prop:bottom-and-top}.
The first requirement in the proposition is satisfied by the
assumption $M \in \as{P}$, and the second by the fact that
$X_M\in\as{\aspify{N} \union \sel{\pred{x}{i0}}{\pred{x}{i0} \in M}}$,
which follows from Lemma~\ref{lemma:correspondence}.

Second, the function $f$ is an injection, i.e., it maps all inputs $M
\in \as{P}$ to distinct outputs $f(M)$.  This holds because the input
$M$ can be recovered from the output $f(M)$.  Namely,
$f(M) \intersection \sig{P} = (M \union X_M) \intersection \sig{P} = M$,
since
$X_M \intersection \sig{P} = \set{\pred{x}{i0} \in X_M} \subseteq M$.

Third, the function $f$ is a surjection, i.e., for every output $I \in
\as{P \union \aspify{N}}$ there is an input $M \in \as{P}$ such that
$f(M)=I$. Indeed, by the ``only if'' direction of
Proposition~\ref{prop:bottom-and-top}, every answer set
$I\in\as{P\union\aspify{N}}$
can be split into the answer sets
$M = I \intersection \sig{P} \in \as{P}$
and
$X = I \intersection \sig{\aspify{N}} \in
\as{\aspify{N} \union \sel{\pred{x}{i0}}{\pred{x}{i0} \in M}}$.
By Lemma~\ref{lemma:correspondence}, the latter must be $X = X_M$.
Therefore, it follows that $I = M \union X = M \union X_M = f(M)$.
Finally, for every input $M \in \as{P}$,
we have
$
e'(f(M))
= e'(X_M)
= \sum_{\pred{x}{ij} \in X} u_{ij}
= \sum_{i = 1}^n \sum_{j = 0}^d u_{ij} x_{ij}
= \weightfun{N}{\matsym{U}}(\vec{x}_M)
= \weightfun{N}{\matsym{A}}(\vec{x}_M)
= \sum_{i = 1}^n a_{i0} x_{i0}
= \sum_{\pred{x}{i0} \in X} a_{i0}
= e(X)
= e(M)$.
The third and the third to last equalities are due to
Lemma~\ref{lemma:correspondence} while the fifth is due to
Lemma~\ref{theorem:propfun-decomposition}.
\end{proof}

It immediately follows that optimal answer sets of the program are preserved.

\begin{corollary}
\label{corollary:optimal-answer-set-bijection}
Let $N$, $n$, $D$, $\matsym{A}$, $\matsym{U}$, $\aspify{N}$, $e$, and $e'$ be
defined as in Theorem~\ref{theorem:answer-set-bijection}.
Then there is a bijection $f : \as{P} \to \as{P \union \aspify{N}}$ such that
each $M \in \as{P}$ is optimal for $\pair{P}{e}$
iff
$f(M)$ is optimal for $\pair{P \union \aspify{N}}{e'}$.
\end{corollary}

\subsection{Formal Performance Analysis}
\label{section:formal}

In this section, we formally analyze the rewriting techniques from
Section~\ref{section:optimization-rewriting}.
Our focus is on the performance of an optimizing ASP solver without
and with optimization rewriting.
We obtain a result showcasing an exponential improvement in favor of
optimization rewriting based on sorting networks and weight propagation on an
example family of optimization programs.
The programs in the family are designed to select subsets of size at least $k$
atoms from among $n$ atoms and to minimize the number of picked atoms.
The subsets with precisely $k$ atoms are then optimal, and the role of
optimization is to rule out all the subsets larger than that.

The result applies in principle to any ASP solver that performs optimization via
\emph{branch-and-bound} search on the optimization value and operates with
nogoods and \emph{propagators} in a manner we detail in the analysis.
These background assumptions reflect existing
conflict driven nogood learning (CDNL) solving techniques on a
simplified level and with the additional assumption that learned
nogoods are kept in memory indefinitely without deleting them.

We begin the rest of this section by briefly discussing the relevant mechanics
of optimizing ASP solvers and necessary formal preliminaries.
Then we prove statements concerning solving difficulty
without and with optimization rewriting.
Finally, we investigate the behavior of actual ASP solvers on sample programs
from the family.
These experimental results likewise show a significant improvement in favor
of optimization rewriting.
This confirmation is meaningful since actual ``black-box'' ASP solvers
generally carry intricate features beyond those of any formal model of a solver.
These experiments that are linked with the formal analysis are later
complemented in Section~\ref{section:experimental} by a broader evaluation on
an extensive set of heterogeneous benchmarks of greater practical relevance.

Regarding ASP optimization, as discussed above, we concentrate on minimization
using the branch-and-bound optimization strategy.
A solver employing this strategy on an optimization program $\pair{P}{e}$
implements a recursive procedure in which it
\begin{enumerate}
  \item
    takes as input a range of integers known to contain the optimal value of the
    pseudo-Boolean expression $e$, and which is initially huge,
  \item
    partitions the range into two nonempty ranges by some heuristic procedure,
    or returns a value in the range if the range contains only a single
    value,
  \item
    searches for an answer set $M$ of $P$ with a value $e(M)$ within the
    lower range, and
  \item
    recursively calls the procedure on either the lower range adjusted to end in
    $e(M)$ or the upper range, depending on whether an answer set was found or
    not, respectively.
\end{enumerate}
Given an optimization program with at least some answer set,
this procedure will eventually find one of the $e$-optimal ones.
The requirement of bounding $e(M)$ to a low range can be represented as a
set of nogoods, i.e., a constraint.
However, because such a set of nogoods is generally prohibitively large, it is
typically represented indirectly by a propagator.
A propagator is essentially a procedure for determining whether an assignment
conflicts with a specific constraint, or is close to conflicting with it, and
which can \emph{explain} such conflicts in terms of nogoods.
Propagators fit into a solving process that implements
\emph{lazy generation} of constraints as follows.
To begin with, an input answer set program is split into two parts:
a regular part and a part with
constraints 
that are initially abstracted away from the solver.
Then, the solver begins a search for an answer set of the regular part only.
In order to adhere to the abstracted constraints,
the solver consults propagators specific to the constraints at various points in
the search process on whether the current assignment satisfies all of them.
As long as all constraints are satisfied, the solver proceeds as usual.
However, in the event that a propagator reports a conflict between its current
assignment and a constraint, the solver learns the nogoods given by the
respective propagator as an explanation for the conflict.
The solver then resolves the conflict, which is now reflected in the nogoods
that the solver is aware of, and continues the search.

This optimization process has the following key properties that we make use of.
The \emph{first} key property is that exactly one of the searches in Step 3\
discovers an answer set $M$ with an optimal value $v$ for $e(M)$, and another
one of the searches imposes a bound $v - 1$ for $e$, and which proves the
optimality of $v$ by yielding no answer sets.
In a hypothetical, ideal scenario,
the search for $M$ proceeds without conflicts and
no searches beyond these two need to be done.
Even in such a best-case scenario, the challenge in the optimization task
includes the inescapable difficulty of proving the optimality of $v$.
That difficulty, however, can be considerable even in the best-case scenario.
In our analysis, we focus on this case for simplicity of analysis, and due to
its computationally challenging and integral role in the optimization task.
That is, we consider the difficulty of searching with a bound $v - 1$
right below the optimal value $v$.
The \emph{second} key property is that the described lazy generation of constraints
brings a variable fraction of nogoods from constraints to the knowledge of the solver
during search. 
This fraction can range from zero to one, even in practice.
In some searches, an answer set is found before
a constraint leads to a significant number of conflicts,
in which case the fraction is low.
In some other searches, a constraint is central in rejecting a large number of
candidate answer sets, and perhaps all otherwise feasible answer sets,
in which case the fraction can be high.
The family of optimization programs we present represents an extreme case,
with a fraction of exactly one.
The family builds on certain ``bottle-neck'' constraints that have been used to
illustrate differences between SMT decision solvers that use either propagators
or encodings \cite{abniolrost13}.
The analysis required here is complicated by both a shift to the context of ASP
from SMT and particularly the consideration of optimization
instances instead of decision instances.
The ASP optimization programs considered here
are parameterized by non-negative integers $n \le k$,
they have $\binom{n}{k}$ optimal answer sets,
and we accordingly name them \emph{binomial optimization programs}.
The answer sets of these programs consist of all subsets of at least $k$
atoms selected from $\range{\pred{x}{1}}{\pred{x}{n}}$.
From those answer sets, the ones with precisely $k$ atoms are optimal.
\begin{definition}
The \emph{binomial program}
$\binomproglt{n}{k}$
consists of the rules
$\set{\rg{\pred{x}{1}}{;}{\pred{x}{n}}}$
and
$\IF \AGGCOUNT{\rg{\pred{x}{1}}{;}{\pred{x}{n}}} < k$.
\end{definition}
\begin{definition}
The \emph{binomial optimization program}
$\binomprogopt{n}{k}$
is the optimization program
$\pair{\binomproglt{n}{k}}{\rangeplus{1 \pred{x}{1}}{1 \pred{x}{n}}}$.
\end{definition}
The optimal value for the objective function
$\rangeplus{1 \pred{x}{1}}{1 \pred{x}{n}}$ is $k$.
Therefore, applying branch-and-bound optimization to $\binomprogopt{n}{k}$
entails, as per the earlier discussion on optimization,
for the binomial program $\binomproglt{n}{k}$
to be solved once under the constraint that the objective function
$\rangeplus{1 \pred{x}{1}}{1 \pred{x}{n}}$ takes a value less than $k$.
Observe that this constraint is impossible to satisfy,
and that it can be represented by the set of nogoods
$
  \sel{
    \sel{\assigntrue{x}}{x \in X}
  }{
    X \subseteq \eset{\pred{x}{1}}{\pred{x}{n}},
    |X| = k
  }
$.
This infeasible decision problem is the target of our subsequent analysis.

We model propagators as simple functions that, in the event of a conflict,
designate a single violated nogood to be the explanation reported back to the
solver.
This is a streamlined definition in comparison to, e.g., the definition given by
Drescher and Walsh \citeyear{drewal12a} according to which propagators take as
arguments partial assignments that may or may not conflict with the constraint
and return sets of nogoods including at least one violated nogood on conflicts.
\begin{definition}
A \emph{propagator} $\pi$ for a constraint $\Gamma$ is a function
from partial assignments $A$ in conflict with $\Gamma$
to nogoods $\delta' \in \Gamma$ in conflict with $A$.
The constraint $\Gamma$ of a propagator $\pi$ is denoted by $\nogoodsof{\pi}$.
\end{definition}
In order to reason about the set of nogoods accumulated by calling a propagator during
search, we below formulate the concept of a history of partial assignments
provided as input to a propagator. Based on such a history, the propagator
generates explanatory nogoods.
\begin{definition}
A \emph{propagator call history} (PCH) for an answer set program $P$ and a propagator $\pi$
is a sequence $\range{A_1}{A_m}$ of partial assignments
such that for all $1 \le i \le m$, the partial assignment $A_i$ satisfies
$\nogoodsforsupp{P}$, $\range{\pi(A_1)}{\pi(A_{i-1})}$
and conflicts with $\nogoodsof{\pi}$.
\end{definition}
Intuitively, a PCH is a record of all the calls a solver makes to a propagator
before the first answer set is found, or before the search space is exhausted
while searching for one.
This definition reflects a number of assumptions we make in modeling the ASP
solving process.
For one, we assume that propagator-produced nogoods are never deleted and that
propagators are called only on partial assignments that satisfy all nogoods
known to the solver, including the nogoods previously generated by the
propagators themselves.
This assumption is behind the requirement in the definition for each assignment
$A_i$ to satisfy all nogoods generated in response to the earlier partial assignments
$\range{A_1}{A_{i-1}}$.
This makes our formal analysis feasible, but
technically demands an ASP solver with infinite memory.
In reality, ASP solvers manage memory by deleting some nogoods
periodically and possibly re-learning them later~\cite{GKS12:aij},
and this includes propagator-produced nogoods~\cite{drewal12a}.
The requirement that $A_i$ also satisfies $\nogoodsforsupp{P}$ reflects another
assumption: the solver makes sure that partial assignments are viable
supported model candidates of $P$ before calling propagators on them.
Enforcing consistency with the supported model semantics like here
is a well established method in ASP solving
\cite{GKS12:aij,DBLP:conf/lpnmr/AlvianoDLR15},
and therefore this assumption maintains
practical relevance of our results.

As mentioned, we take interest in programs that have no answer sets, since they
are important in optimality proofs.
When solving such answer set programs, any used propagators will need to be queried
sufficiently many times, so that the answer set program that is revealed to the solver has
no answer sets either.
We formalize this condition as a property of a PCH.
\begin{definition}
Let $P$ be an answer set program and $\pi$ a propagator such that $P \union \nogoodsof{\pi}$
has no answer sets.
A PCH $\range{A_1}{A_m}$ for $P$ and $\pi$ is \emph{complete} if $P \union
\eset{\pi(A_1)}{\pi(A_m)}$ has no answer sets.
\end{definition}

Given these notions, we are equipped to present a proposition on the
significant difficulty of solving a binomial program 
combined with a constraint that rejects all of its answer sets.
Here we use the length of a PCH as an abstract measure of that solving
difficulty, and in particular, the difficulty due to nogoods generated
by a propagator in order to represent the added constraint.
The length turns out to be exponential even in this simple case.
The result concerns a situation where no optimization rewriting takes place.
The proposition essentially states that a propagator that is responsible for the
optimization statement of a binomial optimization program has to generate an
exponential number of nogoods for the final unsatisfiability proof stage.
Afterwards, we give a result that instead concerns the case where sorting
network based optimization rewriting is used.
An exponential difference in outcomes will be apparent between these two
results.

\begin{proposition}
\label{prop:opposing-card}
Let $n$ and $k$ be non-negative integers
such that $k \le n$,
$\pi$ a propagator for the constraint
\begin{align*}
  \sel{
    \sel{\assigntrue{x}}{x \in X}
  }{
    X \subseteq \eset{\pred{x}{1}}{\pred{x}{n}},
    |X| = k
  },
\end{align*}
and let $\rg{A_1}{,}{A_m}$ be a complete PCH
for the answer set program $\binomproglt{n}{k}$ and $\pi$.
Then $m = \binom{n}{k}$.
\end{proposition}

\begin{proof}
Let $\Delta = \pi(A_1) \union \cdots \union \pi(A_m)$
be the set of nogoods produced by the propagator $\pi$
in response to the PCH $\rg{A_1}{,}{A_m}$.
On the one hand,
each nogood $\delta \in \nogoodsof{\pi}$
corresponds to an answer set
$\sel{x}{\assigntrue{x} \in \delta}$
of the answer set program $\binomproglt{n}{k}$
that also satisfies all the other nogoods, i.e., those in
$\nogoodsof{\pi} \setminus \set{\delta}$.
On the other hand, the clear unsatisfiability of
$\binomproglt{n}{k} \union \nogoodsof{\pi}$
and the completeness of $\rg{A_1}{,}{A_m}$ imply unsatisfiability of
$P \union \Delta$.
No $\delta \in \nogoodsof{\pi}$ can be excluded from $\Delta$
without giving up the unsatisfiability of $P \union \Delta$,
and therefore we must have $\nogoodsof{\pi} \subseteq \Delta$.
Hence $m \ge \binom{n}{k}$.
Also,
certainly $m = |\Delta| \le \binom{n}{k}$,
and therefore $m = \binom{n}{k}$.
\end{proof}

The following lemma is integral in proving our next result.
The lemma states that all nonempty nogoods over the output atoms of a sorting
network can be simplified into singleton nogoods.

\begin{lemma}
\label{lemma:opposing-nogood-equivalence}
Let $n$, $k$, and $d$ be non-negative integers
such that $k \le n$,
$N$ a sorting network of width $n$ and depth $d$,
$\aspify{N}$ the translation of $N$ into an answer set program,
$
\Lambda
=
\nogoodsforsupp{\aspify{N} \union \set{\eset{\pred{x}{10}}{\pred{x}{n0}}}}
$
the supported model constraint of that translation
combined with a choice rule on the input atoms of the network,
and
$\delta$ a nonempty nogood of positive signed literals over
$\range{\pred{x}{1d}}{\pred{x}{nd}}$.
Then
$
\Lambda \union \set{\delta}
\classicequiv
\Lambda \union \set{\set{\assigntrue{\pred{x}{id}}}}
$
where $i = \min \sel{j}{\assigntrue{\pred{x}{jd}} \in \delta}$.
\end{lemma}

\begin{proof}
Let $n$, $k$, $d$, $N$, $\aspify{N}$, $\Lambda$, $\delta$, and $i$ be as above.
Using similar reasoning as in the proof of Lemma~\ref{lemma:correspondence},
it can be shown that in each supported model $M \models \nogoodsforsupp{\aspify{N}}$,
the outputs
$\rg{\pred{x}{1d}}{,}{\pred{x}{nd}}$
are sorted such that false precedes true.
That is,
for each $1 \le j < k \le n$,
if $\pred{x}{jd} \in M$ then $\pred{x}{kd} \in M$.
We will use this to prove the lemma one supported model at a time.
To this end, let us consider any supported model
$M \models \Lambda$
of the translation of the network.
On the one hand, if $\pred{x}{id} \in M$,
then the mentioned sortedness property guarantees that
also $\pred{x}{jd} \in M$ for each $j \in \eset{i}{n}$,
which particularly includes each $j$ such that $\pred{x}{jd} \in \delta$,
and therefore $\delta \subseteq M$.
On the other hand, if $\pred{x}{id} \not\in M$,
then $\delta \not\subseteq M$.
Hence, $M$ satisfies $\set{\assigntrue{\pred{x}{id}}}$ iff it satisfies $\delta$.
As this holds for any $M \models \Lambda$,
we obtain the consequent of the lemma.
\end{proof}

Now it can be shown that the addition of a sorting network
to the setting considered in Proposition~\ref{prop:opposing-card}
yields an improvement in solving difficulty, as measured by PCH length,
from exponential to linear.
This reduction stems from the fact that after the addition of the sorting
network, the constraint that bounds the optimization value can be stated in
terms of the output atoms of the network.
The benefit of this is that, in the context of that network, there is only a
linear number of logically distinct nogoods over its output atoms.
Therefore, any propagator for the constraint may only produce up to a
linear number of nogoods.

\begin{proposition}
\label{prop:opposing-sn}
Let $n$, $k$, and $d$ be non-negative integers
such that $k \le n$,
$N$ a sorting network of width $n$ and depth $d$,
$\aspify{N}$ the translation of $N$ into an answer set program,
$\binomproglt{n}{k}$ a binomial program
on atoms $\rg{\pred{x}{10}}{,}{\pred{x}{n0}}$,
$\pi$ a propagator for the constraint
\begin{align*}
  \sel{
    \sel{\assigntrue{x}}{x \in X}
  }{
    X \subseteq \eset{\pred{x}{1d}}{\pred{x}{nd}},
    |X| = k
  },
\end{align*}
and
let $\rg{A_1}{,}{A_m}$ be a complete PCH for $\aspify{N} \union \binomproglt{n}{k}$ and
$\pi$.
Then $m \le n - k + 1$.
\end{proposition}

\begin{proof}
Let $n$, $k$, $d$, $N$, $\aspify{N}$, $\binomproglt{n}{k}$, $\pi$, and
$\rg{A_1}{,}{A_m}$ be as above,
and define
$
\Lambda
=
\nogoodsforsupp{\aspify{N} \union \set{\eset{\pred{x}{10}}{\pred{x}{n0}}}}
$.
By Lemma~\ref{lemma:opposing-nogood-equivalence},
for each $1 \le i \le m$,
we have
$
\Lambda \union \set{\pi(A_i)}
\classicequiv
\Lambda \union \set{\set{\sigma_i}}
$
where $\sigma_i$ is the signed literal $\assigntrue{\pred{x}{jd}}$
with $j = \min \sel{k}{\pred{x}{kd} \in \pi(A_i)}$.
Also, let $\sim$ be the equivalence relation
that holds for nogoods $\sigma$ and $\sigma'$
if
$
\Delta \union \sigma
\classicequiv
\Delta \union \sigma'
$
where $\Delta = \nogoodsforsupp{\aspify{N} \union \binomproglt{n}{k}}$.
Observe that
$
\Delta =
\Lambda \union \delta
$
where $\delta$ is the constraint
$\IF \AGGCOUNT{\rg{\pred{x}{10}}{;}{\pred{x}{n0}}} < k$.
Because $\classicequiv$
is a congruence relation with respect to addition of constraints,
this implies $\pi(A_i) \sim \set{\sigma_i}$.
Based on the definitions of a propagator and a PCH, we can prove that the
relation $\sim$ holds for no pair of nogoods from $\range{\pi(A_1)}{\pi(A_m)}$.
From the transitivity of the equivalence relation $\sim$
it follows that no two of $\range{\sigma_1}{\sigma_m}$ are identical
and therefore $|\eset{\sigma_1}{\sigma_m}| = m$.
Given that $\sigma_i$ is the first signed literal in the nogood $\pi(A_i)$,
which  forbids a $k$-subset of the output atoms of $N$,
the signed literal $\sigma_i$ must be one of $\range{\assigntrue{\pred{x}{1d}}}{{\assigntrue{\pred{x}{(n - k + 1)(d)}}}}$.
Hence,
$\eset{\sigma_1}{\sigma_m} \subseteq \eset{\assigntrue{\pred{x}{1d}}}{\assigntrue{\pred{x}{(n - k + 1)(d)}}}$
and thus,
$m \le n - k + 1$.
\end{proof}

In order to study the impact of optimization rewriting on binomial optimization
programs $\binomprogopt{n}{k}$ in practice as well we ran experiments using
the preprocessing tool \system{pbtranslate}%
\footnote{
  Available at \url{https://github.com/jbomanson/pbtranslate}~.%
}
and the state-of-the-art ASP solver \system{clasp} (v.~3.3.3) \cite{gekakarosc15a}.
A part of the goal in these experiments is to investigate the difference between
an actual, off-the-shelf ASP solver and the simplified, abstract ASP solver
considered in our preceding analysis.
In particular, these experiments verify that improvements of high
magnitude as in the analysis can also be witnessed in practice.
To keep the results as relevant to practical ASP solving as possible, the solver
\system{clasp} was ran without disabling any of its sophisticated solving techniques.
Moreover, the entire optimization problem was solved, as opposed to only the final
unsatisfiability proofs that were considered in the analysis.
To keep the results consistent between runs and manageable to interpret, a single solving
configuration was fixed, namely ``tweety'', so that \system{clasp} would not
automatically pick different solving configurations between runs.

\begin{table}
  \caption{\label{table:opposing-conflicts}%
    Numbers of conflicts reported by \system{clasp} after solving binomial
    optimization programs $\binomprogopt{n}{\lfloor n/2 \rfloor}$ with varying
    numbers of atoms $n$.
  }
  \begin{tabular}{@{}l*{9}{@{\hspace{0.0pt}}r}@{}}
\toprule
$n$                                               & 5    & 6    & 7    & 8    & 9     & 10    & 15      & 20     & 25    \\
\midrule
\pipeline{norm+tweety}                            & 7    & 12   & 19   & 36   & 65    & 134   & 3.66k   & 248k   & 16.2M   \\
\pipeline{norm+rw+tweety}                         & 5    & 9    & 9    & 18   & 19    & 42    & 167     & 1.72k  & 23.6k   \\
\midrule
\pipeline{norm+tweety+usc}                        & 6    & 14   & 16   & 42   & 42    & 172   & 1.26k   & 11.3k  & 84.8k   \\
\pipeline{tweety}                                 & 4    & 10   & 15   & 35   & 56    & 126   & 3.21k   & 263k   & 17.2M   \\
\pipeline{rw+tweety}                              & 5    & 10   & 15   & 32   & 48    & 124   & 3.27k   & 234k   & 12.6M   \\
\pipeline{tweety+usc}                             & 5    & 14   & 18   & 53   & 80    & 197   & 4.81k   & 1.09M  & 60.5M   \\
\midrule
$\binom{n}{\lfloor n/2 \rfloor}$                  & 10   & 20   & 35   & 70   & 126   & 252   & 6.44k   & 185k   & 5.20M     \\
\bottomrule
\end{tabular}%

\end{table}

\begin{table}
  \caption{\label{table:opposing-cpu-time}%
    CPU times in seconds to complement the numbers of conflicts shown in
    Table~\ref{table:opposing-conflicts}.
  }
  \begin{tabular}{@{}l*{6}{@{\hspace{0.00pt}}r}@{}}
\toprule
$n$                                               & 8       & 9       & 10      & 15     & 20     & 25    \\
\midrule
\pipeline{norm+tweety}                            & 0.0159  & 0.0127  & 0.0154  & 0.048  & 4.09   & 928.7   \\
\pipeline{norm+rw+tweety}                         & 0.0151  & 0.0125  & 0.0125  & 0.0126 & 0.0352 & 0.513   \\
\midrule
\pipeline{norm+tweety+usc}                        & 0.0157  & 0.0137  & 0.0173  & 0.0235 & 0.19   & 1.29   \\
\pipeline{tweety}                                 & 0.0121  & 0.00908 & 0.00621 & 0.0205 & 2.87   & 219.8   \\
\pipeline{rw+tweety}                              & 0.0035  & 0.00586 & 0.00384 & 0.0377 & 4.42   & 790.0   \\
\pipeline{tweety+usc}                             & 0.00934 & 0.0103  & 0.0106  & 0.0529 & 14.1   & 2248.6   \\
\bottomrule
\end{tabular}%

\end{table}

The results are shown in Table~\ref{table:opposing-conflicts} in the form of
numbers of conflicts reported by \system{clasp} for increasing program size
parameters $n$.
These conflicts are of particular interest in relation to the preceding
analysis.
This is because the number of conflicts reported by the solver gives an upper
bound on the number of conflicts due to a nogood produced by a propagator for
the optimization statement.
That number, in turn, corresponds to the PCH length considered in the analysis.
Regarding the bound parameter $k$ of the binomial programs,
only the case
$k = \lfloor n/2 \rfloor$ was studied to simplify parameterization.
This choice of $k$
maximizes the number of optimal answer sets for any given
$n$.
That maximum number is given by the central binomial coefficient
$\binom{n}{\lfloor n/2 \rfloor}$.
These numbers are shown for reference in the
Table~\ref{table:opposing-conflicts}, since they also give the complete PCH
lengths predicted in Proposition~\ref{prop:opposing-card}.

The experiments were repeated with a number of solving pipelines, obtained by
composing different preprocessing and solving options into various combinations.
Initial preprocessing consisted either of sorting network based normalization of
the cardinality constraints in the instances (\pipeline{norm}) or of keeping
them as is so that \system{clasp} can handle them with its internal
propagators.
Optimization was implemented by default by \system{clasp} via the branch-and-bound strategy
and optionally via branch-and-bound after optimization rewriting (\pipeline{rw})
or via (unsatisfiable) \emph{core-guided} optimization (\pipeline{usc}).
This amounts to the $2 \times 3$ systems shown in the table, of which the
pipelines \pipeline{norm+tweety} and \pipeline{norm+rw+tweety} are the most relevant
to the preceding analysis.
In particular, pipeline \pipeline{norm+tweety}
is most representative of the setting
in Proposition~\ref{prop:opposing-card}
and pipeline \pipeline{norm+rw+tweety}
of Proposition~\ref{prop:opposing-sn}.
These pipelines are actual, complex analogues of the abstract, simplified
solving settings considered in the propositions.
Results for the remaining pipelines are provided for reference so that the significance
of the different components in the above pipelines can be evaluated in a useful context.
These reference pipelines contain the core-guided pipelines as well as
pipelines without normalization.
The reason for why we regard pipelines with normalization more relevant to the preceding
analysis is that normalization reduces the number of conflicts
due to cardinality constraints.
Therefore, the numbers of conflicts reported for the pipelines with normalization
are more closely reflective of the numbers of conflicts due to optimization statements,
although still not exact.

The results show that sorting network based optimization rewriting brings the
numbers of conflicts down to minuscule fractions of the original numbers
in pipelines that include normalization, i.e.,
in \pipeline{norm+rw+tweety} and \pipeline{norm+tweety}.
This improvement in conflicts is more significant than what is obtained
with core-guided optimization in \pipeline{norm+tweety+usc},
although both do yield improvements of comparable magnitude.
Regarding normalization, it improves each pipeline to which it is applied
and it is a strong factor in achieving best results in this comparison.
That is, it improves the performance of pipelines whether they use
branch-and-bound or core-guided optimization strategies and
whether or not they use optimization rewriting or not.
Moreover, the improvements due to optimization rewriting and normalization are of
similar magnitudes.
To see this,
one may consider the changes obtained when
adding optimization rewriting (\pipeline{rw}) to pipelines with or without normalization,
and then contrasting them with the changes obtained when
adding normalization (\pipeline{norm}) to pipelines with or without optimization rewriting.
Specifically, going from pipeline \pipeline{tweety} to \pipeline{rw+tweety}
yields a mild improvement,
and from pipeline \pipeline{norm+tweety} to \pipeline{norm+rw+tweety}
a huge improvement.
Likewise, going from pipeline \pipeline{tweety} to \pipeline{norm+tweety}
yields a mild improvement,
and from pipeline \pipeline{rw+tweety} to \pipeline{norm+rw+tweety}
a huge improvement.

Regarding the relation between this experimental evaluation and the preceding
formal analysis, both do highlight improvements due to optimization rewriting,
yet none of the statistics obtained in this evaluation precisely match the ones
predicted in the abstract formal analysis.
For example, Proposition~\ref{prop:opposing-card} predicts an exponential
number of propagator related conflicts to occur during an unsatisfiability
proof when no optimization rewriting is being used.
The precise predicted numbers are given by the central binomial coefficients
in Table~\ref{table:opposing-conflicts}.
However, this coefficient does not provide a consistent lower bound for any of
the pipelines.
Specifically, looking at the numbers of conflicts for pipeline
\pipeline{tweety}, which is the pipeline closest to the setting in the
proposition, the central binomial coefficient provides a lower bound for it
only starting at $n = 20$.
This is an indication that \system{clasp} internally improves upon the abstract
solver model we consider, and that these improvements make a difference at
least for modest program sizes $n < 20$.
On the other hand, Proposition~\ref{prop:opposing-sn} predicts at most a linear
number of propagator related conflicts to occur during an unsatisfiability
proof when optimization rewriting is being used.
That is, it predicts the optimization related propagator to produce an
insignificant number of conflicts during the final unsatisfiability proof.
Nevertheless, the numbers of conflicts for all of the tested systems run
into the thousands and higher when $n = 25$.
This has several potential reasons:
the experiments measure conflicts over the entire optimization process and not
only the final unsatisfiability proof, all conflicts are measured as opposed to
only propagator related conflicts, and that solving techniques such as nogood
deletion are used so that individual propagator related conflicts may occur more
than once.
A tighter comparison between the experiments and the analysis
could be obtained by extending the solver to
separately count the numbers of nogoods generated by different propagators.
Such a comparison is outside the scope of this experimental evaluation, however.

The CPU times required by these experiments are shown in
Table~\ref{table:opposing-cpu-time}.
In light of these CPU times, the picture is primarily similar as before:
pipeline \pipeline{norm+rw+tweety} is the overall winner
and together with \pipeline{norm+tweety+usc} they are in a class of their own
above the rest.
There is one main difference, however, which is that pipeline \pipeline{tweety}
fares relatively better than before, and at $n = 25$ it overtakes the pipelines
\pipeline{norm+tweety} and \pipeline{rw+tweety},
which add normalization or optimization rewriting only, respectively.
This is in line with the fact that both normalization and optimization rewriting
increase the program size, which generally increases the amount of work
per conflict done by the solver.

\subsection{Challenges Due to Heterogeneous Weights}
\label{section:challenges}

In this section, we describe challenges in optimization
rewriting that come with having heterogeneous and possibly large weights in
optimization functions.
This is to contrast with the positive formal results of
Section~\ref{section:formal} that concern optimization functions with only unit
weights or generally uniform weights.
The extent to which the benefits of sorting network based optimization rewriting
survive these challenges in practice is later studied experimentally in
Section~\ref{section:experimental}.

\newsavebox\nestedtikzpicturea
\begin{lrbox}{\nestedtikzpicturea}
  \examplesingletonnetwork{$\pred{x}{10}$,$\pred{x}{20}$}{$\pred{x}{11}$,$\pred{x}{21}$}
\end{lrbox}

\newsavebox\nestedtikzpictureb
\begin{lrbox}{\nestedtikzpictureb}
  \examplesingletonnetwork{\wt{40},\wt{70}}{\wt{0},\wt{0}}
\end{lrbox}

\newsavebox\nestedtikzpicturec
\begin{lrbox}{\nestedtikzpicturec}
  \examplesingletonnetwork{\wt{0},\wt{30}}{\wt{40},\wt{40}}
\end{lrbox}

\tikzset{snpicturestyle/.style={snpicturestylelarge}}
Weight propagation makes it harder to identify certain opportunities
for \emph{inference}.
For illustration, suppose we have a branch-and-bound solver that has already
found an answer set of value $70$ for an optimization program $\pair{P}{e}$
with the objective function $e = 40\pred{x}{10} + 70\pred{x}{20}$.
As per the discussion on branch-and-bound solving in
Section~\ref{section:formal}, the solver
will search for more optimal answer sets
by enforcing the upper bound $e < 70$.
From this point onward, it is reasonable to expect the solver to infer the atom
$\pred{x}{20}$ to be false given that its weight alone surpasses this bound.
However, this immediate inference becomes less immediately obvious once
optimization rewriting is applied based on sorting networks and weight
propagation.
To see this, consider rewriting the objective function $e$ using  a network
$N$ with a single comparator.
The variables and rules of the ASP translation $\aspify{N}$ of $N$ are:
\begin{center}
  \begin{tikzpicture}[xscale=7.5, yscale=-4.5]
    \node[anchor=north west] at (0.0, 0.0) {
      \usebox\nestedtikzpicturea
    };
    \node[anchor=north west] at (0.3, 0.0-0.02) {
        $\begin{array}{l}
            \pred{x}{11}\IF\pred{x}{10}\AND\pred{x}{20}\END \\
            \pred{x}{21}\IF\pred{x}{10}\END \\
            \pred{x}{21}\IF\pred{x}{20}\END
        \end{array}$
    };
  \end{tikzpicture}
\end{center}
In this case, any non-trivial
weight propagation turns the objective function $e$ into $e'$.
These expressions are shown
below as weights over $N$ and as equations:
\begin{center}
  \begin{tikzpicture}[xscale=7.5, yscale=-4.5]
    \node[anchor=north west] at (0.0, 0.0) {
      \usebox\nestedtikzpictureb
    };
    \node[anchor=north west] at (0.3, 0.0+0.10) {
      $e = 40 \pred{x}{10} + 70 \pred{x}{20}$
    };
    \node[anchor=north west] at (0.0, 0.33) {
      \usebox\nestedtikzpicturec
    };
    \node[anchor=north west] at (0.3, 0.33+0.10) {
      $e' = 30 \pred{x}{20} + 40 \pred{x}{11} + 40 \pred{x}{21}$
    };
  \end{tikzpicture}
\end{center}
After this rewriting, it is computationally straightforward again to infer that
$\pred{x}{20}$ is false.
However, the inference now requires the solver to either perform \emph{lookahead} based
on the encoding of the sorting network or to rely on a previously learned
nogood that captures the inference.
Lookahead is an inference technique in which an atom without a truth value is
temporarily and heuristically assigned one, and the logical
consequences of the assignment are explored via propagation.
In the case of $\pred{x}{20}$, if it is assigned true, then unit propagation
finds $\pred{x}{21}$ to be true as well by one of the rules in $\aspify{N}$.
As the total weight of $\pred{x}{20}$ and $\pred{x}{21}$ exceeds the upper bound,
$\pred{x}{20}$ can be inferred to be false.

In summary, and in the terminology of constraint programming and SAT, unit
propagation (UP) on rewritten optimization statements involving non-unit weights
does not maintain \emph{generalized arc consistency (GAC)}.
Regarding this terminology, given a constraint and an
assignment, GAC stands for the desirable condition that every assignment of an
individual atom that follows as their logical consequence is included in the
assignment, as in \cite{abniolrost13}.
Furthermore, for an encoding of a constraint to maintain GAC by UP, it is
required that repeated iteration of UP over the encoding always reaches a state
that satisfies GAC.
When an optimization statement is interpreted as a type of dynamic constraint,
and optimization rewriting is taken to produce an encoding of it, the above
discussed example indicates that there are inferences that are not captured by
UP after rewriting.
This is a drawback of the presented approach.
The significance of it is unclear, however.
Indeed, GAC has been routinely studied for SAT encodings of pseudo-Boolean
constraints and the studies have found both encodings that do and do not
maintain GAC to perform well in practice \cite{a_new_look_at_bdds,picat_2016}.
Hence, even though GAC is a positive feature, maximum pursuit of it has not
always proven fruitful, particularly when it has demanded larger encodings.
Nevertheless, the current lack of GAC-maintenance in optimization rewriting
leaves potential room for finding ways to recuperate the lost propagations and
to benefit even further from optimization rewriting in possible future work.

\subsection{Experimental Evaluation}
\label{section:experimental}

Next we continue the evaluation of the optimization rewriting approach from
Section~\ref{section:optimization-rewriting} by presenting extensive
experimental evaluations.
The approach is implemented in the tool \system{pbtranslate}%
\footnote{%
  Available at \url{https://github.com/jbomanson/pbtranslate} together
  with benchmarks at \url{https://research.ics.aalto.fi/software/asp/bench/}~.%
}
in the form of translations between answer set programs encoded in
\emph{ASP Intermediate Format} (aspif; Gebser et al.~2016b).
To evaluate the novel techniques presented in this paper, we composed solving
pipelines that preprocess with rewriting techniques and then search with the
state-of-the-art ASP solver \system{clasp} (v.~3.3.3) \cite{gekakarosc15a}.
These pipelines are contrasted with reference pipelines that
involve no rewriting.
The purpose of the experiments is to measure the general efficiency of the
approach as well as the impact of the types of used comparator networks and
weight propagation strategies.
A scheme of sorting networks $N$ of depth $O(\log^2 n)$ and size $O(n \log^2
n)$ is taken as a basis for the comparator networks in light of the formal
support for sorting networks established in Section~\ref{section:formal}.
The networks are constructed recursively from small fixed-size sorting networks
and Batcher's odd-even merge sorters \cite{Batcher68:afips}.
These networks are varied by creating copies $L_d$ limited to depths $d$.
This is motivated by the following factors.
For one, even the modest-appearing size growth $O(n \log^2 n)$ is large enough
to be problematic on many optimization statements considered in these
experiments.
Indeed, as discussed further later, even 80-fold instance-size blowup factors
are seen.
Second, even the depth-limited networks sort many input sequences and bring
other sequences closer to sorted states.
Hence, it is reasonable to expect depth-limited networks to retain a part of the
benefits of sorting networks while providing easily manageable rewriting sizes.

The pipelines generally operate as described below unless otherwise noted.
First, the pipelines perform optimization rewriting using sorting networks $N$.
The rewriting techniques rely on weight propagation based on decompositions
$D_1(N)$ that lead to maximally fine grained weight propagation.
Finally, \system{clasp} is ran with the branch and bound optimization strategy.
Based on these processing steps,
we formed individual solving pipelines of which
reference pipeline \pipeline{clasp} skips rewriting;
reference pipeline \pipeline{usc} skips rewriting and uses the
core-guided optimization strategy with disjoint core preprocessing;
pipeline \pipeline{F} includes rewriting;
pipelines \pipeline{$L_d$} rewrite using networks limited to a depth of $d$; and
pipelines \pipeline{$W_k$} rewrite using sparse decompositions $D_k(N)$;
pipeline \pipeline{$W_{\text{--}}$} rewrites using sorting networks without
weight propagation.
That is, pipeline \pipeline{$W_{\text{--}}$} extends the original program with
an encoding of a sorting network without altering the optimization statement.
For clarity, here and in the sequel, different fonts are used to distinguish
the system \system{clasp} and the pipeline \pipeline{clasp}.
The sparseness factor $k$ in the decomposition $D_k(N)$
controls the rough fraction $1 / k$ of nonzero weights produced by the weight
propagation function $\propfun{D_k(N)}$.
In these pipeline labels, an infinite subscript $\infty$
stands for a very large number which causes
pipeline \pipeline{$L_{\infty}$} to essentially apply no depth limit
and
pipeline \pipeline{$W_{\infty}$} to place weights on only the input layer
and the last layer.

For benchmarks, we picked a number of instance sets, each involving non-unit
weights.
Table~\ref{table:letter} includes the results for Bayesian Network Learning
\cite{Cussens11:uai,JSGM10:jmlr} with samples from three data sets.
In abstract terms, the task here is to construct an acyclic graph from certain
building blocks specified in the instance, and to optimize a sum of scores
associated with them.
Also included is Markov Network Learning \cite{jagerinypeco15a}, where the
purpose is to construct a \emph{chordal graph} under certain conditions while
again optimizing a sum of scores.
Moreover, in MaxSAT from the Sixth ASP Competition \cite{gemari15a}, the
Maximum Satisfiability problem is encoded in ASP and solved for a set of
industrial instances from the 2014 MaxSAT Evaluation \cite{maxsat_comp}.
Then there is Curriculum-based Course Timetabling
\cite{basotainsc13a,bocegasc12a}, where the goal is to assign resources in the
form of time slots and rooms to lectures while satisfying and minimizing
additional criteria.
Furthermore, Table~\ref{table:letter-continued} includes Fastfood and Traveling
Salesperson~(TSP) from the Second ASP Competition \cite{contest09a} with newly
generated instance sets that are harder and easier than in the competition,
respectively.
In Fastfood, the task is to essentially pick a subset of nodes from a one
dimensional line in order to minimize the sum of distances from each node to
the closest node in the subset.
In the well-known TSP problem, the task is to pick a subset of edges that
form a path and minimize the sum of weights associated with the chosen edges.

These benchmarks contain only optimization statements with truly heterogeneous
weights, as opposed to optimization statements with only one or a few distinct
weights.
The reason for focusing on heterogeneous weights is that,
as discussed in Section~\ref{section:challenges},
the case with non-uniform weights is particularly challenging.
Moreover, in the complementary case with few distinct weights,
weight propagation is straightforward
in the way that most if not all input weights are
simply moved directly onto output wires.
Weight propagation that moves all weights like this can be very effective.
Namely, with appropriate choices of low depth odd-even sorting networks,
this special case of weight propagation coincides to a large degree with
an optimization rewriting technique introduced
and experimentally evaluated
by Bomanson et al.~\citeyear{BGJ16:iclp} under the label ``64''.
The number 64 refers to a value for a parameter used therein.
The closest counterpart to this parameter
in the terminology and parameterization of this paper would
be the use of roughly depth-21 sorting networks
and maximally coarse grained weight propagation.
The relevant results therein are already strongly positive.
In view that, the challenge posed by uniform weights has been addressed to a
larger extent than the case of non-uniform weights, which therefore remains as
a further, greater challenge that is focused on here.

\newcommand{\lettertablea}{%
  \begin{oldtabular}{@{}r*{9}{@{\hspace{1.00pt}}c}@{}}
  \multicolumn{10}{@{}l}{\letterbenchmarkname{Bayes Alarm}}                                                                                                                                                                                                                                                                                                                                                                                                                                                                                                                                                                             \\
                                                               & \letternum{4}                & \letternum{5}           & \letternum{1}           & \letternum{1}           & \letternum{1}           & \letternum{2}           & \letternum{23}          & \lettercons & \lettermancoosi                                                                                                                                                                                                                                                                                                                                               \\
  \cmidrule{2-10}
  \lettersystem{$\text{clasp\hspace{1.00pt}usc~1}_{\zzz\zzz}$} & \lettercell{O}               & \lettercell{S}          & \lettercell{S}          & \lettercell{S}          & \lettercell{S}          & \lettercell{S}          & \lettercell{S}          & \letterreal{4.4}                                                                                                                                                                                                                                                                                                                       &   22.4                             \\
  \lettersystem{$\text{clasp}/L_{0\zzz}$}                      & \lettercell{O}               & \lettercell{O}          & \lettercell{O}          & \lettercell{S}          & \lettercell{O}          & \lettercell{S}          & \lettercell{S}          & \letterreal{4.4}                                                                                                                                                                                                                                                                                                                       &   \textbf{87.5}                    \\
  \lettersystem{$L_{4\zzz}$}                                   & \lettercell{O}               & \lettercell{O}          & \lettercell{O}          & \lettercell{O}          & \lettercell{S}          & \lettercell{S}          & \lettercell{S}          & \letterreal{4.8}                                                                                                                                                                                                                                                                                                                       &   77.2                             \\
  \lettersystem{$L_{8\zzz}$}                                   & \lettercell{\textbf{O}}      & \lettercell{\textbf{O}} & \lettercell{\textbf{O}} & \lettercell{\textbf{O}} & \lettercell{\textbf{O}} & \lettercell{\textbf{O}} & \lettercell{S}          & \letterreal{4.9}                                                                                                                                                                                                                                                                                                                       &   80.1                             \\
  \lettersystem{$L_{16}$}                                      & \lettercell{\textbf{O}}      & \lettercell{\textbf{O}} & \lettercell{\textbf{O}} & \lettercell{\textbf{O}} & \lettercell{\textbf{O}} & \lettercell{\textbf{O}} & \lettercell{S}          & \letterreal{5.1}                                                                                                                                                                                                                                                                                                                       &   81.7                             \\
  \lettersystem{$L_{32}$}                                      & \lettercell{\textbf{O}}      & \lettercell{\textbf{O}} & \lettercell{\textbf{O}} & \lettercell{\textbf{O}} & \lettercell{\textbf{O}} & \lettercell{\textbf{O}} & \lettercell{S}          & \letterreal{5.4}                                                                                                                                                                                                                                                                                                                       &   84.7                             \\
  \lettersystem{$\text{F}/L_{\infty}/W_{1\zzz}$}               & \lettercell{\textbf{O}}      & \lettercell{\textbf{O}} & \lettercell{\textbf{O}} & \lettercell{\textbf{O}} & \lettercell{\textbf{O}} & \lettercell{\textbf{O}} & \lettercell{S}          & \letterreal{5.8}                                                                                                                                                                                                                                                                                                                       &   66.4                             \\
  \lettersystem{$W_{4\zzz}$}                                   & \lettercell{\textbf{O}}      & \lettercell{\textbf{O}} & \lettercell{\textbf{O}} & \lettercell{\textbf{O}} & \lettercell{\textbf{O}} & \lettercell{\textbf{O}} & \lettercell{S}          & \letterreal{5.8}                                                                                                                                                                                                                                                                                                                       &   56.6                             \\
  \lettersystem{$W_{8\zzz}$}                                   & \lettercell{\textbf{O}}      & \lettercell{\textbf{O}} & \lettercell{\textbf{O}} & \lettercell{\textbf{O}} & \lettercell{\textbf{O}} & \lettercell{\textbf{O}} & \lettercell{S}          & \letterreal{5.8}                                                                                                                                                                                                                                                                                                                       &   64.1                             \\
  \lettersystem{$W_{16}$}                                      & \lettercell{O}               & \lettercell{O}          & \lettercell{O}          & \lettercell{O}          & \lettercell{S}          & \lettercell{O}          & \lettercell{S}          & \letterreal{5.8}                                                                                                                                                                                                                                                                                                                       &   59.7                             \\
  \lettersystem{$W_{32}$}                                      & \lettercell{O}               & \lettercell{O}          & \lettercell{S}          & \lettercell{S}          & \lettercell{S}          & \lettercell{S}          & \lettercell{S}          & \letterreal{5.8}                                                                                                                                                                                                                                                                                                                       &   44.2                             \\
  \lettersystem{$W_{\infty}$}                                  & \lettercell{O}               & \lettercell{O}          & \lettercell{S}          & \lettercell{S}          & \lettercell{S}          & \lettercell{S}          & \lettercell{S}          & \letterreal{5.8}                                                                                                                                                                                                                                                                                                                       &   44.8                             \\
  \lettersystem{$W_{\text{--}\zzz}$}                           & \lettercell{O}               & \lettercell{O}          & \lettercell{S}          & \lettercell{S}          & \lettercell{S}          & \lettercell{S}          & \lettercell{S}          & \letterreal{5.8}                                                                                                                                                                                                                                                                                                                       &   59.8                             \\
  \end{oldtabular}%
}
\newcommand{\lettertableb}{%
  \begin{oldtabular}{@{}r*{20}{@{\hspace{1.00pt}}c}@{}}
  \multicolumn{21}{@{}l}{\letterbenchmarkname{Bayes Hailfinder}}                                                                                                                                                                                                                                                                                                                                                                                                                                                                                                                                                                        \\
                                                               & \letternum{4}                & \letternum{12}          & \letternum{1}           & \letternum{2}           & \letternum{1}           & \letternum{1}           & \letternum{3}           & \letternum{1}           & \letternum{1}           & \letternum{1}           & \letternum{2}           & \letternum{2}           & \letternum{2}           & \letternum{1}           & \letternum{2}           & \letternum{1}           & \letternum{1}           & \letternum{18}        & \lettercons & \lettermancoosi                                                   \\
  \cmidrule{2-21}
  \lettersystem{$\text{clasp\hspace{1.00pt}usc~1}_{\zzz\zzz}$} & \lettercell{O}               & \lettercell{S}          & \lettercell{S}          & \lettercell{S}          & \lettercell{S}          & \lettercell{S}          & \lettercell{S}          & \lettercell{S}          & \lettercell{S}          & \lettercell{S}          & \lettercell{S}          & \lettercell{S}          & \lettercell{S}          & \lettercell{S}          & \lettercell{S}          & \lettercell{S}          & \lettercell{S}          & \lettercell{S}        & \letterreal{4.1}                          &   21.6                              \\
  \lettersystem{$\text{clasp}/L_{0\zzz}$}                      & \lettercell{O}               & \lettercell{O}          & \lettercell{O}          & \lettercell{O}          & \lettercell{O}          & \lettercell{O}          & \lettercell{O}          & \lettercell{O}          & \lettercell{O}          & \lettercell{O}          & \lettercell{O}          & \lettercell{O}          & \lettercell{O}          & \lettercell{O}          & \lettercell{O}          & \lettercell{S}          & \lettercell{S}          & \lettercell{S}        & \letterreal{4.1}                          &   93.0                              \\
  \lettersystem{$L_{4\zzz}$}                                   & \lettercell{O}               & \lettercell{O}          & \lettercell{O}          & \lettercell{O}          & \lettercell{O}          & \lettercell{O}          & \lettercell{O}          & \lettercell{O}          & \lettercell{O}          & \lettercell{O}          & \lettercell{O}          & \lettercell{O}          & \lettercell{O}          & \lettercell{O}          & \lettercell{O}          & \lettercell{O}          & \lettercell{S}          & \lettercell{S}        & \letterreal{4.6}                          &   \textbf{94.6}                     \\
  \lettersystem{$L_{8\zzz}$}                                   & \lettercell{O}               & \lettercell{O}          & \lettercell{O}          & \lettercell{O}          & \lettercell{O}          & \lettercell{O}          & \lettercell{O}          & \lettercell{O}          & \lettercell{O}          & \lettercell{O}          & \lettercell{O}          & \lettercell{O}          & \lettercell{O}          & \lettercell{O}          & \lettercell{O}          & \lettercell{O}          & \lettercell{S}          & \lettercell{S}        & \letterreal{4.8}                          &   91.8                              \\
  \lettersystem{$L_{16}$}                                      & \lettercell{O}               & \lettercell{O}          & \lettercell{O}          & \lettercell{O}          & \lettercell{O}          & \lettercell{O}          & \lettercell{O}          & \lettercell{O}          & \lettercell{O}          & \lettercell{O}          & \lettercell{O}          & \lettercell{O}          & \lettercell{O}          & \lettercell{O}          & \lettercell{O}          & \lettercell{O}          & \lettercell{S}          & \lettercell{S}        & \letterreal{5.0}                          &   89.7                              \\
  \lettersystem{$L_{32}$}                                      & \lettercell{\textbf{O}}      & \lettercell{\textbf{O}} & \lettercell{\textbf{O}} & \lettercell{\textbf{O}} & \lettercell{\textbf{O}} & \lettercell{\textbf{O}} & \lettercell{\textbf{O}} & \lettercell{\textbf{O}} & \lettercell{\textbf{O}} & \lettercell{\textbf{O}} & \lettercell{\textbf{O}} & \lettercell{\textbf{O}} & \lettercell{\textbf{O}} & \lettercell{\textbf{O}} & \lettercell{\textbf{O}} & \lettercell{\textbf{O}} & \lettercell{\textbf{O}} & \lettercell{S}        & \letterreal{5.2}                          &   92.2                              \\
  \lettersystem{$\text{F}/L_{\infty}/W_{1\zzz}$}               & \lettercell{O}               & \lettercell{O}          & \lettercell{O}          & \lettercell{O}          & \lettercell{O}          & \lettercell{O}          & \lettercell{O}          & \lettercell{O}          & \lettercell{O}          & \lettercell{O}          & \lettercell{O}          & \lettercell{O}          & \lettercell{S}          & \lettercell{S}          & \lettercell{S}          & \lettercell{S}          & \lettercell{S}          & \lettercell{S}        & \letterreal{5.6}                          &   68.5                              \\
  \lettersystem{$W_{4\zzz}$}                                   & \lettercell{O}               & \lettercell{O}          & \lettercell{O}          & \lettercell{O}          & \lettercell{O}          & \lettercell{S}          & \lettercell{O}          & \lettercell{O}          & \lettercell{O}          & \lettercell{O}          & \lettercell{O}          & \lettercell{S}          & \lettercell{S}          & \lettercell{O}          & \lettercell{S}          & \lettercell{S}          & \lettercell{S}          & \lettercell{S}        & \letterreal{5.6}                          &   65.6                              \\
  \lettersystem{$W_{8\zzz}$}                                   & \lettercell{O}               & \lettercell{O}          & \lettercell{O}          & \lettercell{O}          & \lettercell{O}          & \lettercell{O}          & \lettercell{O}          & \lettercell{O}          & \lettercell{S}          & \lettercell{O}          & \lettercell{S}          & \lettercell{S}          & \lettercell{S}          & \lettercell{S}          & \lettercell{S}          & \lettercell{S}          & \lettercell{S}          & \lettercell{S}        & \letterreal{5.6}                          &   60.6                              \\
  \lettersystem{$W_{16}$}                                      & \lettercell{O}               & \lettercell{O}          & \lettercell{O}          & \lettercell{O}          & \lettercell{S}          & \lettercell{O}          & \lettercell{O}          & \lettercell{S}          & \lettercell{O}          & \lettercell{S}          & \lettercell{S}          & \lettercell{S}          & \lettercell{S}          & \lettercell{S}          & \lettercell{S}          & \lettercell{S}          & \lettercell{S}          & \lettercell{S}        & \letterreal{5.6}                          &   57.5                              \\
  \lettersystem{$W_{32}$}                                      & \lettercell{O}               & \lettercell{O}          & \lettercell{O}          & \lettercell{S}          & \lettercell{O}          & \lettercell{S}          & \lettercell{S}          & \lettercell{S}          & \lettercell{O}          & \lettercell{S}          & \lettercell{S}          & \lettercell{S}          & \lettercell{S}          & \lettercell{S}          & \lettercell{S}          & \lettercell{S}          & \lettercell{S}          & \lettercell{S}        & \letterreal{5.6}                          &   47.6                              \\
  \lettersystem{$W_{\infty}$}                                  & \lettercell{O}               & \lettercell{O}          & \lettercell{S}          & \lettercell{O}          & \lettercell{O}          & \lettercell{O}          & \lettercell{S}          & \lettercell{O}          & \lettercell{S}          & \lettercell{S}          & \lettercell{S}          & \lettercell{S}          & \lettercell{S}          & \lettercell{S}          & \lettercell{S}          & \lettercell{S}          & \lettercell{S}          & \lettercell{S}        & \letterreal{5.6}                          &   60.6                              \\
  \lettersystem{$W_{\text{--}\zzz}$}                           & \lettercell{O}               & \lettercell{O}          & \lettercell{O}          & \lettercell{O}          & \lettercell{S}          & \lettercell{O}          & \lettercell{O}          & \lettercell{O}          & \lettercell{O}          & \lettercell{S}          & \lettercell{O}          & \lettercell{O}          & \lettercell{O}          & \lettercell{S}          & \lettercell{S}          & \lettercell{S}          & \lettercell{S}          & \lettercell{S}        & \letterreal{5.6}                          &   81.6                              \\
  \end{oldtabular}%
}
\newcommand{\lettertablec}{%
  \begin{oldtabular}{@{}r*{11}{@{\hspace{1.00pt}}c}@{}}
  \multicolumn{12}{@{}l}{\letterbenchmarkname{Bayes Water}}                                                                                                                                                                                                                                                                                                                                                                                                                                                                                                                                                                             \\
                                                               & \letternum{4}                & \letternum{12}          & \letternum{1}           & \letternum{2}           & \letternum{1}           & \letternum{2}           & \letternum{1}           & \letternum{1}           & \letternum{8}           & \lettercons & \lettermancoosi                                                                                                                                                                                                                                                                                           \\
  \cmidrule{2-12}
  \lettersystem{$\text{clasp\hspace{1.00pt}usc~1}_{\zzz\zzz}$} & \lettercell{O}               & \lettercell{S}          & \lettercell{S}          & \lettercell{S}          & \lettercell{S}          & \lettercell{S}          & \lettercell{S}          & \lettercell{S}          & \lettercell{S}          & \letterreal{2.9}                                                                                                                                                                                                                                                                  &   19.6                              \\
  \lettersystem{$\text{clasp}/L_{0\zzz}$}                      & \lettercell{O}               & \lettercell{O}          & \lettercell{S}          & \lettercell{S}          & \lettercell{S}          & \lettercell{S}          & \lettercell{S}          & \lettercell{S}          & \lettercell{S}          & \letterreal{2.9}                                                                                                                                                                                                                                                                  &   63.8                              \\
  \lettersystem{$L_{4\zzz}$}                                   & \lettercell{O}               & \lettercell{O}          & \lettercell{O}          & \lettercell{O}          & \lettercell{O}          & \lettercell{O}          & \lettercell{S}          & \lettercell{S}          & \lettercell{S}          & \letterreal{3.9}                                                                                                                                                                                                                                                                  &   83.9                              \\
  \lettersystem{$L_{8\zzz}$}                                   & \lettercell{O}               & \lettercell{O}          & \lettercell{O}          & \lettercell{O}          & \lettercell{O}          & \lettercell{O}          & \lettercell{O}          & \lettercell{S}          & \lettercell{S}          & \letterreal{4.1}                                                                                                                                                                                                                                                                  &   90.8                              \\
  \lettersystem{$L_{16}$}                                      & \lettercell{O}               & \lettercell{O}          & \lettercell{O}          & \lettercell{O}          & \lettercell{O}          & \lettercell{O}          & \lettercell{O}          & \lettercell{S}          & \lettercell{S}          & \letterreal{4.3}                                                                                                                                                                                                                                                                  &   93.1                              \\
  \lettersystem{$L_{32}$}                                      & \lettercell{\textbf{O}}      & \lettercell{\textbf{O}} & \lettercell{\textbf{O}} & \lettercell{\textbf{O}} & \lettercell{\textbf{O}} & \lettercell{\textbf{O}} & \lettercell{\textbf{O}} & \lettercell{\textbf{O}} & \lettercell{S}          & \letterreal{4.6}                                                                                                                                                                                                                                                                  &   \textbf{95.3}                     \\
  \lettersystem{$\text{F}/L_{\infty}/W_{1\zzz}$}               & \lettercell{O}               & \lettercell{O}          & \lettercell{O}          & \lettercell{O}          & \lettercell{O}          & \lettercell{O}          & \lettercell{O}          & \lettercell{S}          & \lettercell{S}          & \letterreal{4.8}                                                                                                                                                                                                                                                                  &   91.3                              \\
  \lettersystem{$W_{4\zzz}$}                                   & \lettercell{O}               & \lettercell{O}          & \lettercell{O}          & \lettercell{O}          & \lettercell{O}          & \lettercell{O}          & \lettercell{S}          & \lettercell{S}          & \lettercell{S}          & \letterreal{4.8}                                                                                                                                                                                                                                                                  &   87.9                              \\
  \lettersystem{$W_{8\zzz}$}                                   & \lettercell{O}               & \lettercell{O}          & \lettercell{O}          & \lettercell{O}          & \lettercell{O}          & \lettercell{O}          & \lettercell{O}          & \lettercell{S}          & \lettercell{S}          & \letterreal{4.8}                                                                                                                                                                                                                                                                  &   85.3                              \\
  \lettersystem{$W_{16}$}                                      & \lettercell{O}               & \lettercell{O}          & \lettercell{O}          & \lettercell{O}          & \lettercell{O}          & \lettercell{S}          & \lettercell{S}          & \lettercell{S}          & \lettercell{S}          & \letterreal{4.8}                                                                                                                                                                                                                                                                  &   75.7                              \\
  \lettersystem{$W_{32}$}                                      & \lettercell{O}               & \lettercell{O}          & \lettercell{O}          & \lettercell{O}          & \lettercell{S}          & \lettercell{S}          & \lettercell{S}          & \lettercell{S}          & \lettercell{S}          & \letterreal{4.8}                                                                                                                                                                                                                                                                  &   73.7                              \\
  \lettersystem{$W_{\infty}$}                                  & \lettercell{O}               & \lettercell{O}          & \lettercell{O}          & \lettercell{S}          & \lettercell{S}          & \lettercell{S}          & \lettercell{S}          & \lettercell{S}          & \lettercell{S}          & \letterreal{4.8}                                                                                                                                                                                                                                                                  &   68.1                              \\
  \lettersystem{$W_{\text{--}\zzz}$}                           & \lettercell{O}               & \lettercell{O}          & \lettercell{O}          & \lettercell{O}          & \lettercell{S}          & \lettercell{S}          & \lettercell{S}          & \lettercell{S}          & \lettercell{S}          & \letterreal{4.8}                                                                                                                                                                                                                                                                  &   75.0                              \\
  \end{oldtabular}%
}
\newcommand{\lettertabled}{%
  \begin{oldtabular}{@{}r*{10}{@{\hspace{1.00pt}}c}@{}}
  \multicolumn{11}{@{}l}{\letterbenchmarkname{Fastfood}}                                                                                                                                                                                                                                                                                                                                                                                                                                                                                                                                                                                \\
                                                               & \letternum{12}               & \letternum{1}           & \letternum{1}           & \letternum{4}           & \letternum{1}           & \letternum{3}           & \letternum{1}           & \letternum{7}           & \lettercons & \lettermancoosi                                                                                                                                                                                                                                                                                                                     \\
  \cmidrule{2-11}
  \lettersystem{$\text{clasp\hspace{1.00pt}usc~1}_{\zzz\zzz}$} & \lettercell{S}               & \lettercell{S}          & \lettercell{S}          & \lettercell{S}          & \lettercell{S}          & \lettercell{S}          & \lettercell{S}          & \lettercell{S}          & \letterreal{4.7}                                                                                                                                                                                                                                                                                            &  7.14                               \\
  \lettersystem{$\text{clasp}/L_{0\zzz}$}                      & \lettercell{O}               & \lettercell{O}          & \lettercell{O}          & \lettercell{S}          & \lettercell{S}          & \lettercell{S}          & \lettercell{S}          & \lettercell{S}          & \letterreal{4.7}                                                                                                                                                                                                                                                                                            &  76.0                               \\
  \lettersystem{$L_{4\zzz}$}                                   & \lettercell{O}               & \lettercell{O}          & \lettercell{O}          & \lettercell{S}          & \lettercell{S}          & \lettercell{S}          & \lettercell{S}          & \lettercell{S}          & \letterreal{5.0}                                                                                                                                                                                                                                                                                            &  71.0                               \\
  \lettersystem{$L_{8\zzz}$}                                   & \lettercell{O}               & \lettercell{O}          & \lettercell{O}          & \lettercell{O}          & \lettercell{S}          & \lettercell{S}          & \lettercell{S}          & \lettercell{S}          & \letterreal{5.1}                                                                                                                                                                                                                                                                                            &  82.4                               \\
  \lettersystem{$L_{16}$}                                      & \lettercell{\textbf{O}}      & \lettercell{\textbf{O}} & \lettercell{\textbf{O}} & \lettercell{\textbf{O}} & \lettercell{\textbf{O}} & \lettercell{\textbf{O}} & \lettercell{\textbf{O}} & \lettercell{S}          & \letterreal{5.2}                                                                                                                                                                                                                                                                                            &  95.2                               \\
  \lettersystem{$L_{32}$}                                      & \lettercell{\textbf{O}}      & \lettercell{\textbf{O}} & \lettercell{\textbf{O}} & \lettercell{\textbf{O}} & \lettercell{\textbf{O}} & \lettercell{\textbf{O}} & \lettercell{\textbf{O}} & \lettercell{S}          & \letterreal{5.5}                                                                                                                                                                                                                                                                                            &  96.0                               \\
  \lettersystem{$\text{F}/L_{\infty}/W_{1\zzz}$}               & \lettercell{O}               & \lettercell{O}          & \lettercell{O}          & \lettercell{O}          & \lettercell{O}          & \lettercell{O}          & \lettercell{S}          & \lettercell{S}          & \letterreal{5.8}                                                                                                                                                                                                                                                                                            &  \textbf{96.7}                      \\
  \lettersystem{$W_{4\zzz}$}                                   & \lettercell{O}               & \lettercell{O}          & \lettercell{O}          & \lettercell{O}          & \lettercell{O}          & \lettercell{S}          & \lettercell{S}          & \lettercell{S}          & \letterreal{5.8}                                                                                                                                                                                                                                                                                            &  91.2                               \\
  \lettersystem{$W_{8\zzz}$}                                   & \lettercell{O}               & \lettercell{O}          & \lettercell{O}          & \lettercell{O}          & \lettercell{S}          & \lettercell{S}          & \lettercell{S}          & \lettercell{S}          & \letterreal{5.8}                                                                                                                                                                                                                                                                                            &  85.5                               \\
  \lettersystem{$W_{16}$}                                      & \lettercell{O}               & \lettercell{O}          & \lettercell{S}          & \lettercell{S}          & \lettercell{S}          & \lettercell{S}          & \lettercell{S}          & \lettercell{S}          & \letterreal{5.8}                                                                                                                                                                                                                                                                                            &  62.6                               \\
  \lettersystem{$W_{32}$}                                      & \lettercell{O}               & \lettercell{S}          & \lettercell{S}          & \lettercell{S}          & \lettercell{S}          & \lettercell{S}          & \lettercell{S}          & \lettercell{S}          & \letterreal{5.8}                                                                                                                                                                                                                                                                                            &  51.7                               \\
  \lettersystem{$W_{\infty}$}                                  & \lettercell{O}               & \lettercell{S}          & \lettercell{S}          & \lettercell{S}          & \lettercell{S}          & \lettercell{S}          & \lettercell{S}          & \lettercell{S}          & \letterreal{5.8}                                                                                                                                                                                                                                                                                            &  52.1                               \\
  \lettersystem{$W_{\text{--}\zzz}$}                           & \lettercell{O}               & \lettercell{O}          & \lettercell{S}          & \lettercell{S}          & \lettercell{S}          & \lettercell{S}          & \lettercell{S}          & \lettercell{S}          & \letterreal{5.8}                                                                                                                                                                                                                                                                                            &  68.1                               \\
  \end{oldtabular}%
}
\newcommand{\lettertablee}{%
  \begin{oldtabular}{@{}r*{9}{@{\hspace{1.00pt}}c}@{}}
  \multicolumn{10}{@{}l}{\letterbenchmarkname{Markov Network}}                                                                                                                                                                                                                                                                                                                                                                                                                                                                                                                                                                          \\
                                                               & \letternum{13}               & \letternum{2}           & \letternum{1}           & \letternum{1}           & \letternum{2}           & \letternum{51}          & \letternum{2}           & \lettercons & \lettermancoosi                                                                                                                                                                                                                                                                                                                                               \\
  \cmidrule{2-10}
  \lettersystem{$\text{clasp\hspace{1.00pt}usc~1}_{\zzz\zzz}$} & \lettercell{S}               & \lettercell{S}          & \lettercell{S}          & \lettercell{S}          & \lettercell{S}          & \lettercell{S}          & \lettercell{T}          & \letterreal{4.4}                                                                                                                                                                                                                                                                                                                      &   6.94                              \\
  \lettersystem{$\text{clasp}/L_{0\zzz}$}                      & \lettercell{O}               & \lettercell{O}          & \lettercell{S}          & \lettercell{S}          & \lettercell{S}          & \lettercell{S}          & \lettercell{S}          & \letterreal{4.4}                                                                                                                                                                                                                                                                                                                      &   \textbf{81.9}                     \\
  \lettersystem{$L_{4\zzz}$}                                   & \lettercell{\textbf{O}}      & \lettercell{\textbf{O}} & \lettercell{\textbf{O}} & \lettercell{\textbf{O}} & \lettercell{\textbf{O}} & \lettercell{S}          & \lettercell{S}          & \letterreal{4.7}                                                                                                                                                                                                                                                                                                                      &   70.5                              \\
  \lettersystem{$L_{8\zzz}$}                                   & \lettercell{\textbf{O}}      & \lettercell{\textbf{O}} & \lettercell{\textbf{O}} & \lettercell{\textbf{O}} & \lettercell{\textbf{O}} & \lettercell{S}          & \lettercell{S}          & \letterreal{4.8}                                                                                                                                                                                                                                                                                                                      &   70.9                              \\
  \lettersystem{$L_{16}$}                                      & \lettercell{\textbf{O}}      & \lettercell{\textbf{O}} & \lettercell{\textbf{O}} & \lettercell{\textbf{O}} & \lettercell{\textbf{O}} & \lettercell{S}          & \lettercell{S}          & \letterreal{4.9}                                                                                                                                                                                                                                                                                                                      &   70.3                              \\
  \lettersystem{$L_{32}$}                                      & \lettercell{O}               & \lettercell{O}          & \lettercell{O}          & \lettercell{O}          & \lettercell{S}          & \lettercell{S}          & \lettercell{S}          & \letterreal{5.1}                                                                                                                                                                                                                                                                                                                      &   68.3                              \\
  \lettersystem{$\text{F}/L_{\infty}/W_{1\zzz}$}               & \lettercell{O}               & \lettercell{O}          & \lettercell{O}          & \lettercell{O}          & \lettercell{S}          & \lettercell{S}          & \lettercell{S}          & \letterreal{5.5}                                                                                                                                                                                                                                                                                                                      &   65.1                              \\
  \lettersystem{$W_{4\zzz}$}                                   & \lettercell{O}               & \lettercell{O}          & \lettercell{O}          & \lettercell{O}          & \lettercell{S}          & \lettercell{S}          & \lettercell{S}          & \letterreal{5.5}                                                                                                                                                                                                                                                                                                                      &   68.0                              \\
  \lettersystem{$W_{8\zzz}$}                                   & \lettercell{O}               & \lettercell{O}          & \lettercell{O}          & \lettercell{O}          & \lettercell{S}          & \lettercell{S}          & \lettercell{S}          & \letterreal{5.5}                                                                                                                                                                                                                                                                                                                      &   67.5                              \\
  \lettersystem{$W_{16}$}                                      & \lettercell{O}               & \lettercell{O}          & \lettercell{O}          & \lettercell{O}          & \lettercell{S}          & \lettercell{S}          & \lettercell{S}          & \letterreal{5.5}                                                                                                                                                                                                                                                                                                                      &   63.8                              \\
  \lettersystem{$W_{32}$}                                      & \lettercell{O}               & \lettercell{O}          & \lettercell{O}          & \lettercell{S}          & \lettercell{S}          & \lettercell{S}          & \lettercell{S}          & \letterreal{5.5}                                                                                                                                                                                                                                                                                                                      &   66.8                              \\
  \lettersystem{$W_{\infty}$}                                  & \lettercell{O}               & \lettercell{S}          & \lettercell{S}          & \lettercell{S}          & \lettercell{S}          & \lettercell{S}          & \lettercell{S}          & \letterreal{5.5}                                                                                                                                                                                                                                                                                                                      &   67.4                              \\
  \lettersystem{$W_{\text{--}\zzz}$}                           & \lettercell{O}               & \lettercell{S}          & \lettercell{S}          & \lettercell{S}          & \lettercell{S}          & \lettercell{S}          & \lettercell{S}          & \letterreal{5.5}                                                                                                                                                                                                                                                                                                                      &   65.3                              \\
  \end{oldtabular}%
}
\newcommand{\lettertablef}{%
  \begin{oldtabular}{@{}r*{12}{@{\hspace{1.00pt}}c}@{}}
  \multicolumn{13}{@{}l}{\letterbenchmarkname{MaxSAT}}                                                                                                                                                                                                                                                                                                                                                                                                                                                                                                                                                                                  \\
                                                               & \letternum{3}                & \letternum{1}           & \letternum{2}           & \letternum{1}           & \letternum{1}           & \letternum{1}           & \letternum{1}           & \letternum{5}           & \letternum{1}           & \letternum{4}           & \lettercons & \lettermancoosi                                                                                                                                                                                                                                                                 \\
  \cmidrule{2-13}
  \lettersystem{$\text{clasp\hspace{1.00pt}usc~1}_{\zzz\zzz}$} & \lettercell{\textbf{O}}      & \lettercell{S}          & \lettercell{\textbf{O}} & \lettercell{\textbf{O}} & \lettercell{\textbf{O}} & \lettercell{\textbf{O}} & \lettercell{\textbf{O}} & \lettercell{\textbf{O}} & \lettercell{\textbf{O}} & \lettercell{S}          & \letterreal{5.2}                                                                                                                                                                                                                                        &   \textbf{82.1}                     \\
  \lettersystem{$\text{clasp}/L_{0\zzz}$}                      & \lettercell{O}               & \lettercell{O}          & \lettercell{S}          & \lettercell{S}          & \lettercell{O}          & \lettercell{O}          & \lettercell{O}          & \lettercell{S}          & \lettercell{S}          & \lettercell{S}          & \letterreal{5.2}                                                                                                                                                                                                                                        &   73.2                              \\
  \lettersystem{$L_{4\zzz}$}                                   & \lettercell{O}               & \lettercell{O}          & \lettercell{O}          & \lettercell{S}          & \lettercell{O}          & \lettercell{O}          & \lettercell{O}          & \lettercell{S}          & \lettercell{S}          & \lettercell{S}          & \letterreal{5.3}                                                                                                                                                                                                                                        &   76.8                              \\
  \lettersystem{$L_{8\zzz}$}                                   & \lettercell{O}               & \lettercell{O}          & \lettercell{O}          & \lettercell{O}          & \lettercell{O}          & \lettercell{O}          & \lettercell{O}          & \lettercell{S}          & \lettercell{S}          & \lettercell{S}          & \letterreal{5.3}                                                                                                                                                                                                                                        &   76.4                              \\
  \lettersystem{$L_{16}$}                                      & \lettercell{O}               & \lettercell{O}          & \lettercell{O}          & \lettercell{O}          & \lettercell{O}          & \lettercell{O}          & \lettercell{O}          & \lettercell{S}          & \lettercell{S}          & \lettercell{S}          & \letterreal{5.4}                                                                                                                                                                                                                                        &   75.4                              \\
  \lettersystem{$L_{32}$}                                      & \lettercell{O}               & \lettercell{O}          & \lettercell{O}          & \lettercell{O}          & \lettercell{O}          & \lettercell{S}          & \lettercell{S}          & \lettercell{S}          & \lettercell{S}          & \lettercell{S}          & \letterreal{5.6}                                                                                                                                                                                                                                        &   59.6                              \\
  \lettersystem{$\text{F}/L_{\infty}/W_{1\zzz}$}               & \lettercell{O}               & \lettercell{O}          & \lettercell{O}          & \lettercell{O}          & \lettercell{S}          & \lettercell{S}          & \lettercell{S}          & \lettercell{S}          & \lettercell{S}          & \lettercell{S}          & \letterreal{5.9}                                                                                                                                                                                                                                        &   66.8                              \\
  \lettersystem{$W_{4\zzz}$}                                   & \lettercell{O}               & \lettercell{O}          & \lettercell{O}          & \lettercell{O}          & \lettercell{S}          & \lettercell{S}          & \lettercell{S}          & \lettercell{S}          & \lettercell{S}          & \lettercell{S}          & \letterreal{5.9}                                                                                                                                                                                                                                        &   68.2                              \\
  \lettersystem{$W_{8\zzz}$}                                   & \lettercell{O}               & \lettercell{O}          & \lettercell{O}          & \lettercell{O}          & \lettercell{S}          & \lettercell{S}          & \lettercell{S}          & \lettercell{S}          & \lettercell{S}          & \lettercell{S}          & \letterreal{5.9}                                                                                                                                                                                                                                        &   67.5                              \\
  \lettersystem{$W_{16}$}                                      & \lettercell{O}               & \lettercell{O}          & \lettercell{O}          & \lettercell{O}          & \lettercell{S}          & \lettercell{S}          & \lettercell{S}          & \lettercell{S}          & \lettercell{S}          & \lettercell{S}          & \letterreal{5.9}                                                                                                                                                                                                                                        &   64.3                              \\
  \lettersystem{$W_{32}$}                                      & \lettercell{O}               & \lettercell{O}          & \lettercell{O}          & \lettercell{O}          & \lettercell{S}          & \lettercell{S}          & \lettercell{S}          & \lettercell{S}          & \lettercell{S}          & \lettercell{S}          & \letterreal{5.9}                                                                                                                                                                                                                                        &   68.6                              \\
  \lettersystem{$W_{\infty}$}                                  & \lettercell{O}               & \lettercell{O}          & \lettercell{O}          & \lettercell{O}          & \lettercell{S}          & \lettercell{S}          & \lettercell{S}          & \lettercell{S}          & \lettercell{S}          & \lettercell{S}          & \letterreal{5.9}                                                                                                                                                                                                                                        &   71.8                              \\
  \lettersystem{$W_{\text{--}\zzz}$}                           & \lettercell{O}               & \lettercell{O}          & \lettercell{S}          & \lettercell{S}          & \lettercell{O}          & \lettercell{O}          & \lettercell{S}          & \lettercell{S}          & \lettercell{M}          & \lettercell{S}          & \letterreal{5.9}                                                                                                                                                                                                                                        &   56.8                              \\
  \end{oldtabular}%
}
\newcommand{\lettertableg}{%
  \begin{oldtabular}{@{}r*{20}{@{\hspace{1.00pt}}c}@{}}
  \multicolumn{21}{@{}l}{\letterbenchmarkname{Timetabling}}                                                                                                                                                                                                                                                                                                                                                                                                                                                                                                                                                                             \\
                                                               & \letternum{2}                & \letternum{1}           & \letternum{6}           & \letternum{4}           & \letternum{2}           & \letternum{4}           & \letternum{3}           & \letternum{2}           & \letternum{2}           & \letternum{1}           & \letternum{2}           & \letternum{7}           & \letternum{1}           & \letternum{11}          & \letternum{2}           & \letternum{1}           & \letternum{2}           & \letternum{4}         & \lettercons & \lettermancoosi                                                   \\
  \cmidrule{2-21}
  \lettersystem{$\text{clasp\hspace{1.00pt}usc~1}_{\zzz\zzz}$} & \lettercell{\textbf{O}}      & \lettercell{\textbf{O}} & \lettercell{\textbf{O}} & \lettercell{\textbf{O}} & \lettercell{\textbf{O}} & \lettercell{\textbf{O}} & \lettercell{\textbf{O}} & \lettercell{\textbf{O}} & \lettercell{\textbf{O}} & \lettercell{\textbf{O}} & \lettercell{\textbf{O}} & \lettercell{\textbf{O}} & \lettercell{\textbf{O}} & \lettercell{S}          & \lettercell{T}          & \lettercell{S}          & \lettercell{T}          & \lettercell{T}        & \letterreal{5.2}                          &   74.9                              \\
  \lettersystem{$\text{clasp}/L_{0\zzz}$}                      & \lettercell{O}               & \lettercell{O}          & \lettercell{O}          & \lettercell{O}          & \lettercell{O}          & \lettercell{S}          & \lettercell{O}          & \lettercell{O}          & \lettercell{S}          & \lettercell{S}          & \lettercell{S}          & \lettercell{S}          & \lettercell{S}          & \lettercell{S}          & \lettercell{S}          & \lettercell{S}          & \lettercell{S}          & \lettercell{S}        & \letterreal{5.2}                          &   80.2                              \\
  \lettersystem{$L_{4\zzz}$}                                   & \lettercell{O}               & \lettercell{O}          & \lettercell{O}          & \lettercell{O}          & \lettercell{O}          & \lettercell{O}          & \lettercell{O}          & \lettercell{S}          & \lettercell{S}          & \lettercell{O}          & \lettercell{S}          & \lettercell{S}          & \lettercell{S}          & \lettercell{S}          & \lettercell{S}          & \lettercell{S}          & \lettercell{S}          & \lettercell{M}        & \letterreal{5.5}                          &   70.7                              \\
  \lettersystem{$L_{8\zzz}$}                                   & \lettercell{O}               & \lettercell{O}          & \lettercell{O}          & \lettercell{O}          & \lettercell{O}          & \lettercell{O}          & \lettercell{O}          & \lettercell{O}          & \lettercell{O}          & \lettercell{O}          & \lettercell{S}          & \lettercell{S}          & \lettercell{S}          & \lettercell{S}          & \lettercell{S}          & \lettercell{S}          & \lettercell{S}          & \lettercell{M}        & \letterreal{5.6}                          &   78.1                              \\
  \lettersystem{$L_{16}$}                                      & \lettercell{O}               & \lettercell{O}          & \lettercell{O}          & \lettercell{O}          & \lettercell{O}          & \lettercell{O}          & \lettercell{O}          & \lettercell{O}          & \lettercell{O}          & \lettercell{O}          & \lettercell{O}          & \lettercell{S}          & \lettercell{S}          & \lettercell{S}          & \lettercell{S}          & \lettercell{S}          & \lettercell{S}          & \lettercell{M}        & \letterreal{5.8}                          &   \textbf{81.1}                     \\
  \lettersystem{$L_{32}$}                                      & \lettercell{O}               & \lettercell{O}          & \lettercell{O}          & \lettercell{O}          & \lettercell{S}          & \lettercell{O}          & \lettercell{S}          & \lettercell{S}          & \lettercell{O}          & \lettercell{S}          & \lettercell{O}          & \lettercell{S}          & \lettercell{S}          & \lettercell{S}          & \lettercell{S}          & \lettercell{S}          & \lettercell{S}          & \lettercell{M}        & \letterreal{6.0}                          &   74.2                              \\
  \lettersystem{$\text{F}/L_{\infty}/W_{1\zzz}$}               & \lettercell{O}               & \lettercell{O}          & \lettercell{S}          & \lettercell{S}          & \lettercell{S}          & \lettercell{S}          & \lettercell{M}          & \lettercell{S}          & \lettercell{S}          & \lettercell{S}          & \lettercell{S}          & \lettercell{S}          & \lettercell{M}          & \lettercell{S}          & \lettercell{S}          & \lettercell{M}          & \lettercell{M}          & \lettercell{M}        & \letterreal{6.6}                          &   25.8                              \\
  \lettersystem{$W_{4\zzz}$}                                   & \lettercell{O}               & \lettercell{O}          & \lettercell{S}          & \lettercell{S}          & \lettercell{S}          & \lettercell{S}          & \lettercell{M}          & \lettercell{S}          & \lettercell{S}          & \lettercell{S}          & \lettercell{S}          & \lettercell{S}          & \lettercell{M}          & \lettercell{S}          & \lettercell{S}          & \lettercell{M}          & \lettercell{M}          & \lettercell{M}        & \letterreal{6.6}                          &   24.3                              \\
  \lettersystem{$W_{8\zzz}$}                                   & \lettercell{O}               & \lettercell{S}          & \lettercell{S}          & \lettercell{S}          & \lettercell{S}          & \lettercell{S}          & \lettercell{M}          & \lettercell{S}          & \lettercell{S}          & \lettercell{S}          & \lettercell{S}          & \lettercell{S}          & \lettercell{M}          & \lettercell{S}          & \lettercell{S}          & \lettercell{M}          & \lettercell{M}          & \lettercell{M}        & \letterreal{6.6}                          &   24.2                              \\
  \lettersystem{$W_{16}$}                                      & \lettercell{O}               & \lettercell{S}          & \lettercell{S}          & \lettercell{S}          & \lettercell{S}          & \lettercell{S}          & \lettercell{M}          & \lettercell{S}          & \lettercell{S}          & \lettercell{S}          & \lettercell{S}          & \lettercell{S}          & \lettercell{M}          & \lettercell{S}          & \lettercell{S}          & \lettercell{M}          & \lettercell{M}          & \lettercell{M}        & \letterreal{6.6}                          &   19.8                              \\
  \lettersystem{$W_{32}$}                                      & \lettercell{O}               & \lettercell{S}          & \lettercell{S}          & \lettercell{S}          & \lettercell{S}          & \lettercell{S}          & \lettercell{M}          & \lettercell{S}          & \lettercell{S}          & \lettercell{S}          & \lettercell{S}          & \lettercell{S}          & \lettercell{M}          & \lettercell{S}          & \lettercell{S}          & \lettercell{M}          & \lettercell{M}          & \lettercell{M}        & \letterreal{6.6}                          &   20.8                              \\
  \lettersystem{$W_{\infty}$}                                  & \lettercell{O}               & \lettercell{O}          & \lettercell{S}          & \lettercell{S}          & \lettercell{S}          & \lettercell{S}          & \lettercell{M}          & \lettercell{S}          & \lettercell{S}          & \lettercell{S}          & \lettercell{S}          & \lettercell{S}          & \lettercell{M}          & \lettercell{S}          & \lettercell{S}          & \lettercell{M}          & \lettercell{M}          & \lettercell{M}        & \letterreal{6.6}                          &   34.8                              \\
  \lettersystem{$W_{\text{--}\zzz}$}                           & \lettercell{O}               & \lettercell{O}          & \lettercell{O}          & \lettercell{S}          & \lettercell{S}          & \lettercell{S}          & \lettercell{M}          & \lettercell{S}          & \lettercell{S}          & \lettercell{S}          & \lettercell{S}          & \lettercell{S}          & \lettercell{M}          & \lettercell{S}          & \lettercell{S}          & \lettercell{M}          & \lettercell{M}          & \lettercell{M}        & \letterreal{6.6}                          &   49.4                              \\
  \end{oldtabular}%
}
\newcommand{\lettertableh}{%
  \begin{oldtabular}{@{}r*{5}{@{\hspace{1.00pt}}c}@{}}
                                                               & \lettercell{Tsp2018range100} &                         &                         &                                                                                                                                                                                                                                                                                                                                                                                                                                                                                     \\
                                         & \letternum{57}        & \lettercons & \lettermancoosi                                                                                                                                                                                                                                                                                                                                                                                                                                                                                                                                        \\
  \cmidrule{2-6}
  \lettersystem{$\text{clasp\hspace{1.00pt}usc~1}_{\zzz\zzz}$} & \textbf{O}            & 3.2                                                                                                                                                                                                                                                                                                                                                                                                                                                                                                                                            \\
  $\text{clasp}/L_{0\zzz}$                   & \textbf{O}            & 3.2                                                                                                                                                                                                                                                                                                                                                                                                                                                                                                                                                              \\
  $L_{4\zzz}$                                & \textbf{O}            & 3.3                                                                                                                                                                                                                                                                                                                                                                                                                                                                                                                                                              \\
  $L_{8\zzz}$                                & \textbf{O}            & 3.4                                                                                                                                                                                                                                                                                                                                                                                                                                                                                                                                                              \\
  $L_{16}$                               & \textbf{O}            & 3.5                                                                                                                                                                                                                                                                                                                                                                                                                                                                                                                                                                  \\
  $L_{32}$                               & \textbf{O}            & 3.6                                                                                                                                                                                                                                                                                                                                                                                                                                                                                                                                                                  \\
  $\text{F}/L_{\infty}/W_{1\zzz}$            & \textbf{O}            & 3.6                                                                                                                                                                                                                                                                                                                                                                                                                                                                                                                                                              \\
  $W_{4\zzz}$                                & \textbf{O}            & 3.6                                                                                                                                                                                                                                                                                                                                                                                                                                                                                                                                                              \\
  $W_{8\zzz}$                                & \textbf{O}            & 3.6                                                                                                                                                                                                                                                                                                                                                                                                                                                                                                                                                              \\
  $W_{16}$                               & \textbf{O}            & 3.6                                                                                                                                                                                                                                                                                                                                                                                                                                                                                                                                                                  \\
  $W_{32}$                               & \textbf{O}            & 3.6                                                                                                                                                                                                                                                                                                                                                                                                                                                                                                                                                                  \\
  $W_{\infty}$                           & \textbf{O}            & 3.6                                                                                                                                                                                                                                                                                                                                                                                                                                                                                                                                                                  \\
  $W_{\text{--}\zzz}$                        & \textbf{O}            & 3.6                                                                                                                                                                                                                                                                                                                                                                                                                                                                                                                                                              \\
  \end{oldtabular}%
}
\newcommand{\lettertablei}{%
  \begin{oldtabular}{@{}r*{8}{@{\hspace{1.00pt}}c}@{}}
                                                               & \lettercell{Tsp2018range200} &                         &                         &                         &                         &                         &                                                                                                                                                                                                                                                                                                                                                                                                       \\
                                                               & \letternum{50}               & \letternum{2}           & \letternum{1}           & \letternum{2}           & \letternum{1}           & \letternum{1}           & \lettercons & \lettermancoosi                                                                                                                                                                                                                                                                                                                                                                         \\
  \cmidrule{2-9}
  \lettersystem{$\text{clasp\hspace{1.00pt}usc~1}_{\zzz\zzz}$} & \lettercell{S}               & \lettercell{S}          & \lettercell{S}          & \lettercell{S}          & \lettercell{S}          & \lettercell{S}          & \letterreal{3.2}                                                                                                                                                                                                                                                                                                                                                                                      \\
  \lettersystem{$\text{clasp}/L_{0\zzz}$}                      & \lettercell{O}               & \lettercell{O}          & \lettercell{O}          & \lettercell{S}          & \lettercell{S}          & \lettercell{S}          & \letterreal{3.2}                                                                                                                                                                                                                                                                                                                                                                                      \\
  \lettersystem{$L_{4\zzz}$}                                   & \lettercell{\textbf{O}}      & \lettercell{\textbf{O}} & \lettercell{\textbf{O}} & \lettercell{\textbf{O}} & \lettercell{\textbf{O}} & \lettercell{\textbf{O}} & \letterreal{3.4}                                                                                                                                                                                                                                                                                                                                                                                      \\
  \lettersystem{$L_{8\zzz}$}                                   & \lettercell{O}               & \lettercell{O}          & \lettercell{O}          & \lettercell{O}          & \lettercell{O}          & \lettercell{S}          & \letterreal{3.5}                                                                                                                                                                                                                                                                                                                                                                                      \\
  \lettersystem{$L_{16}$}                                      & \lettercell{\textbf{O}}      & \lettercell{\textbf{O}} & \lettercell{\textbf{O}} & \lettercell{\textbf{O}} & \lettercell{\textbf{O}} & \lettercell{\textbf{O}} & \letterreal{3.7}                                                                                                                                                                                                                                                                                                                                                                                      \\
  \lettersystem{$L_{32}$}                                      & \lettercell{\textbf{O}}      & \lettercell{\textbf{O}} & \lettercell{\textbf{O}} & \lettercell{\textbf{O}} & \lettercell{\textbf{O}} & \lettercell{\textbf{O}} & \letterreal{3.8}                                                                                                                                                                                                                                                                                                                                                                                      \\
  \lettersystem{$\text{F}/L_{\infty}/W_{1\zzz}$}               & \lettercell{\textbf{O}}      & \lettercell{\textbf{O}} & \lettercell{\textbf{O}} & \lettercell{\textbf{O}} & \lettercell{\textbf{O}} & \lettercell{\textbf{O}} & \letterreal{3.8}                                                                                                                                                                                                                                                                                                                                                                                      \\
  \lettersystem{$W_{4\zzz}$}                                   & \lettercell{O}               & \lettercell{O}          & \lettercell{O}          & \lettercell{O}          & \lettercell{O}          & \lettercell{S}          & \letterreal{3.8}                                                                                                                                                                                                                                                                                                                                                                                      \\
  \lettersystem{$W_{8\zzz}$}                                   & \lettercell{O}               & \lettercell{O}          & \lettercell{O}          & \lettercell{O}          & \lettercell{S}          & \lettercell{S}          & \letterreal{3.8}                                                                                                                                                                                                                                                                                                                                                                                      \\
  \lettersystem{$W_{16}$}                                      & \lettercell{O}               & \lettercell{O}          & \lettercell{S}          & \lettercell{S}          & \lettercell{S}          & \lettercell{S}          & \letterreal{3.8}                                                                                                                                                                                                                                                                                                                                                                                      \\
  \lettersystem{$W_{32}$}                                      & \lettercell{O}               & \lettercell{S}          & \lettercell{S}          & \lettercell{S}          & \lettercell{S}          & \lettercell{S}          & \letterreal{3.8}                                                                                                                                                                                                                                                                                                                                                                                      \\
  \lettersystem{$W_{\infty}$}                                  & \lettercell{O}               & \lettercell{S}          & \lettercell{S}          & \lettercell{S}          & \lettercell{S}          & \lettercell{S}          & \letterreal{3.8}                                                                                                                                                                                                                                                                                                                                                                                      \\
  \lettersystem{$W_{\text{--}\zzz}$}                           & \lettercell{O}               & \lettercell{O}          & \lettercell{O}          & \lettercell{S}          & \lettercell{S}          & \lettercell{S}          & \letterreal{3.8}                                                                                                                                                                                                                                                                                                                                                                                      \\
  \end{oldtabular}%
}
\newcommand{\lettertablej}{%
  \begin{oldtabular}{@{}r*{21}{@{\hspace{1.00pt}}c}@{}}
  \multicolumn{22}{@{}l}{\letterbenchmarkname{TSP}}                                                                                                                                                                                                                                                                                                                                                                                                                                                                                                                                                                                     \\
                                                               & \letternum{22}               & \letternum{1}           & \letternum{1}           & \letternum{1}           & \letternum{2}           & \letternum{1}           & \letternum{1}           & \letternum{1}           & \letternum{1}           & \letternum{1}           & \letternum{1}           & \letternum{3}           & \letternum{2}           & \letternum{1}           & \letternum{1}           & \letternum{1}           & \letternum{1}           & \letternum{1}         & \letternum{14}        & \lettercons & \lettermancoosi                           \\
  \cmidrule{2-22}
  \lettersystem{$\text{clasp\hspace{1.00pt}usc~1}_{\zzz\zzz}$} & \lettercell{S}               & \lettercell{S}          & \lettercell{S}          & \lettercell{S}          & \lettercell{S}          & \lettercell{S}          & \lettercell{S}          & \lettercell{S}          & \lettercell{S}          & \lettercell{S}          & \lettercell{S}          & \lettercell{S}          & \lettercell{S}          & \lettercell{S}          & \lettercell{S}          & \lettercell{S}          & \lettercell{S}          & \lettercell{S}        & \lettercell{S}        & \letterreal{3.2}  &   7.14                              \\
  \lettersystem{$\text{clasp}/L_{0\zzz}$}                      & \lettercell{O}               & \lettercell{S}          & \lettercell{O}          & \lettercell{O}          & \lettercell{O}          & \lettercell{O}          & \lettercell{O}          & \lettercell{O}          & \lettercell{S}          & \lettercell{S}          & \lettercell{S}          & \lettercell{S}          & \lettercell{S}          & \lettercell{S}          & \lettercell{S}          & \lettercell{S}          & \lettercell{S}          & \lettercell{S}        & \lettercell{S}        & \letterreal{3.2}  &   75.4                              \\
  \lettersystem{$L_{4\zzz}$}                                   & \lettercell{O}               & \lettercell{O}          & \lettercell{O}          & \lettercell{O}          & \lettercell{O}          & \lettercell{O}          & \lettercell{O}          & \lettercell{O}          & \lettercell{O}          & \lettercell{O}          & \lettercell{O}          & \lettercell{S}          & \lettercell{S}          & \lettercell{S}          & \lettercell{S}          & \lettercell{S}          & \lettercell{S}          & \lettercell{S}        & \lettercell{S}        & \letterreal{3.5}  &   81.3                              \\
  \lettersystem{$L_{8\zzz}$}                                   & \lettercell{O}               & \lettercell{O}          & \lettercell{O}          & \lettercell{O}          & \lettercell{O}          & \lettercell{O}          & \lettercell{O}          & \lettercell{O}          & \lettercell{O}          & \lettercell{O}          & \lettercell{O}          & \lettercell{O}          & \lettercell{S}          & \lettercell{O}          & \lettercell{S}          & \lettercell{S}          & \lettercell{S}          & \lettercell{S}        & \lettercell{S}        & \letterreal{3.6}  &   89.5                              \\
  \lettersystem{$L_{16}$}                                      & \lettercell{O}               & \lettercell{O}          & \lettercell{O}          & \lettercell{O}          & \lettercell{O}          & \lettercell{O}          & \lettercell{O}          & \lettercell{O}          & \lettercell{O}          & \lettercell{O}          & \lettercell{O}          & \lettercell{O}          & \lettercell{O}          & \lettercell{S}          & \lettercell{S}          & \lettercell{S}          & \lettercell{S}          & \lettercell{O}        & \lettercell{S}        & \letterreal{3.8}  &   87.7                              \\
  \lettersystem{$L_{32}$}                                      & \lettercell{\textbf{O}}      & \lettercell{\textbf{O}} & \lettercell{\textbf{O}} & \lettercell{\textbf{O}} & \lettercell{\textbf{O}} & \lettercell{\textbf{O}} & \lettercell{\textbf{O}} & \lettercell{\textbf{O}} & \lettercell{\textbf{O}} & \lettercell{\textbf{O}} & \lettercell{\textbf{O}} & \lettercell{\textbf{O}} & \lettercell{\textbf{O}} & \lettercell{\textbf{O}} & \lettercell{\textbf{O}} & \lettercell{\textbf{O}} & \lettercell{\textbf{O}} & \lettercell{S}        & \lettercell{S}        & \letterreal{4.0}  &   91.4                              \\
  \lettersystem{$\text{F}/L_{\infty}/W_{1\zzz}$}               & \lettercell{O}               & \lettercell{O}          & \lettercell{O}          & \lettercell{O}          & \lettercell{O}          & \lettercell{O}          & \lettercell{O}          & \lettercell{O}          & \lettercell{O}          & \lettercell{O}          & \lettercell{O}          & \lettercell{O}          & \lettercell{O}          & \lettercell{O}          & \lettercell{O}          & \lettercell{O}          & \lettercell{S}          & \lettercell{S}        & \lettercell{S}        & \letterreal{4.0}  &   \textbf{92.1}                     \\
  \lettersystem{$W_{4\zzz}$}                                   & \lettercell{O}               & \lettercell{O}          & \lettercell{O}          & \lettercell{O}          & \lettercell{O}          & \lettercell{S}          & \lettercell{O}          & \lettercell{O}          & \lettercell{O}          & \lettercell{O}          & \lettercell{S}          & \lettercell{S}          & \lettercell{S}          & \lettercell{S}          & \lettercell{O}          & \lettercell{S}          & \lettercell{S}          & \lettercell{S}        & \lettercell{S}        & \letterreal{4.0}  &   82.5                              \\
  \lettersystem{$W_{8\zzz}$}                                   & \lettercell{O}               & \lettercell{O}          & \lettercell{O}          & \lettercell{O}          & \lettercell{O}          & \lettercell{O}          & \lettercell{O}          & \lettercell{S}          & \lettercell{S}          & \lettercell{O}          & \lettercell{S}          & \lettercell{S}          & \lettercell{S}          & \lettercell{S}          & \lettercell{S}          & \lettercell{S}          & \lettercell{S}          & \lettercell{S}        & \lettercell{S}        & \letterreal{4.0}  &   71.4                              \\
  \lettersystem{$W_{16}$}                                      & \lettercell{O}               & \lettercell{O}          & \lettercell{O}          & \lettercell{S}          & \lettercell{S}          & \lettercell{S}          & \lettercell{S}          & \lettercell{S}          & \lettercell{O}          & \lettercell{S}          & \lettercell{O}          & \lettercell{S}          & \lettercell{S}          & \lettercell{S}          & \lettercell{S}          & \lettercell{S}          & \lettercell{S}          & \lettercell{S}        & \lettercell{S}        & \letterreal{4.0}  &   66.5                              \\
  \lettersystem{$W_{32}$}                                      & \lettercell{O}               & \lettercell{O}          & \lettercell{S}          & \lettercell{S}          & \lettercell{S}          & \lettercell{S}          & \lettercell{S}          & \lettercell{S}          & \lettercell{S}          & \lettercell{S}          & \lettercell{S}          & \lettercell{S}          & \lettercell{S}          & \lettercell{S}          & \lettercell{S}          & \lettercell{S}          & \lettercell{S}          & \lettercell{S}        & \lettercell{S}        & \letterreal{4.0}  &   63.2                              \\
  \lettersystem{$W_{\infty}$}                                  & \lettercell{O}               & \lettercell{O}          & \lettercell{S}          & \lettercell{O}          & \lettercell{S}          & \lettercell{S}          & \lettercell{S}          & \lettercell{S}          & \lettercell{S}          & \lettercell{S}          & \lettercell{S}          & \lettercell{S}          & \lettercell{S}          & \lettercell{S}          & \lettercell{S}          & \lettercell{S}          & \lettercell{S}          & \lettercell{S}        & \lettercell{S}        & \letterreal{4.0}  &   64.8                              \\
  \lettersystem{$W_{\text{--}\zzz}$}                           & \lettercell{O}               & \lettercell{O}          & \lettercell{O}          & \lettercell{O}          & \lettercell{O}          & \lettercell{O}          & \lettercell{S}          & \lettercell{O}          & \lettercell{O}          & \lettercell{S}          & \lettercell{S}          & \lettercell{S}          & \lettercell{S}          & \lettercell{S}          & \lettercell{S}          & \lettercell{S}          & \lettercell{S}          & \lettercell{S}        & \lettercell{S}        & \letterreal{4.0}  &   73.2                              \\
  \end{oldtabular}%
}

\newcommand{\letterbenchmarkname}[1]{\textbf{\footnotesize #1}\vspace{3pt}}
\newcommand{\letternum}[1]{\hspace{4pt}\clap{#1}\hspace{4pt}}
\newcommand{\lettercell}[1]{\clap{#1}}
\newcommand{\lettersystem}[1]{#1\hspace{20pt}}
\newcommand{\letterreal}[1]{\hspace{4pt}#1}
\newcommand{\lettercons}{\letternum{cons}}
\newcommand{\lettermancoosi}{\hspace{8pt}$S_1$\hspace{16pt}}
\newcommand{\zzz}{\phantom{0}}
\begin{table} 
  \caption{\label{table:letter}%
    Solving performance of \system{clasp} using core-guided optimization
    (\pipeline{usc}), branch-and-bound (\pipeline{clasp}), and branch-and-bound
    after optimization rewriting.
    The rows correspond to pipelines, the numbered columns $n$ to subsets of
    instances, and the letters O,S,M, and T to a classification of results.
    In addition, $S_1$ indicates solution quality scores computed with the same
    formula as in the Seventh ASP Competition as well as the Mancoosi
    International Solver Competition.
    Note however, that the $S_1$ score of a single solver depends on its
    performance relative to the other pipelines in the comparison, and
    therefore the scores here are not directly comparable to those in the
    competitions.
    Best pipelines per benchmark are highlighted in view of both
    classifications and $S_1$ scores.
    Moreover, the ``cons'' columns give the average base $10$ logarithms
    of the numbers of constraints remaining after rewriting.
    Rewriting is based on either
    (\pipeline{$L_{d}$}) depth $d$ comparator networks,
    (\pipeline{$F$}) full sorting networks,
    (\pipeline{$W_{k}$}) full sorting networks and decompositions $D_{k}$
    with sparseness factors $k$ that limit numbers of weights produced
    by weight propagation roughly to fractions $1 / k$,
    or (\pipeline{$W_{\text{--}}$}) full sorting networks without weight
    propagation.
    The pipelines \pipeline{clasp} and \pipeline{$L_{0}$} coincide, as do the pipelines \pipeline{$F$}, \pipeline{$L_{\infty}$},
    and \pipeline{$W_{1}$}.
  }%
  {
    \scriptsize
    \begin{tabular}{@{}lr@{}}
      \toprule
      \renewcommand{\lettersystem}[1]{\hspace{20pt}\llap{#1}}      
      \lettertablea
      &
      \renewcommand{\lettersystem}[1]{}               
      \lettertableb \\
      \midrule
      \renewcommand{\lettersystem}[1]{\hspace{20pt}\llap{#1}}      
      \lettertablec
      &
      \renewcommand{\lettersystem}[1]{}               
      \lettertablee \\
      \midrule
      \renewcommand{\lettersystem}[1]{\hspace{20pt}\llap{#1}}      
      \lettertablef
      &
      \renewcommand{\lettersystem}[1]{}               
      \lettertableg \\
      \bottomrule
    \end{tabular}%
  }%
\end{table}

\begin{table}
  \caption{\label{table:letter-continued}%
    Further results in the same form as in Table~\ref{table:letter}.
  }%
  {
    \scriptsize
    \begin{tabular}{@{}lr@{}}
      \toprule
      \renewcommand{\lettersystem}[1]{\llap{#1}}      
      \lettertabled
      &
      \renewcommand{\lettersystem}[1]{}               
      \lettertablej \\
      \bottomrule
    \end{tabular}%
  }%
\end{table}

Tables~\ref{table:letter} and \ref{table:letter-continued} show the results.
After running each \emph{pipeline-instance} combination with a 10 minute time
limit and a 3GB memory limit on Linux machines with Intel Xeon CPU E5-2680 v3
2.50GHz processors, each run was classified as
(O) finished and solved with a confirmed optimal solution,
(S) unfinished, but with some solutions found,
(T) unfinished without any solutions found in time, or
(M) aborted due to memory excess.
In the table, rows represent pipelines and columns represent disjoint sets of
instances with mutually identical run classifications.
The column numbers give the counts of instances in these sets.
Generally, the better a pipeline is, the higher is the sum of instance counts
related to its ``O'' entries.
Moreover, if a pipeline has an ``O'' entry in a column where another pipeline does
not, then it solves optimally at least some instances the other one does not.
If this holds mutually for a pair of rows, then the respective pipelines are
complementary in the sense that a \emph{virtual best solver} (VBS) combining
them would perform better than either one alone.
In order to complement this thorough view of the classification results, the
tables additionally show solution quality scores $S_1$ following a scheme
from the Mancoosi International Solver Competition\footnote{%
  http://www.mancoosi.org/misc/%
}
also used in the Sixth \cite{gemari15a} and Seventh ASP Competitions
\cite{seventh_asp_competition_design}.
The score for a pipeline $S$ among $M$ pipelines over a domain $D$ with $N$
benchmark instances is computed as $S_1 = \frac{100}{MN} \sum_{I \in D} M_S(I)$, where
$M_S(I)$ is 0, if $S$ did not find a single solution; or otherwise the number of
pipelines that found no solutions of higher quality, where a confirmed optimal
solution is preferred over an unconfirmed one.
The Seventh ASP Competition ranks solvers also based on an alternative score,
which awards points based on only the number of confirmed optimal solutions,
and the scores are scaled to a range of 0-100 per benchmark.
These scores are not presented in the tables, but the corresponding unscaled
scores can be found out by computing the weighted sums of ``O'' letters in each
row.
Hence, the winners by this alternative score coincide with the winners by the
``O'' letters that are highlighted.

The results for the pipelines
\pipeline{$L_d$}
illustrate the impact of tuning the depth limit $d$.
Of these pipelines,
\pipeline{$L_0$} and \pipeline{$L_{\infty}$} correspond to \pipeline{clasp} and
\pipeline{$F$}, which use no rewriting and full depth rewriting, respectively.
The results for these extremes are mixed, so that pipeline \pipeline{$F$}
improves performance over \pipeline{clasp} on some benchmarks and deteriorates
performance on others.
On the contrary, the intermediate depth limits $d = 8$ and $d = 16$ yield
robust performance.
Namely, both of the respective pipelines \pipeline{$L_8$} and
\pipeline{$L_{16}$} solve optimally all the instances that \pipeline{clasp} and
\pipeline{$F$} do, and more, and this holds over all the benchmarks.
The benefit of depth limits is strong enough so that on multiple benchmarks,
namely Bayes Hailfinder, MaxSAT, and Timetabling, rewriting with depth limits
as in \pipeline{$L_8$} and \pipeline{$L_{16}$} accelerates solving, even though
the use of full depth rewriting in pipeline \pipeline{$F$} decelerates it.
The benefit of depth limits appears strongest on Bayes Alarm, Bayes Hailfinder
and Timetabling.
The instances in these benchmarks are on the larger end among the considered
benchmarks in terms of optimization statement sizes after rewriting, shown in
Table~\ref{table:weight-count-blowup}.
Especially Timetabling stands out in this respect, and indeed the impact of
depth limits is greatest on it as well.
Overall, the use of these depth limits mitigates the size increase caused by
rewriting by up to an order of magnitude, as measured by logarithmic numbers of
constraints shown in the rightmost columns in the tables.
As regards pipelines with other depth limits, \pipeline{$L_{32}$} is sometimes
better and sometimes worse than the pipelines \pipeline{$L_d$} with lower depth
limits $d < 32$.

In light of these observations, the significance of depth limits is likely due
to their strong and direct connection to how much optimization rewriting
increases instance size.
Namely, network depth is a factor of network size
and therefore also of the size of the corresponding ASP translations.
On the other hand, the choice of weight propagation does not affect translation
size.
Regarding the magnitude of the depth factor, with full rewriting as in
\pipeline{$F$} that is based on odd-even merge sorters, the depth is
$O(\log^2 n)$ in the length $n$ of an optimization statement.
In practice, the implementation in \system{pbtranslate} employs some micro
optimizations and manages to produce networks
with depths in the range $\range{9}{105}$ for $n \in \eset{10}{10\,000}$.
These ranges are relevant since the size of the most substantial optimization
statements in the considered benchmarks are in the thousands, except on
Timetabling where they range from hundreds up to over a million, and on TSP
where they range between one and two hundred.
Therefore, even though the size of full depth optimization rewriting is
only polynomial, the increase is considerable in practice.
Indeed, the tables indicate instance size increases from $0.7$ to $1.9$ orders
of magnitude on different benchmarks with an average of around $1.2$.
The range corresponds to a $5$ to $80$ fold increase.
One may obtain these numbers by deducting the logarithmic-scale numbers of constraints
shown in the ``cons'' column
for pipeline \pipeline{$F$} from those for \pipeline{clasp}.
A more manageable size growth is achieved via
the use of depth limits $d$ in pipelines \pipeline{$L_d$},
which yield linearly growing network sizes.
This is reflected in the ``cons'' values of the respective rows, which lie in
between those for the extreme cases of \pipeline{clasp} and \pipeline{$F$}.
Based on the results, these more modest instance sizes appear to yield
generally fruitful tradeoffs between the benefits and costs of optimization
rewriting.

The results for the pipelines
\pipeline{$W_k$}
illustrate the impact of tuning weight propagation from fine to coarse grained
propagation.
At one end, \pipeline{$W_1$} propagates weights as much as possible, producing a
high number of generally low weights to be optimized.
At the other end, \pipeline{$W_{\infty}$} propagates weights as little as
possible, producing weights only on the first and last layers.
The results indicate a clear gradual trend in favor of fine grained propagation
among the different Bayes classes, Markov Network, Fastfood, and TSP classes.
Interestingly, these classes are separated from the remaining MaxSAT and
Timetabling classes in having more heterogeneous weights in their optimization
statements, which may have a connection with the trend.
The heterogeneity of weights can be quantified by measuring the proportional
increase in atoms being optimized caused by rewriting, shown in
Table~\ref{table:weight-count-blowup}, as this is an indication of how many
nonzero weights remain after weight propagation, which is strongly dependent on
the heterogeneity of weights.
The observed trend makes intuitive sense, since as the ratio approaches $1$,
the different pipelines $W_{k}$ converge, so any differences ought to manifest
with higher ratios.
Formally analyzing the differences in the performance impact of different weight
propagation methods on a finer level is challenging.
In contrast to the formal analysis carried out in
Section~\ref{section:formal}, such an investigation would be most meaningful in
the context of heterogeneous weights.
Moreover, the abstraction level would have to be detailed enough to capture
differences specifically due to the representation of the optimization
statement as opposed to the encoding of the network, as the latter is
independent of weight propagation.
Hence such an analysis would be involved, and given the lack of experimental
evidence in favor of sparseness factors other than $k = 1$, also weakly motivated at
present.
Thus it is left for potential future work.

\begin{table}[t]
  \caption{\label{table:weight-count-blowup}%
    Average numbers of atoms of atoms in optimization statements before and
    after full rewriting, and the average ratios of the latter over the former.
  }%
  {
    \small
    \begin{tabular}{lrrr}
      \toprule
                       & \multicolumn{1}{l}{atoms before}  & \multicolumn{1}{l}{atoms after}  & \multicolumn{1}{l}{ratio} \\
      \midrule
      Bayes Alarm      & 3,016.1                           & 20,783.1                         & 6.0                       \\
      Bayes Hailfinder & 2,102.6                           & 13,605.1                         & 6.1                       \\
      Bayes Water      & 470.2                             & 2,342.2                          & 4.5                       \\
      Markov Networks  & 1,559.9                           & 8,060.1                          & 4.6                       \\
      MaxSAT           & 3,531.1                           & 10,570.3                         & 2.3                       \\
      Timetabling      & 29,314.5                          & 123,178.1                        & 3.8                       \\
      Fastfood         & 3,147.3                           & 16,393.4                         & 5.2                       \\
      TSP              & 115.6                             & 499.6                            & 4.3                       \\
      \bottomrule
    \end{tabular}
  }
\end{table}

Furthermore, we note that the performance of core-guided optimization
is clearly different from the branch-and-bound based strategy used in the other
pipelines.
On the different Bayes classes, Markov Network, Fastfood, and TSP classes, core-guided optimization
falls behind.
On the remaining benchmarks, MaxSAT and Timetabling, the situation is different
and in fact the core-guided pipeline constitutes a VBS on its own if we overlook a
single MaxSAT instance.
However, on these classes, the most highly performing rewriting pipelines
improve on branch-and-bound
\pipeline{clasp} and almost halve the performance gap that separates it from
core-guided optimization.

Finally, the $S_1$ scores that measure solution quality in a more sensitive
manner give a similar picture overall with a few specific differences.
Namely, \pipeline{clasp} fares better in this light on some of the benchmarks.
This means that within the scope of those benchmarks, \pipeline{clasp} manages
to find higher quality solutions than the rewriting pipelines in cases where
neither reach an optimality proof.
Hence, it appears that some of the strength of the rewriting pipelines lies in
their ability to provide optimality proofs, which would be in accordance with
the formal analysis in Section~\ref{section:formal}, and that when those proofs
are nevertheless out of reach, this strength is weakened.
Moreover, in the results on Timetabling, the $S_1$ score penalizes core-guided
solving, as it runs into several timeouts on this class.

It is naturally difficult to predict the impact of optimization
rewriting based on the syntactic properties and structure of the instances and
their optimization statements when dealing with heterogeneous application
problems.
On the practical level, an investigation of such connections is best left for
dedicated work such as that behind the portfolio solver \system{me-asp}
\cite{me_asp_multi_level}.
\system{me-asp} is equipped with a number of black-box solvers and it uses
machine learning to decide which one to apply to each given instance
based on various problem features.
Nevertheless, some expectations for performance can be set here.
To this end, recall the setting from the formal analysis in
Section~\ref{section:formal} where
(i) there is a large number of non-optimal answer sets that need to be
rejected,
(ii) the rejection necessitates a large number of nogoods on the original
atoms, and
(iii) the rejection is possible with a small number of nogoods on the output
atoms of a sorting network.
Benefits similar to those witnessed in the analysis may occur if all three
of these items apply to a given problem.
Moreover, the increase in instance size must be reasonable so as not to
outweight these benefits.
Even for diverse application problems, Item (i) represents the typical case and
it indeed holds for all of the benchmarks considered here.
However, Items (ii-iii) are not as easily satisfied and may be hard to detect
from syntactic features.
Moreover, the number of nogoods required to reject non-optimal answer sets can
depend entirely on the encoding of a problem.
For example, it is possible to apply encoding techniques that are analogous to
optimization rewriting, in which case any benefit from subsequent optimization
rewriting is reduced and potentially even nullified.
To understand the practical scenarios where Items (ii-iii) might apply, let us
consider the task of formulating an optimization statement from two angles.
From a \emph{declarative angle}, the task is to declare the intent to optimize
some desired criteria and the formulation is successful if it correctly
declares ones intent.
From a \emph{number system angle}, the task is to design a weighted number
system in which to represent the optimization value by referring to the atoms
of the program.
In this case, a formulation is successful if it satisfies properties such as
lack of ambiguity and ease of comparability.
Regarding these properties, in case of ASP optimization, an ideal
representation enables to impose any single bound on the optimization statement
with a single nogood.
For example, the rewritten optimization statements analyzed in the formal
analysis in Section~\ref{section:formal} embody a unary number system that
enables this.
Moreover, a number system such as the binary number system would be
unambiguous, but would not allow the expression of bounds with a single nogood.
An encoding may be written from either of these two perspectives and the
declarative angle is likely leave more room for improvement due to optimization
rewriting since in contrast the number system angle is likely to lead to more
efficient encodings for optimization, and thus also to reduce the improvement
potential accessible via optimization rewriting.
In the considered benchmarks, the used ASP encodings fall more on the
declarative side, and in this respect, all of the benchmarks appear to be
potentially fruitful targets for optimization rewriting.
In summary, all of the considered benchmarks show basic promise for optimization
rewriting.
The remaining question is then whether that promise realizes in benefits that
outweight the increase in instance size.
Unfortunately, this is hard to predict based on the structure of the instances.

Figure~\ref{fig:performance} shows cactus plots of solving time over optimally
solved instances and optimization values over the same
benchmark classes considered earlier in Tables~\ref{table:letter} and
\ref{table:letter-continued}.
The included pipelines are the same as before with one added pipeline described
shortly, \pipeline{$L_{8} W_{\text{--}}$}.
The pipelines are grouped into three partially overlapping categories:
one with pipelines \pipeline{$L_{d}$},
one with pipelines \pipeline{$W_{k}$},
and one with selected few representatives of these and other pipelines.
The solving times, which are shown in the three plots on the left column, form a
picture that is in line with the results described previously.
In more detail, at the top left, pipelines \pipeline{$L_{8}$} and
\pipeline{$L_{16}$} emerge as the overall best performing ones among those
using depth limits.
In the middle left, pipelines using finer grained weight propagation consistently
outperform those using coarser grained propagation.
At the bottom left, rewriting pipelines with depth limits lead in overall
performance, followed by \system{clasp} with branch-and-bound optimization but
without rewriting (\pipeline{clasp}), and finally, by \system{clasp} using
core-guided optimization (\pipeline{usc}).
This bottom left plot also includes the one rewriting based pipeline not present
in
the prior results, \pipeline{$L_{8} W_{\text{--}}$}, which uses depth $8$
networks without weight propagation so that all weights are kept on
the input layer.
This pipeline \pipeline{$L_{8} W_{\text{--}}$} is included here in order to
gauge whether a combination of the techniques behind
$L_{d}$
and
$W_{\text{--}}$
improves upon the individual pipelines.
A single parameter value $d = 8$ was fixed for simplicity.
No such improvement is seen, however, as \pipeline{$L_{8} W_{\text{--}}$}
fares significantly worse than \pipeline{$L_8$}.

The respective plots of optimization values on the right have been normalized to
a range from 0 to 1 corresponding to the best and worst values achieved by the
pipelines.
Regarding differences between pipelines,
the pipelines \pipeline{$W_k$} perform similarly with different sparseness
factors $k$, except with $k = \infty$, which stands out and gives the best
results.
Moreover, as seen in the bottom right plot, \pipeline{clasp} without rewriting yields
overall lowest optimization values.
This reveals that although \pipeline{clasp} without rewriting does not achieve
the greatest numbers of optimally solved instances,
it is exceptionally often close in optimization values to
whichever pipeline solves each instance at hand most optimally.
In comparison, the best pipelines with rewriting solve more instance optimally,
but fall behind slightly more clearly on the other instances.

\begin{figure}
  \includegraphics{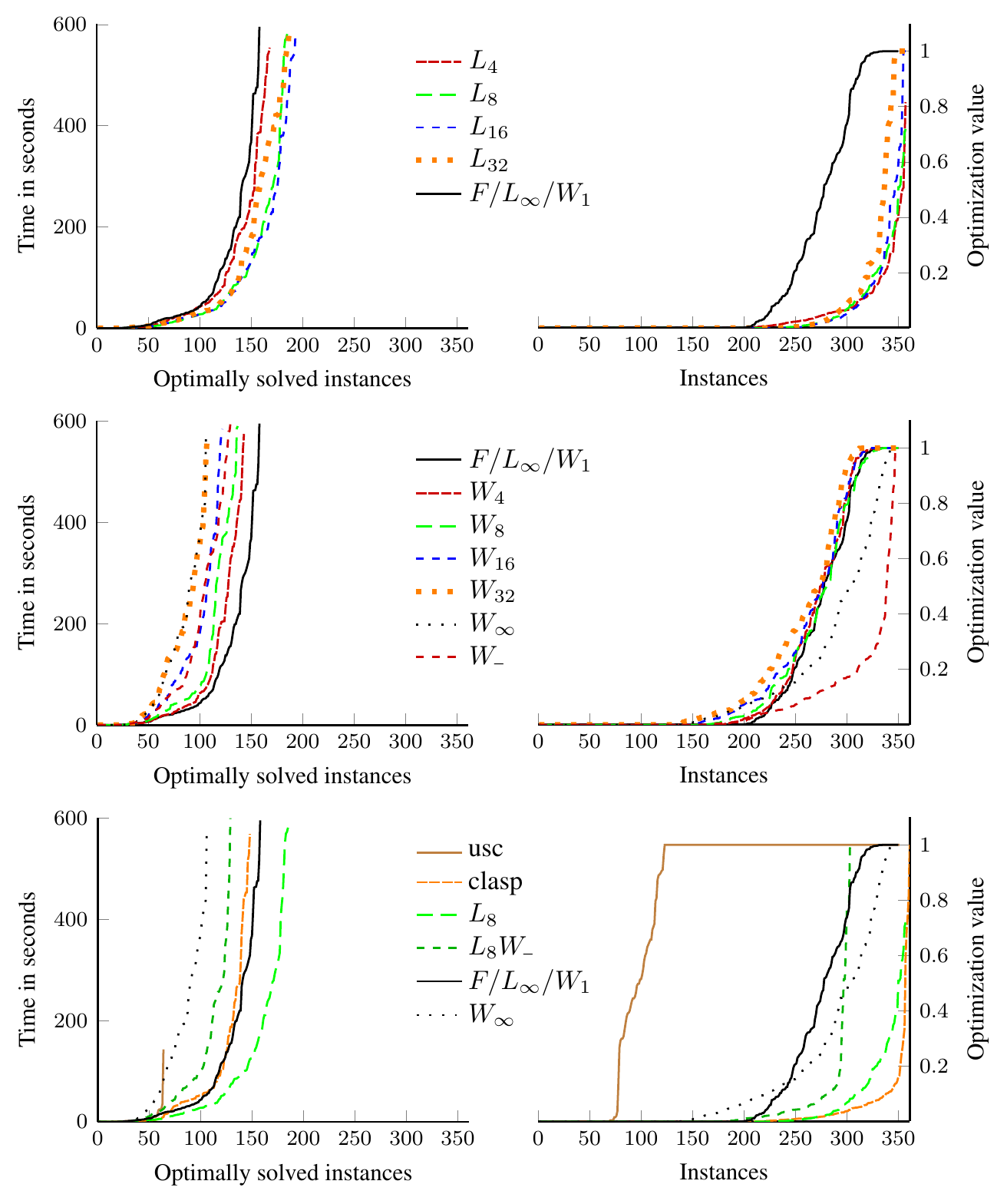}
  \caption{%
    Cactus plots of time
    and optimization value for runs
    in Tables~\ref{table:letter} and \ref{table:letter-continued}.
    Closeness to the bottom right is better in both cases.
    For each instance, the best final optimization value reported by a system is
    mapped to 0 and the worst to 1.
    From top to bottom,
    the horizontal pairs of plots concern
    systems \pipeline{$L_d$} using rewriting with depth $d$ networks,
    systems \pipeline{$W_k$} using rewriting with weights on every $k$th layer,
    and a selection of these systems together with \system{clasp}
    using core-guided optimization (usc) and branch-and-bound (clasp).%
    \label{fig:performance}}
\end{figure}

In the process of these experiments, additional preliminary screening
was performed on various aspects of the benchmark setting.
The results are only briefly described here due to lacking significance in the
outcomes.
For one, a significantly extended timeout of one hour had little impact on the
relative standing of the pipelines.
Likewise,
use of a stratification heuristic for handling weights in \system{clasp}
when employing core-guided optimization made little difference.
Moreover, the combination of rewriting and core-guided optimization performed
otherwise similarly to plain core-guided optimization, but resulted in more
timeouts.

Further benchmark classes were also screened for an impact due to optimization
rewriting.
To this end, Steiner Tree, Valves Location, and Video Streaming from the ASP
Competitions were considered, but none showed clear signs of improvements to
report.
These benchmarks are generally large, and have geometric averages 5.7, 5.5, and
4.5 of constraints on a logarithmic scale before rewriting, respectively.
On Valves Location, the results were neutral and on the other two, performance
was reduced.
On Valves Location, the explanation may be in that the optimization statements
are minuscule relative to the entire instances, which may reflect their
significance as sources of nogoods.
On Video Streaming, the performance reduction is minor.
On Steiner Tree, the reduction is greater, and the reason may simply be that
the size blowup is particularly costly given the exceptionally large size of
the instances already before rewriting.
In general, an exceptionally large input size for a solvable instance is a sign
of the instance being ``large but easy'' rather than ``small but hard'', and it
is likely that its difficulty does not stem as heavily from a combinatorial
explosion as with smaller instances.
This may hinder the usefulness of the proposed optimization rewriting since
based on the formal analysis, optimization rewriting yields benefits in cases
where the task of optimization presents a clear combinatorial challenge.

In summary, our experiments provide support in favor of different solving pipelines on
different benchmark classes. Thus, an ideal portfolio solver incorporating the
proposed optimization rewriting techniques and native optimization approaches
would surpass either approach standing alone. In particular, our optimization
rewriting pipelines improve in several cases over plain \system{clasp} as well as \system{clasp}
with core-guided optimization. Based on the results, we recommend
considering optimization rewriting based on depth limited comparator networks as
in \pipeline{$L_{d}$} where $d$ is ideally tuned based on the benchmark set at
hand or fixed to a modest value such as $d = 8$, which yielded robust
performance in these experiments.
Limiting network depth limits rewriting size while compromising some of the
benefits of rewriting and generally, although not always, this tradeoff is
worthwhile.
Moreover, due to differences in the result statuses and the $S_1$ scores, we
hypothesize that the proposed optimization rewriting techniques are
particularly suitable for solving instances all the way to optimality.

\section{Related Work}
\label{section:related}

The odd-even merge sort scheme \cite{Batcher68:afips} is a widely used scheme
for generating practical and useful sorting networks.
Sorting networks have existing applications in, e.g., SAT encodings of
pseudo-Boolean constraints \cite{ES06:jsat}.
These SAT encodings simulate summation in a binary system with digits encoded in
a unary number system.
The number of sorting networks used is proportional to the maximum bit width of
the weights in the pseudo-Boolean constraint being encoded.
Moreover, the number of inputs to each sorting network is proportional to the
number of inputs to the constraint.
We have carried these techniques over to ASP in our previous work on
the normalization of cardinality rules
\cite{BJ13:lpnmr} and weight rules \cite{BGJ14:jelia}.
For weight rules, a set of structurally shared sorters
is used to compensate for a bit-width induced blowup factor
originating from sorting network based SAT encodings \cite{ES06:jsat}.
This line of research led to optimization rewriting
\cite{BGJ16:iclp}, which is also the focus of
Section~\ref{section:application}. In contrast to normalization, where special
constraints are compiled away altogether, the goal of rewriting
is more relaxed and concerns the reformulation of optimization
statements received as input with the help of additional atoms and
rules.
Section~\ref{section:application} furnishes the weight propagation concept
introduced in this paper to provide a novel optimization rewriting technique.
The technique is distinguished by the fact
that it always requires only a single comparator network regardless of the bit
width of input weights.
Moreover, the network needs to be only a comparator network,
as opposed to a sorting network.
These properties make it easier to find a network
that meets any desired depth and size parameters,
which can be chosen to maximize performance on a given class of problems.

In the context of optimization, sorting networks have also been used to
express sets of mutually similar cardinality constraints generated during
solving by the (unsatisfiable) \emph{core-guided} Maximum Satisfiability
(MaxSAT) solver \system{MSCG} \cite{modoma14}.
MaxSAT is closely related to ASP optimization, witnessed by the fact that
the mentioned solver builds on the algorithm OLL originally devised
for ASP optimization \cite{ankamasc12a} in the ASP solver \system{clasp}.
Core-guided optimization methods such as these start from the unfeasible
region of the search space, solving progressively relaxed, unsatisfiable
subproblems until a solution is found.
In this solving process, reasons for unsatisfiability are characterized by
so-called unsatisfiable
cores that form the basis of analysis in the OLL algorithm, which yields
cardinality constraints used to guide the relaxation steps.
In the more traditional, \emph{model-guided} optimization strategy, such as
branch-and-bound, search begins from the
feasible region of the search space, as models of successively improved value
are sought until no improvement is possible
\cite{wasp_optimization}.
In experimental evaluations presented in
Section~\ref{section:experimental}, optimization rewriting and model-guided
optimization are compared against core-guided optimization,
and the techniques are found to excel at different benchmarks.
However, in contrast to core-guided methods,
optimization rewriting
works in principle with any solving
approach
by virtue of being a preprocessing technique.
This includes model-guided optimization, to which it contributes
the benefits of comparator network encodings.

Another recent development \cite{implicit_hitting_set_asp} in ASP optimization
is to isolate all arithmetic reasoning into a separate \emph{implicit hitting
set (IHS)} problem \cite{implicit_hitting_set_journal}.
The approach stems from MaxSAT~\cite{implicit_hitting_set_maxsat}, and can be
seen as a variation of core-guided optimization where the relaxation steps are
performed differently in order to cope better with non-unit weights.
Namely, each core is encoded together with the optimization criteria in an IHS
problem that can be solved via integer linear programming (ILP).
The solution to the IHS problem is used in relaxing the problem in a way that
requires no added rules.
The use of an ILP solver for this task results in a hybrid approach for ASP
optimization that alternates between ASP decision solving for core extraction
and IHS solving via ILP for problem relaxation \cite{implicit_hitting_set_asp}.
The specific avoidance of new rules and atoms, and the outsourcing of
arithmetic computations to an external solver are in stark contrast with this
work and indicate that the approach seeks performance benefits via a highly
orthogonal manner.
Indeed, the approach does not empower the ASP computation with the potential
benefits of new auxiliary atoms discussed in Section~\ref{section:application}.
Moreover, a combination of the approaches is possible and may prove fruitful,
but such a study is left out of the scope of this work.

There is existing work on formally analyzing abstract solvers for SAT and ASP.
The analysis in Section~\ref{section:formal} is distinguished from existing work
that we are aware of in that it shows an exponential improvement in propagator
based solving of optimization problems in answer set programming.
Regarding the related work, Anger et al.~\citeyear{AGJS06:ecai} demonstrated that
simple program transformations that add structure to an answer set program can
exponentially reduce the search space explored by a state-of-the-art answer set
solver of the time.
The added structure was in the form of new intermediate atoms and rules used
to explicitly express rule bodies.
A formal proof system was provided by Gebser and
Schaub~\citeyear{asp_tableau_calculi_journal} and used to prove exponentially
different best-case computation lengths between different ASP algorithms.
The proof system was a form of tableaux calculi.
It was extended by J{\"{a}}rvisalo and Oikarinen
\citeyear{DBLP:journals/tplp/JarvisaloO08} to form Extended ASP Tableaux, which
further defines an extension rule based on the addition of redundant structure
in the form of added rules.
This addition of redundant structure was proven to enable polynomial length
proofs on a family of normal logic programs on which proofs in the original
proof system were of exponential length at minimum.
In the field of SAT and SMT, an abstract framework was put forth by J{\"{a}}rvisalo
and Oikarinen \citeyear{DBLP:journals/jacm/NieuwenhuisOT06} for describing
standard search procedures for SAT.
In the framework, graphs are used to capture the behavior of solving
algorithms.
In the graphs, nodes represent solver states and directed edges represent
various actions that the algorithms may perform to move between states.
The graphs facilitate formally precise description of algorithms
as well as analysis of their properties.
Based on this framework for SAT and SMT, a similar framework was developed for
describing, analyzing and comparing various ASP solving algorithms
\cite{DBLP:journals/tplp/Lierler11}.
More recently,
a framework was developed for integrating multiple reasoning formalisms, such as
SAT and ASP, together as equal standing components
\cite{DBLP:journals/ai/LierlerT16}.
The framework is able to incorporate the semantics of both propositional
theories and logic programs, and moreover SMT can be translated into it.

\section{Conclusion}
\label{section:conclusion}

In this paper, we present a novel technique for rewriting pseudo-Boolean
expressions deployed as objective functions in ASP and other constraint-based
paradigms.
The technique is based on the novel idea of connecting an encoding of a
comparator network to the literals of an objective function and redistributing
the coefficients of the objective function systematically over the structure of
the network as weights.
When translated into a target formalism, such as ASP, auxiliary atoms used to
express the structure of the network offer the underlying back-end solver
additional branching points and concepts to learn about.
In this paper, we formally analyzed and experimentally evaluated the idea in the
context of ASP.
As part of this, we provide a formal analysis that highlights an exponential
separation in solving performance on an example family of answer set programs.
We implemented the approach in a tool called \system{pbtranslate},
which we evaluated in computational experiments.
In the experiments, we obtained positive experimental results using
\system{pbtranslate} for rewriting optimization statements and \system{clasp}
(v.~3.3.3) as the back-end ASP solver.
We found several benchmark problems where the search for an optimal answer set
is significantly accelerated using designs based on sorting networks.
This holds in comparison to both branch-and-bound and core-guided
optimization strategies of \system{clasp}.
Although rewritten optimization statements have an increased size, the
introduction of useful auxiliary variables and the redistribution of weights
more than compensates for this cost on these benchmark problems.
The idea is moreover completely general and we anticipate further
applications of this technique in neighboring paradigms in addition to ASP.
As regards future work, we believe it is worthwhile to consider
similar techniques using more general types of networks, such as
\emph{permutation networks} \cite{Waksman68:jacm}.

\section*{Acknowledgements}

We would like to thank the anonymous reviewers and Dr.~Martin Gebser
for valuable comments and suggestions.  This work has been supported
in part by the Finnish centre of excellence in
\emph{Computational Inference Research} (COIN)
(Academy of Finland, project \#251170).
Moreover, Jori Bomanson has been supported by Helsinki Doctoral
Network in Information and Communication Technology (HICT) and Tomi
Janhunen partially by the Academy of Finland project
\emph{Ethical AI for the Governance of Society}
(ETAIROS, grant \#327352).

\bibliographystyle{acmtrans}


\begin{thebibliography}{}

\bibitem[\protect\citeauthoryear{Ab{\'{\i}}o, Nieuwenhuis, Oliveras,
  Rodr{\'{\i}}guez{-}Carbonell, and Mayer{-}Eichberger}{Ab{\'{\i}}o
  et~al\mbox{.}}{2012}]{a_new_look_at_bdds}
{\sc Ab{\'{\i}}o, I.}, {\sc Nieuwenhuis, R.}, {\sc Oliveras, A.}, {\sc
  Rodr{\'{\i}}guez{-}Carbonell, E.}, {\sc and} {\sc Mayer{-}Eichberger, V.}
  2012.
\newblock A new look at bdds for pseudo-boolean constraints.
\newblock {\em Journal of Artificial Intelligence Research\/}~{\em 45},
  443--480.

\bibitem[\protect\citeauthoryear{Ab{\'{\i}}o, Nieuwenhuis, Oliveras,
  Rodr{\'{\i}}guez{-}Carbonell, and Stuckey}{Ab{\'{\i}}o
  et~al\mbox{.}}{2013}]{abniolrost13}
{\sc Ab{\'{\i}}o, I.}, {\sc Nieuwenhuis, R.}, {\sc Oliveras, A.}, {\sc
  Rodr{\'{\i}}guez{-}Carbonell, E.}, {\sc and} {\sc Stuckey, P.~J.} 2013.
\newblock To encode or to propagate? {T}he best choice for each constraint in
  {SAT}.
\newblock In {\em Proceedings of CP 2013}, {C.~Schulte}, Ed. LNCS, vol. 8124.
  Springer, 97--106.

\bibitem[\protect\citeauthoryear{Alviano, Dodaro, Leone, and Ricca}{Alviano
  et~al\mbox{.}}{2015}]{DBLP:conf/lpnmr/AlvianoDLR15}
{\sc Alviano, M.}, {\sc Dodaro, C.}, {\sc Leone, N.}, {\sc and} {\sc Ricca, F.}
  2015.
\newblock Advances in {WASP}.
\newblock In {\em Proceedings of LPNMR 2015}, {F.~Calimeri}, {G.~Ianni}, {and}
  {M.~Truszczynski}, Eds. LNCS, vol. 9345. Springer, 40--54.

\bibitem[\protect\citeauthoryear{Alviano, Dodaro, Marques-Silva, and
  Ricca}{Alviano et~al\mbox{.}}{2015}]{wasp_optimization}
{\sc Alviano, M.}, {\sc Dodaro, C.}, {\sc Marques-Silva, J.}, {\sc and} {\sc
  Ricca, F.} 2015.
\newblock Optimum stable model search: algorithms and implementation.
\newblock {\em Journal of Logic and Computation\/}.

\bibitem[\protect\citeauthoryear{Andres, Kaufmann, Matheis, and Schaub}{Andres
  et~al\mbox{.}}{2012}]{ankamasc12a}
{\sc Andres, B.}, {\sc Kaufmann, B.}, {\sc Matheis, O.}, {\sc and} {\sc Schaub,
  T.} 2012.
\newblock Unsatisfiability-based optimization in clasp.
\newblock See \citeN{iclp-lipics12}, 212--221.

\bibitem[\protect\citeauthoryear{Anger, Gebser, Janhunen, and Schaub}{Anger
  et~al\mbox{.}}{2006}]{AGJS06:ecai}
{\sc Anger, C.}, {\sc Gebser, M.}, {\sc Janhunen, T.}, {\sc and} {\sc Schaub,
  T.} 2006.
\newblock What's a head without a body?
\newblock In {\em Proceedings of ECAI 2006}. IOS Press, 769--770.

\bibitem[\protect\citeauthoryear{Apt, Blair, and Walker}{Apt
  et~al\mbox{.}}{1987}]{apblwa87a}
{\sc Apt, K.}, {\sc Blair, H.}, {\sc and} {\sc Walker, A.} 1987.
\newblock Towards a theory of declarative knowledge.
\newblock In {\em Foundations of Deductive Databases and Logic Programming},
  {J.~Minker}, Ed. Morgan Kaufmann Publishers, Chapter~2, 89--148.

\bibitem[\protect\citeauthoryear{Balduccini and Janhunen}{Balduccini and
  Janhunen}{2017}]{lpnmr17}
{\sc Balduccini, M.} {\sc and} {\sc Janhunen, T.}, Eds. 2017.
\newblock {\em Proceedings of LPNMR 2017}. LNCS, vol. 10377. Springer.

\bibitem[\protect\citeauthoryear{Banbara, Soh, Tamura, Inoue, and
  Schaub}{Banbara et~al\mbox{.}}{2013}]{basotainsc13a}
{\sc Banbara, M.}, {\sc Soh, T.}, {\sc Tamura, N.}, {\sc Inoue, K.}, {\sc and}
  {\sc Schaub, T.} 2013.
\newblock Answer set programming as a modeling language for course timetabling.
\newblock {\em Theory and Practice of Logic Programming\/}~{\em 13,\/}~4-5,
  783--798.

\bibitem[\protect\citeauthoryear{Batcher}{Batcher}{1968}]{Batcher68:afips}
{\sc Batcher, K.~E.} 1968.
\newblock Sorting networks and their applications.
\newblock In {\em AFIPS {S}pring joint computer conference}. ACM, Thomson Book
  Company, 307--314.

\bibitem[\protect\citeauthoryear{Bomanson}{Bomanson}{2017}]{bomanson17}
{\sc Bomanson, J.} 2017.
\newblock lp2normal - {A} normalization tool for extended logic programs.
\newblock See \citeN{lpnmr17}, 222--228.

\bibitem[\protect\citeauthoryear{Bomanson, Gebser, and Janhunen}{Bomanson
  et~al\mbox{.}}{2014}]{BGJ14:jelia}
{\sc Bomanson, J.}, {\sc Gebser, M.}, {\sc and} {\sc Janhunen, T.} 2014.
\newblock Improving the normalization of weight rules in answer set programs.
\newblock In {\em Proceedings of JELIA 2014}. LNCS, vol. 8761. Springer,
  166--180.

\bibitem[\protect\citeauthoryear{Bomanson, Gebser, and Janhunen}{Bomanson
  et~al\mbox{.}}{2016}]{BGJ16:iclp}
{\sc Bomanson, J.}, {\sc Gebser, M.}, {\sc and} {\sc Janhunen, T.} 2016.
\newblock Rewriting optimization statements in answer-set programs.
\newblock In {\em Technical Communications of ICLP 2016}. OASIcs, vol.~52.
  Schloss Dagstuhl--Leibniz-Zentrum f\"ur Informatik, 5:1--5:15.
\newblock Article 5.

\bibitem[\protect\citeauthoryear{Bomanson and Janhunen}{Bomanson and
  Janhunen}{2013}]{BJ13:lpnmr}
{\sc Bomanson, J.} {\sc and} {\sc Janhunen, T.} 2013.
\newblock Normalizing cardinality rules using merging and sorting
  constructions.
\newblock In {\em Proceedings of LPNMR 2013}. LNCS, vol. 8148. Springer,
  187--199.

\bibitem[\protect\citeauthoryear{Bonutti, {De Cesco}, {Di Gaspero}, and
  Schaerf}{Bonutti et~al\mbox{.}}{2012}]{bocegasc12a}
{\sc Bonutti, A.}, {\sc {De Cesco}, F.}, {\sc {Di Gaspero}, L.}, {\sc and} {\sc
  Schaerf, A.} 2012.
\newblock Benchmarking curriculum-based course timetabling: Formulations, data
  formats, instances, validation, visualization, and results.
\newblock {\em Annals of Operations Research\/}~{\em 194,\/}~1, 59--70.

\bibitem[\protect\citeauthoryear{Brewka, Eiter, and Truszczy{\'n}ski}{Brewka
  et~al\mbox{.}}{2011}]{BET11:cacm}
{\sc Brewka, G.}, {\sc Eiter, T.}, {\sc and} {\sc Truszczy{\'n}ski, M.} 2011.
\newblock Answer set programming at a glance.
\newblock {\em Communications of the ACM\/}~{\em 54,\/}~12, 92--103.

\bibitem[\protect\citeauthoryear{Calimeri, Faber, Gebser, Ianni, Kaminski,
  Krennwallner, Leone, Ricca, and Schaub}{Calimeri
  et~al\mbox{.}}{2013}]{aspcore201c}
{\sc Calimeri, F.}, {\sc Faber, W.}, {\sc Gebser, M.}, {\sc Ianni, G.}, {\sc
  Kaminski, R.}, {\sc Krennwallner, T.}, {\sc Leone, N.}, {\sc Ricca, F.}, {\sc
  and} {\sc Schaub, T.} 2013.
\newblock {ASP-Core-2}: 4th {ASP} {C}ompetition official input language format.
\newblock Available as
  \url{http://www.mat.unical.it/aspcomp2013/files/ASP-CORE-2.01c.pdf}.

\bibitem[\protect\citeauthoryear{Calimeri, Ianni, and Truszczynski}{Calimeri
  et~al\mbox{.}}{2015}]{lpnmr15}
{\sc Calimeri, F.}, {\sc Ianni, G.}, {\sc and} {\sc Truszczynski, M.}, Eds.
  2015.
\newblock {\em Proceedings of LPNMR 2015}. LNCS, vol. 9345. Springer.

\bibitem[\protect\citeauthoryear{Clark}{Clark}{1978}]{clark78a}
{\sc Clark, K.} 1978.
\newblock Negation as failure.
\newblock In {\em Logic and Data Bases}. Plenum Press, 293--322.

\bibitem[\protect\citeauthoryear{Cussens}{Cussens}{2011}]{Cussens11:uai}
{\sc Cussens, J.} 2011.
\newblock Bayesian network learning with cutting planes.
\newblock In {\em Proceedings of UAI 2011}, {F.~Cozman} {and} {A.~Pfeffer},
  Eds. AUAI Press, 153--160.

\bibitem[\protect\citeauthoryear{Davies and Bacchus}{Davies and
  Bacchus}{2011}]{implicit_hitting_set_maxsat}
{\sc Davies, J.} {\sc and} {\sc Bacchus, F.} 2011.
\newblock Solving {MAXSAT} by solving a sequence of simpler {SAT} instances.
\newblock In {\em Proceedings of CP 2011}, {J.~H. Lee}, Ed. LNCS, vol. 6876.
  Springer, 225--239.

\bibitem[\protect\citeauthoryear{Denecker, Vennekens, Bond, Gebser, and
  Truszczy{\'n}ski}{Denecker et~al\mbox{.}}{2009}]{contest09a}
{\sc Denecker, M.}, {\sc Vennekens, J.}, {\sc Bond, S.}, {\sc Gebser, M.}, {\sc
  and} {\sc Truszczy{\'n}ski, M.} 2009.
\newblock The second answer set programming competition.
\newblock In {\em Proceedings of LPNMR 2009}, {E.~Erdem}, {F.~Lin}, {and}
  {T.~Schaub}, Eds. LNAI, vol. 5753. Springer, 637--654.

\bibitem[\protect\citeauthoryear{Dovier and {Santos Costa}}{Dovier and {Santos
  Costa}}{2012}]{iclp-lipics12}
{\sc Dovier, A.} {\sc and} {\sc {Santos Costa}, V.}, Eds. 2012.
\newblock {\em Technical Communications of ICLP 2012}. Vol.~17. Leibniz
  International Proceedings in Informatics (LIPIcs).

\bibitem[\protect\citeauthoryear{Drescher and Walsh}{Drescher and
  Walsh}{2012}]{drewal12a}
{\sc Drescher, C.} {\sc and} {\sc Walsh, T.} 2012.
\newblock Answer set solving with lazy nogood generation.
\newblock See \citeN{iclp-lipics12}, 188--200.

\bibitem[\protect\citeauthoryear{E{\'e}n and S{\"o}rensson}{E{\'e}n and
  S{\"o}rensson}{2006}]{ES06:jsat}
{\sc E{\'e}n, N.} {\sc and} {\sc S{\"o}rensson, N.} 2006.
\newblock Translating {P}seudo-{B}oolean constraints into {SAT}.
\newblock {\em Journal on Satisfiability, Boolean Modeling and
  Computation\/}~{\em 2,\/}~1--4, 1--26.

\bibitem[\protect\citeauthoryear{Gebser, Kaminski, Kaufmann, Romero, and
  Schaub}{Gebser et~al\mbox{.}}{2015}]{gekakarosc15a}
{\sc Gebser, M.}, {\sc Kaminski, R.}, {\sc Kaufmann, B.}, {\sc Romero, J.},
  {\sc and} {\sc Schaub, T.} 2015.
\newblock Progress in clasp series 3.
\newblock See \citeN{lpnmr15}, 368--383.

\bibitem[\protect\citeauthoryear{Gebser, Kaufmann, and Schaub}{Gebser
  et~al\mbox{.}}{2012}]{GKS12:aij}
{\sc Gebser, M.}, {\sc Kaufmann, B.}, {\sc and} {\sc Schaub, T.} 2012.
\newblock Conflict-driven answer set solving: From theory to practice.
\newblock {\em Artificial Intelligence\/}~{\em 187}, 52--89.

\bibitem[\protect\citeauthoryear{Gebser, Maratea, and Ricca}{Gebser
  et~al\mbox{.}}{2015}]{gemari15a}
{\sc Gebser, M.}, {\sc Maratea, M.}, {\sc and} {\sc Ricca, F.} 2015.
\newblock The design of the sixth answer set programming competition.
\newblock See \citeN{lpnmr15}, 531--544.

\bibitem[\protect\citeauthoryear{Gebser, Maratea, and Ricca}{Gebser
  et~al\mbox{.}}{2017}]{seventh_asp_competition_design}
{\sc Gebser, M.}, {\sc Maratea, M.}, {\sc and} {\sc Ricca, F.} 2017.
\newblock The design of the seventh answer set programming competition.
\newblock See \citeN{lpnmr17}, 3--9.

\bibitem[\protect\citeauthoryear{Gebser and Schaub}{Gebser and
  Schaub}{2013}]{asp_tableau_calculi_journal}
{\sc Gebser, M.} {\sc and} {\sc Schaub, T.} 2013.
\newblock Tableau calculi for logic programs under answer set semantics.
\newblock {\em {ACM} Transactions on Computational Logic\/}~{\em 14,\/}~2,
  15:1--15:40.

\bibitem[\protect\citeauthoryear{Jaakkola, Sontag, Globerson, and
  Meila}{Jaakkola et~al\mbox{.}}{2010}]{JSGM10:jmlr}
{\sc Jaakkola, T.}, {\sc Sontag, D.}, {\sc Globerson, A.}, {\sc and} {\sc
  Meila, M.} 2010.
\newblock Learning {B}ayesian network structure using {LP} relaxations.
\newblock In {\em Proceedings of AISTATS 2010}. JMLR Proceedings, vol.~9. JMLR,
  358--365.

\bibitem[\protect\citeauthoryear{Janhunen, Gebser, Rintanen, Nyman, Pensar, and
  Corander}{Janhunen et~al\mbox{.}}{2017}]{jagerinypeco15a}
{\sc Janhunen, T.}, {\sc Gebser, M.}, {\sc Rintanen, J.}, {\sc Nyman, H.}, {\sc
  Pensar, J.}, {\sc and} {\sc Corander, J.} 2017.
\newblock Learning discrete decomposable graphical models via constraint
  optimization.
\newblock {\em Statistics and Computing\/}~{\em 27,\/}~1, 115--130.

\bibitem[\protect\citeauthoryear{Janhunen and Niemel{\"{a}}}{Janhunen and
  Niemel{\"{a}}}{2016}]{JN16:aimag}
{\sc Janhunen, T.} {\sc and} {\sc Niemel{\"{a}}, I.} 2016.
\newblock The answer set programming paradigm.
\newblock {\em {AI} Magazine\/}~{\em 37,\/}~3, 13--24.

\bibitem[\protect\citeauthoryear{J{\"{a}}rvisalo and Oikarinen}{J{\"{a}}rvisalo
  and Oikarinen}{2008}]{DBLP:journals/tplp/JarvisaloO08}
{\sc J{\"{a}}rvisalo, M.} {\sc and} {\sc Oikarinen, E.} 2008.
\newblock Extended {ASP} tableaux and rule redundancy in normal logic programs.
\newblock {\em {Theory and Practice of Logic Programming}\/}~{\em 8,\/}~5-6,
  691--716.

\bibitem[\protect\citeauthoryear{Lierler}{Lierler}{2011}]{DBLP:journals/tplp/Lierler11}
{\sc Lierler, Y.} 2011.
\newblock Abstract answer set solvers with backjumping and learning.
\newblock {\em {Theory and Practice of Logic Programming}\/}~{\em 11,\/}~2-3,
  135--169.

\bibitem[\protect\citeauthoryear{Lierler and Truszczynski}{Lierler and
  Truszczynski}{2016}]{DBLP:journals/ai/LierlerT16}
{\sc Lierler, Y.} {\sc and} {\sc Truszczynski, M.} 2016.
\newblock On abstract modular inference systems and solvers.
\newblock {\em Artificial Intelligence\/}~{\em 236}, 65--89.

\bibitem[\protect\citeauthoryear{Lifschitz and Razborov}{Lifschitz and
  Razborov}{2006}]{LR06:acmtocl}
{\sc Lifschitz, V.} {\sc and} {\sc Razborov, A.~A.} 2006.
\newblock Why are there so many loop formulas?
\newblock {\em {ACM} Transactions on Computational Logic\/}~{\em 7,\/}~2,
  261--268.

\bibitem[\protect\citeauthoryear{Lifschitz and Turner}{Lifschitz and
  Turner}{1994}]{liftur94a}
{\sc Lifschitz, V.} {\sc and} {\sc Turner, H.} 1994.
\newblock Splitting a logic program.
\newblock In {\em Proceedings of ICLP 1994}. MIT Press, 23--37.

\bibitem[\protect\citeauthoryear{Maratea, Pulina, and Ricca}{Maratea
  et~al\mbox{.}}{2015}]{me_asp_multi_level}
{\sc Maratea, M.}, {\sc Pulina, L.}, {\sc and} {\sc Ricca, F.} 2015.
\newblock Multi-level algorithm selection for {ASP}.
\newblock See \citeN{lpnmr15}, 439--445.

\bibitem[\protect\citeauthoryear{{MaxSAT-Comp}}{{MaxSAT-Comp}}{2014}]{maxsat_comp}
{\sc {MaxSAT-Comp}}. 2014.
\newblock Ninth {Max-SAT} evaluation.
\newblock Available as \url{http://www.maxsat.udl.cat/14/.}

\bibitem[\protect\citeauthoryear{Moreno{-}Centeno and Karp}{Moreno{-}Centeno
  and Karp}{2013}]{implicit_hitting_set_journal}
{\sc Moreno{-}Centeno, E.} {\sc and} {\sc Karp, R.~M.} 2013.
\newblock The implicit hitting set approach to solve combinatorial optimization
  problems with an application to multigenome alignment.
\newblock {\em Operations Research\/}~{\em 61,\/}~2, 453--468.

\bibitem[\protect\citeauthoryear{Morgado, Dodaro, and Marques{-}Silva}{Morgado
  et~al\mbox{.}}{2014}]{modoma14}
{\sc Morgado, A.}, {\sc Dodaro, C.}, {\sc and} {\sc Marques{-}Silva, J.} 2014.
\newblock Core-guided maxsat with soft cardinality constraints.
\newblock In {\em Proceedings of CP 2014}, {B.~O'Sullivan}, Ed. LNCS, vol.
  8656. Springer, 564--573.

\bibitem[\protect\citeauthoryear{Nieuwenhuis, Oliveras, and
  Tinelli}{Nieuwenhuis
  et~al\mbox{.}}{2006}]{DBLP:journals/jacm/NieuwenhuisOT06}
{\sc Nieuwenhuis, R.}, {\sc Oliveras, A.}, {\sc and} {\sc Tinelli, C.} 2006.
\newblock Solving {SAT} and {SAT} modulo theories: From an abstract
  {D}avis--{P}utnam--{L}ogemann--{L}oveland procedure to {DPLL}(\emph{T}).
\newblock {\em Journal of the ACM\/}~{\em 53,\/}~6, 937--977.

\bibitem[\protect\citeauthoryear{Saikko, Dodaro, Alviano, and
  J{\"{a}}rvisalo}{Saikko et~al\mbox{.}}{2018}]{implicit_hitting_set_asp}
{\sc Saikko, P.}, {\sc Dodaro, C.}, {\sc Alviano, M.}, {\sc and} {\sc
  J{\"{a}}rvisalo, M.} 2018.
\newblock A hybrid approach to optimization in answer set programming.
\newblock In {\em Proceedings of KR 2018}, {M.~Thielscher}, {F.~Toni}, {and}
  {F.~Wolter}, Eds. {AAAI} Press, 32--41.

\bibitem[\protect\citeauthoryear{Waksman}{Waksman}{1968}]{Waksman68:jacm}
{\sc Waksman, A.} 1968.
\newblock A permutation network.
\newblock {\em Journal of the ACM\/}~{\em 15,\/}~1, 159--163.

\bibitem[\protect\citeauthoryear{Zhou and Kjellerstrand}{Zhou and
  Kjellerstrand}{2016}]{picat_2016}
{\sc Zhou, N.} {\sc and} {\sc Kjellerstrand, H.} 2016.
\newblock The picat-sat compiler.
\newblock In {\em Proceedings of PADL 2016}, {M.~Gavanelli} {and} {J.~H.
  Reppy}, Eds. LNCS, vol. 9585. Springer, 48--62.

\end{thebibliography}

\end{document}